\documentclass{article}
\usepackage{graphicx}
\usepackage{amsfonts}
\usepackage{amsmath}
\usepackage{amssymb}
\usepackage{url}
\usepackage{fancyhdr}
\usepackage{indentfirst}
\usepackage{subfigure}
\usepackage{enumerate}

\usepackage{dsfont}
\usepackage[colorlinks=true,citecolor=blue,hyperindex,breaklinks]{hyperref}
\usepackage{amsthm}
\usepackage{color}
\usepackage{natbib}
\usepackage{comment}
\usepackage{diagbox}
\usepackage{capt-of}
\usepackage{booktabs}
\usepackage{multirow}

\usepackage{mathtools}

\def\d{\mathrm{d}}

\newcommand{\ES}{\mathrm{ES}}
\newcommand{\LES}{\mathrm{LES}}

\newcommand{\E}{\mathbb{E}}

\newcommand{\R}{\mathbb{R}}
\newcommand{\N}{\mathbb{N}}
\newcommand{\p}{\mathbb{P}}

\newcommand{\M}{\mathcal{M}}

\newcommand{\id}{\mathds{1}}

\renewcommand{\ge}{\geqslant}
\renewcommand{\le}{\leqslant}
\renewcommand{\geq}{\geqslant}
\renewcommand{\leq}{\leqslant}
\renewcommand{\epsilon}{\varepsilon}

\renewcommand{\cdots}{\dots}

\theoremstyle{plain}
\newtheorem{theorem}{Theorem}

\newtheorem{lemma}{Lemma}
\newtheorem{proposition}{Proposition}
\theoremstyle{definition}
\newtheorem{definition}{Definition}
\newtheorem{example}{Example}

\theoremstyle{remark}
\newtheorem{remark}{Remark}

\newcommand{\cet}{\begin{center}}
	\newcommand{\ecet}{\end{center}}

\usepackage{setspace}

\begin{document}

\title{Convolution Bounds on Quantile Aggregation}

\author{
	Jose Blanchet\thanks{
		Department of Management Science and Engineering,
		Stanford University, USA.
		Email: \texttt{jose.blanchet@stanford.edu}}
	\and
	Henry Lam\thanks{  Department of Industrial Engineering and Operations Research, Columbia University, USA.
		Email: \texttt{khl2114@columbia.edu}}
	\and
	Yang Liu\thanks{School of Science and Engineering, The Chinese University of Hong Kong, Shenzhen, China. Email: \texttt{yangliu16@cuhk.edu.cn}}
    %Department of Management Science and Engineering, Stanford University, USA. Email: \texttt{yangliu3@stanford.edu}}  %Department of Mathematical Sciences, Tsinghua University, China. Email: \texttt{yang-liu16@mails.tsinghua.edu.cn}}
	\and Ruodu Wang\thanks{  Department of Statistics and Actuarial Science, University of Waterloo, Canada. Email: \texttt{wang@uwaterloo.ca}}}
\date{}
\maketitle

\begin{abstract}
Quantile aggregation with dependence uncertainty has a long history in probability theory with wide applications in finance, risk management, statistics, and operations research. Using a recent result on inf-convolution of quantile-based risk measures, we establish new analytical bounds for quantile aggregation which we call convolution bounds. Convolution bounds both unify every analytical result available in quantile aggregation and enlighten our understanding of these methods. These bounds are the best available in general. Moreover, convolution bounds are easy to compute, and we show that they are sharp in many relevant cases. They also allow for interpretability on the extremal dependence structure. The results directly lead to bounds on the distribution of the sum of random variables with arbitrary dependence. We discuss relevant applications in risk management and economics.

\end{abstract}
%\newpage

%\tableofcontents
\section{Introduction}\label{sec:0}

The problem of quantile aggregation with dependence uncertainty refers to finding possible values of quantiles of an aggregate variable $S=X_1+\dots+X_n$ (often representing a total risk, but it can also represent the completion time of a task). The random variables $X_1,\dots,X_n$ have given marginal distributions, but unspecified dependence structure. More precisely, given marginal distributions $\mu_1,\dots,\mu_n$ on $\R$, the following quantities are of interest:
\begin{equation}\label{eq:intrinsic_prob}
	\sup\{ q_t(X_1+\dots+X_n): X_i\sim \mu_i,~i=1,\dots,n\}
\end{equation}
and
\begin{equation}\label{eq:intrinsic_prob2} \inf\{ q_t(X_1+\dots+X_n): X_i\sim \mu_i,~i=1,\dots,n\},\end{equation}
where $q_t(X)$ stands for a (left or right) quantile of a random variable $X$ at probability level $t\in [0,1]$. The optimization problems  \eqref{eq:intrinsic_prob} and \eqref{eq:intrinsic_prob2} are, respectively, referred to as the worst-case and the best-case quantile aggregation.
An equivalent problem is to find the maximum and the minimum values of $\p(S\le x)$ for a given $x\in \R$.
This problem has a long history in probability theory; see \cite{M81} and \cite{R82} for early results. It has also been studied in combinatorial optimization  with applications in  statistical testing (e.g., \cite{VW20, VWW22}) and risk management  (e.g., \cite{EPR13, EWW15}); see Section \ref{sec:ex} and Appendix \ref{app:applications} for several applications. A key feature of quantile aggregation is the ``arbitrary dependence" structure imposed. Naturally, this level of generality leads to robust estimates, although these can be conservative in some situations.

%Given its fundamental importance,  quantile aggregation has  wide applications in many areas,
%A key feature of quantile aggregation is ``arbitrary dependence", thus having no dependence assumption at all.
%Naturally, such a setting may lead to very conservative values in some situations; nevertheless, it has wide appearance in many areas, including economic theory, operations research, statistics, risk management, engineering, and finance. We explain a few examples in Section \ref{sec:ex} which are by no means exhaustive.

Because of the level of generality imposed both in marginal distributions and dependence, the quantile aggregation problems \eqref{eq:intrinsic_prob} and \eqref{eq:intrinsic_prob2} rarely have analytical tractability. In the literature, some analytical bounds for the homogeneous setting (i.e., identical marginal distributions) are obtained by \cite{EP06}, \cite{WPY13} and \cite{PR13}, and approximating algorithms are available such as the rearrangement algorithm (RA) in \cite{PR12} and \cite{EPR13}. The sharpness of these bounds is rarely obtained with the exception of \cite{WPY13} and \cite{PR13} under some strong conditions. The RA only gives a lower bound on the quantile aggregation, and its convergence is not guaranteed. 
%\textcolor{blue}{
As a variant of optimal transport problem, discrete versions of problems \eqref{eq:intrinsic_prob} and \eqref{eq:intrinsic_prob2} admit  a linear programming reformulation, which involves exponentially many variables, and is computationally difficult for moderate dimensions (e.g., $n\ge 6$); details are explained in Appendix \ref{sec:R2-1}. 
%}%The polynomial number of variables in the dual problem \eqref{eq:R2-3-dual} provides a potential for yielding poly$(m,n)$-time methods for the primal problem \eqref{eq:R2-3}. 
% In fact, the assembly line crew scheduling problem (i.e. the discrete version of quantile aggregation) is NP-hard; see \cite{CY84} and \cite{H84}. 
So, it is sensible to discuss bounds that can be shown to be sharp in continuous relaxations as the ones that we consider here.

In this paper, we propose a class of bounds on Range-Value-at-Risk (RVaR) based on the inf-convolution formulas introduced by \cite{ELW18}. We will call them \emph{convolution bounds}. Since RVaR includes the two regulatory risk measures, Value-at-Risk (VaR) and the Expected Shortfall (ES, also known as CVaR), as special cases, the  results on RVaR give rise to useful bounds on quantile aggregation problems \eqref{eq:intrinsic_prob} and \eqref{eq:intrinsic_prob2}.

As our main contributions, convolution bounds can provide by far the most convenient and sharpest theoretical results on quantile aggregation in a wide range of practical settings, and they can be applied to any marginal distributions, discrete, continuous, or mixed. As such, convolution bounds enjoy multifaceted advantages. They can be applied to both the quantile and RVaR aggregations (Theorems \ref{th:qa-4prime} and \ref{th:qa-4primeprime}); they combine different existing sharpness results of quantile aggregation and some new cases into a unified form (Theorems \ref{th:qa-4} and \ref{th:qa-2});  they lead to tractable extremal dependence structures for interpretation or approximation (Theorem \ref{th:qa-5}), and they are computationally convenient and efficient. %; (e) connects to the dual bound (Theorem 4.17 of \cite{R13}) and additionally proves its sharpness in relevant cases for the first time.
%The convolution bounds in Theorems \ref{th:qa-4} and \ref{th:qa-2} provide by far the most convenient theoretical results on  quantile aggregation, and they can be applied to any marginal distributions, discrete, continuous, or mixed.
To the best of our knowledge, there is no other theoretical result on quantile aggregation which cannot be covered by our convolution bounds. Moreover, our results provide novel bounds on RVaR  aggregation and establish sharpness for the dual bound (Theorem 4.17 of \cite{R13}). Although the sharpness of convolution bounds requires some conditions, their numerical performance suggests that they are generally very accurate even in cases where sharpness cannot be theoretically proved.  As we mentioned above, our results on quantile aggregation can be directly applied to compute bounds on the distribution of the sum of random variables with arbitrary dependence. 
{Our technical development builds on some results on risk sharing in \cite{ELW18}. Our target problem and theoretical contributions are very different from those on risk sharing, which aim  to optimally allocate a fixed total risk (random variable $X$) to different agents (several random variables that sum to $X$). Our objective, on the contrary, aims to solve the max/min values of the quantile of the sum random variable given known marginals (the total risk is not fixed). This problem is called the problem of robust risk aggregation in the literature; see Section \ref{sec:qrm}.} %{One of our main observations is that one can borrow some techniques from the risk sharing problems studied in   [1], in the form of Lemma EC.1, to study the robust risk aggregation problem; thus, some duality exists between the two problems of optimal risk sharing and robust risk aggregation. This connection was never discovered in the literature before.}

%{\color{red} Throughout, we will prove our main contributions by using a lot of existing results in the literature. Particularly, in the proof of convolution bounds in Theorems \ref{th:qa-4prime}-\ref{th:qa-4}, we adopt the convolution inequality from \cite{ELW18}. In the proof of the case of monotone densities, we adopt the results from \cite{JHW16}. In the technical details, Proposition \ref{prop:basic} uses the result of \cite{LW16}. However, the current paper contributes to establishing the theoretical analysis and application of convolution bounds, which is original compared to the literature.  
%}

We can relate and contrast our investigation to the recent fast-growing literature on distributionally robust optimization (DRO) (e.g.,~\cite{GS10,DY10,WKS14})  and chance constrained optimization (e.g.,~\cite{NS07}, \cite{CSST10} and Chapter 4 of \cite{SDR21}). Viewed as a distributional analog to (deterministic) robust optimization (\cite{BBC11,BEN09}), this literature tackles decision-making where the underlying parameter in a stochastic problem is uncertain.
This leads to the optimization of decision under the worst-case scenario, where the worst case is over a region in which the uncertain parameter is believed to lie in,  often known as the uncertainty set or ambiguity set. %This typically results in a minimax problem where the inner maximization is decided on the uncertain parameter.
In DRO, the uncertain parameter, and hence the decision variable in the inner maximization, is the underlying probability distribution. Common constraints to characterize the belief on uncertain distributions include neighborhood balls formed by statistical distances such as the Wasserstein distance (\cite{EK18,GK16,BM19}) and $\phi$-divergence (\cite{BGK18,BDDMR13}), moments and supports (\cite{BP05}), geometric shape (\cite{P05}), and marginal information (\cite{DN12,MNPTL14}).\footnote{There is a large literature on DRO problems with various formulations, in addition to the few papers mentioned. We refer to \cite{BGYZ19,BKM19} and the references therein for recent developments on DRO with Wasserstein distance,
and to \cite{HH13, GX14,JG16,L16, BHM20} for DRO problems with $\phi$-divergence. DRO formulated by moments and supports are also studied by, e.g., \cite{DY10,GS10,GL19}. See also \cite{VGK16,LM17,LJM19} for various settings of DRO with geometric shape.} The RVaR and quantile aggregation considered in this paper can be regarded as an optimization over distributions having a marginal information constraint (i.e., the latest class listed above).  When placed as a constraint, quantile or percentile criterion can be converted into a chance or probabilistic constraint (e.g., \cite{DM10}). The worst-case VaR under various settings of model uncertainty  is also popular in robust portfolio optimization; see e.g., \cite{EOO03}, \cite{ZF09} and \cite{ZKR13}.
However, in contrast to the DRO literature which often focuses on solution methods via convex reformulations, here our problem is knowingly computationally intractable, and our goal is to obtain tractable analytical bounds that are provably tight in important cases.
% optimization,  and  placed the uncertain parameter where the worst-case In robust optimization

The rest of the paper is organized as follows.  Section \ref{sec:ex} presents 
two motivating examples of robust risk management 
and  the O-ring model in economics (\cite{K93}),  and 
Section \ref{sec:1} contains   technical preliminaries.  
 The (upper) convolution bounds on the quantile and RVaR aggregations are established in Sections \ref{sec:rvar}-\ref{sec:2}.  A general extremal dependence structure and some explicit approximations are presented in Section \ref{sec:new6}.  The dual formulation of the quantile aggregation problems is studied in Section \ref{sec:7} (Theorem \ref{th:qa-6}). The numerical advantages of the new bounds are carefully examined in Section \ref{sec:8}. 
 The two  motivating examples are revisited in  Section \ref{sec:OR}, where we apply our main results and discuss their implications.  Section \ref{sec:9} concludes the paper. To better illustrate our main ideas, the lower convolution bounds  and related discussions are postponed to  Appendix \ref{sec:lower} (in particular, Theorems \ref{th:qa-4primeprime} and \ref{th:qa-2}). Appendices \ref{app:proof}-\ref{app:applications} include all proofs, the counter-examples, technical discussions and other operations research applications.

\section{Motivating examples}\label{sec:ex}

In this section, we list two examples where the quantile aggregation problems  \eqref{eq:intrinsic_prob} and \eqref{eq:intrinsic_prob2} become natural in various  contexts relevant to modern operations research. We will revisit these examples with our theoretical results and numerical illustrations in Section \ref{sec:OR}. 

\subsection{Robust risk management}
\label{sec:qrm}  
	The worst-case value  of a risk measure $\rho$ evaluating an aggregate risk is extensively studied in the  risk management literature, known as the problem of robust risk aggregation. 
	 It is motivated by the context in which data from different correlated products are separately collected and hence their dependence information is not available; see \cite{EPR13, EWW15} and the references therein.  As  a  specific example,\footnote{We thank an anonymous referee for providing the context in this example.} the European Union has established a solidarity fund since 2002 to help member states in case of some catastrophic events. While the loss curves are well estimated in each country, the fund has to pay for the sum of all losses. An independence assumption cannot be justified for climate-related events, which may affect several countries. Since an estimate of the copula is not available, the worst-case analysis in this section provides some bounds on the total losses. 
	 
% Second, in the recent Silicon Valley Bank failure, we know that very rapid losses in the bank balance’s sheet occurred because of the synchronization of multiple risk factors which, in isolation, maybe would not have caused the failure of the bank. One risk factor is the interest rate increase by the FED (which affects all banks). Another factors are the customer base of SVB, which is largely formed by startup companies in the tech industry. These are customers with substantial uninsured deposits which, due to the contraction in venture capital funding, have been withdrawing funds from the bank. As we now know (and with the benefit of hindsight) these risk factors (both on the assets and liability sides) correlated strongly eroding the value of the bank. This provides the latest example stressing the importance of a disciplined approach to robust risk aggregation.  
	 
	 %{\color{red}Now we establish our mathematical model formally.} 
	 Suppose   that  there are $n$ random losses $X_1, \dots, X_n$ with known marginal distributions $\mu_1,\dots,\mu_n$  and unknown dependence structure. To calculate the regulatory margin conservatively, one relies on the worst-case aggregate risk, that is, $$
	\sup\{ \rho(X_1+\dots+X_n): X_i\sim \mu_i,~i=1,\dots,n\}.
	$$ 
	If $\rho$ is a convex risk measure such as an ES, then its worst-case value is easy to compute due to convexity; see \cite{R13}. 
	In case $\rho$ is the VaR at level $t$, this quantity is   \eqref{eq:intrinsic_prob}, which is highly non-trivial because of non-convexity of the quantile.
 Due to the connection of quantiles to risk measures like VaR, quantile aggregation is a popular problem in risk management (Section 8.4 of \cite{MFE15}), and many useful technical results were developed in this literature, e.g., \cite{EP06}, \cite{WPY13}, \cite{EWW15} and \cite{JHW16}.
	Our main results directly address this problem for the cases that $\rho$ is a VaR or RVaR; some numerical illustrations are presented in  Section \ref{sec:num_var}.
	
Next, we bring the worst-case risk calculation to the context of portfolio selection.
 The traditional problem of VaR-based portfolio selection (e.g., \cite{BS01}) is formulated as 
	 $$ 
	 	\mbox{maximize} ~ \E[u(\boldsymbol{\lambda} \cdot \mathbf{X})]~ \mbox{over $\boldsymbol{\lambda} \in \overline{\Delta}_{n-1}$}, ~~~\text{subject to~}   {q_t}(\boldsymbol{\lambda} \cdot (-\mathbf{X})) \leq x, 
	 $$
	 where $\boldsymbol{\lambda}$ represents a  portfolio weight vector,   $\mathbf X$ represents future asset values,   $x$ is a constant risk limit,   $u:\R \to \R$ is a strictly concave and increasing utility function,
	 $t\in(0,1)$ is close to $1$,
	  and $\overline{\Delta}_{n-1} $ is the   standard $n$-simplex, that is, $$ \overline{\Delta}_{n-1}= \left\{( \lambda_1, \dots, \lambda_n) \in [0,1]^{n}: \sum_{i=1}^n \lambda_i = 1 \right\}.$$
	  Note that the quantile constraint can be equivalently formulated as a chance (exceedance probability) constraint, popular in the literature of stochastic programming.
	  This formulation requires a full specification on the joint distribution of all the assets $(X_1, \cdots, X_n)$ which can be difficult to obtain.  In the presence of dependence uncertainty, we consider the following robust optimization problem, for a given tuple  $\boldsymbol \mu$ of marginal distributions, 
	\begin{align}
	\label{eq:r1-robust} 
	\mbox{maximize} ~&  \inf_{\mathbf{X} \sim \boldsymbol{\mu}} \E[u(\boldsymbol{\lambda} \cdot \mathbf{X})]~\mbox{over $\boldsymbol{\lambda} \in \overline{\Delta}_{n-1}$}, ~~~
		\text{subject to~}   \sup_{\mathbf{X} \sim \boldsymbol{\mu}} {q_t}(\boldsymbol{\lambda} \cdot (-\mathbf{X})) \leq x, 	 
	\end{align}
	%\com{YL: remove $q_t^+$}
	where $\mathbf{X} \sim \boldsymbol{\mu}$ represents the marginal conditions $ X_i\sim \mu_i,~i=1,\dots,n$. The problem \eqref{eq:r1-robust} has robustness implications. Assume that the marginal distribution is well-specified and we solve \eqref{eq:r1-robust}. Denote the optimal solution by $\boldsymbol{\lambda}^*$ and the optimal value by $v^*$. They satisfy the guarantee that
	$q_t(\boldsymbol{\lambda}^* \cdot (-\mathbf{X}_0)) \leq x$
	and
	$\E[u(\boldsymbol{\lambda}^* \cdot \mathbf{X}_0)] \geq v^*$, 
	where $\mathbf{X}_0$ follows the unknown true distribution. In other words, we guarantee that the quantile constraint under the true distribution is satisfied, while the attained objective value under the true distribution has at least a performance level $v^*$.  
		
	The problem \eqref{eq:r1-robust}  is challenging due to the non-convexity of VaR. Portfolio optimization with dependence uncertainty has been studied by, e.g., \cite{PP18}, but there are no results on the case of VaR.  Our results on quantile aggregation can  be applied to address this problem.
As our analysis in Section \ref{sec:OR} shows, although this worst-case approach is generally   conservative,   the obtained optimal strategies are quite intuitive.

\subsection{The O-ring model}
\label{sec:o-ring}  
	The O-ring theory of economic development was proposed by \cite{K93}; see also  the recent work of \cite{BTZ21} and the references therein. The O-ring model can be formulated in a stochastic context. Assume that there are continuums of firms and $n$ types of workers. Each firm requires $n$  workers, one in each type, to format   a team in production. Let $\omega$ be a firm in the continuum. The product value of the firm $\omega$ is denoted by $Z(\omega) \in (0, \infty)$. For a type-$i$ worker matched with the firm $\omega$, the probability to successfully complete his/her task is denoted by $X_i (\omega) \in (0, 1)$. A high-skilled worker has a higher value of $X_i$.  Among all firms, the value $Z$ has a distribution $\mu_Z$. Among all workers of type $i$, the value $X_i$ has a distribution $\mu_{i}$. The product of a firm is considered successful if all $n$ workers in the firm complete their individual tasks (this explains the name of the O-ring model). It is customary as in \cite{K93} to assume that $n$ individual events, in which the $i$-th worker completes his/her task, $i = 1, \cdots, n$, are independent  for a fixed firm. Hence, the production function of the firm $\omega$ is the product value times the probability of success, that is,
	\begin{equation}\label{eq:o-ring}
		y(X_1(\omega), \cdots, X_n(\omega), Z(\omega)) = Z(\omega) \cdot \prod_{i = 1}^n X_i(\omega).
	\end{equation}
	A classic problem is to seek a global matching between multiple heterogeneous workers into teams at heterogeneous firms in order to maximize $\E[y(X_1, \cdots, X_n, Z)]$ among all kinds of dependence structures with the given marginal distributions. The solution of the optimal sorting is positively dependent; more precisely, $Z, X_1, \cdots, X_n$ are comonotonic. The interpretation  is that the good workers ($X_i$ all have a higher value) should work together in a good firm ($Z$ also has a higher value). This partially explains the assignment of global economic industries between the developed and developing countries as argued by \cite{K93}.
	
	As argued by \cite{BTZ21, BTWZ23}, labour matching observed in the labour market does not show the comonotonic pattern as implied by the classic O-ring theory. Below, we explain that a quantile aggregation problem    leads to a richer matching pattern which can be solved using the results in this paper.\footnote{It is not our intention to say that the real labour market follows such a model; this issue would require a separate study. Our model provides a way to generate rich matching patterns. 
	This is also the approach taken by \cite{BTZ21, BTWZ23} for different settings.}  
	 There is a recently increasing interest in quantiles  as  decision criteria in economics; see  \cite{R10} and \cite{dG19} for   theoretical advances and \cite{dGNQ22} for experimental analysis. 
	 
	Instead of optimizing the expected production in \eqref{eq:o-ring}  across firms, one may be  concerned about how many productions have low values below a certain threshold $y_0>0$, e.g., a level that is unacceptable by the society. That is, one investigates the  deficiency proportion minimization problem
	\begin{equation}\label{prob:o-ring}
		\min \left\{\p(y(X_1, \cdots, X_n, Z)\le y_0):  Z \sim \mu_Z,~ X_i \sim \mu_i,~ i = 1, \cdots, n\right\},
	\end{equation}
	%This problem reflects the managerial viewpoint to guarantee the safety loading of production, to avoid catastrophes, etc.
	where the probability $\p$ measures  the proportion of productions that falls below the   deficiency threshold.  Since the problems of quantile aggregation and probability bounds translate to each other,  for \eqref{prob:o-ring}  it suffices to solve the problem of quantile aggregation on $\log(Z) + \sum_{i=1}^n \log(X_i)$.  
	The extremal dependence structure attaining \eqref{prob:o-ring} illustrates the optimal matching pattern, which is the topic of Section \ref{sec:new6}. 
	 Section \ref{sec:OR} contains a detailed illustration.
	Our results can also be  applied to  the model of \cite{BTZ21}, where the product is considered successful if at least one worker, instead of all, is able to complete the task.

\section{Notation and preliminaries}\label{sec:1}
%\subsection{Notation}
Let $\M$ be the set of (Borel) probability measures on $\R$ and $\M_1$ be the set of  probability measures on $\R$ with finite mean.
For $\boldsymbol \mu= (\mu_1,\dots,\mu_n)\in \M^n$,
let
$\Gamma (\boldsymbol \mu)$ be the set of probability measures on $\R^n$ that have one-dimensional marginals $\mu_1,\dots,\mu_n$.
For a probability measure $\mu $ on $\R^n$, define   $\lambda_\mu \in \M$ by $$\lambda_\mu (-\infty,x]  = \mu(\{ (x_1,\dots,x_n)\in \R^n: x_1+\dots+x_n\le x\}),~ x\in \R.$$
In other words, $\lambda_\mu$ is the distribution measure of $\sum_{i=1}^n X_i$  where the random vector $(X_1,\dots,X_n)$ follows $\mu$.
Moreover, let
$\Lambda (\boldsymbol \mu) =\{\lambda_\mu: \mu\in \Gamma (\boldsymbol \mu)\}.$
Thus, $\Lambda (\boldsymbol \mu)$ is the set of the aggregate distribution measures with specified marginals $\boldsymbol \mu$.
For $t\in (0,1]$, define the left quantile functional
$$q^-_t(\mu)=\inf\{x\in \R: \mu(-\infty,x]\ge t\}, ~~\mu\in\M,$$
and for $t\in [0,1)$, define
the right quantile functional
$$q^+_t(\mu)=\inf\{x\in \R: \mu(-\infty,x]> t\}, ~~\mu\in\M.$$
%In addition, let $q_0^+(\mu)=q_0^+^+(\mu) $
%and $q_1^-^+(\mu)=q_1^-(\mu) $,
% $\mu\in\M$.
%Since there is only one version of the quantile at $t=0$ and $t=1$,
%we simply write $q_0 = q_0^+$ and $q_1= q_1^-$.
The two extreme cases $q_0^+$ and $q_1^-$ correspond to the essential infimum and the essential supremum.
%, respectively, and they will be the   important objects in our theoretical development, as they directly lead to results on all other quantiles.
Note that $q^\pm_t$   is defined on $\M$ instead of on the set of random variables as in the introduction.
%; we simply use $q^-_t(X)=q^-_t(\mu)$  and $q^+_t(X)=q^+_t(\mu)$ if $X\sim \mu$.
The most important objects in this paper are the average quantile functionals which we define next.
For $ 0 \le \beta < \beta+\alpha \le 1$,  define \begin{equation}\label{eq:r1}
	R_{\beta, \alpha} (\mu) = \frac{1}{\alpha} \int _{\beta} ^{\beta+\alpha} q^+_{1-t}(\mu)\d t, ~~\mu\in\M.\end{equation}
By definition, $ R_{\beta,\alpha} (\mu) $ is the average of the quantile\footnote{We can use either $q^+$ or $q^-$ in the integral, as the two quantities are the same almost everywhere on $[0,1]$.} of $\mu$ over $[1-\beta-\alpha,1-\beta]$. The functional $R_{\beta,\alpha} $, introduced originally by \cite{CDS10}, is called an RVaR  by \cite{WBT15}.
The value $R_{\alpha,\beta}(\mu)$ in \eqref{eq:r1}
is always finite for $\beta>0$ and $\alpha+\beta<1$,
and it may take the value $\infty$ or $-\infty$ in case one of $\beta=0$ or $\alpha+\beta=1$.
For the special case in which $\beta=0$ and $\alpha=1$, $R_{0,1}$ is precisely the mean, and it is only well defined on the set $\M_1$ of distributions with   finite mean.
The left and right quantiles can be obtained as limiting cases of $R_{\beta,\alpha}$ for $\beta \in(0,1)$ via
\begin{equation}
	\lim_{ \alpha\downarrow 0} R_{\beta,\alpha}(\mu)=q_{1-\beta}^-(\mu)
	\mbox{~~~and~~~}
	\lim_{ \alpha\downarrow 0} R_{\beta-\alpha,\alpha}(\mu)=q_{1-\beta}^+(\mu), ~~ \mu \in \M.\label{eq:quantileconvergence}
\end{equation}
Two other useful special cases are  ES and the left-tail ES (LES), defined, respectively, at level $\alpha\in (0,1)$ via
$$
\ES_\alpha(\mu) = R_{0,\alpha}(\mu) = \frac{1}{\alpha}\int_{1-\alpha}^{1} q_u^+ (\mu) \d u, ~~ \mu \in \M,
$$
and
$$
\LES_\alpha(\mu) = R_{1-\alpha,\alpha}(\mu) =  \frac{1}{\alpha}\int_{0}^{\alpha} q_u^+ (\mu) \d u, ~~ \mu \in \M.
$$
As explained by \cite{ELW18}, the RVaR functional $R$ bridges the gap between quantiles (VaR) and ES, the two most popular risk measures in banking and insurance.

It is sometimes convenient to slightly abuse the notation by using $R_{\beta,\alpha}(X)$  or $q_t(X)$  for $ R_{\beta,\alpha}(\mu)  $ or $q_t(\mu)$
where $X\sim \mu$. All random variables  appearing in the paper live in an atomless probability space $(\Omega,\mathcal F,\p)$.\footnote{A probability space is atomless if there exists a continuously distributed random variable on this space.}
We use $\bigvee_{i=1}^n \alpha_i $ for the maximum of real numbers  $\alpha_1,\dots,\alpha_n$.

\section{Convolution bounds on RVaR aggregation}\label{sec:rvar}
%{Preliminaries}

Our starting point is that an upper bound on RVaR aggregation, which we shall refer to as {convolution bounds}, can be obtained from an inequality on RVaR from \cite{ELW18}. 
More precisely, Theorem 2 of \cite{ELW18} gives the following inf-convolution formula, 
for any integrable random variable $X$ and $\alpha_1,\dots,\alpha_n,\beta_1,\dots,\beta_n\in [0,1]$ with $\beta  + \alpha \le 1$   
where $\beta=\sum_{i=1}^n \beta_i$ and $\alpha=\bigvee_{i=1}^n \alpha_i$, 
\begin{equation}\label{eq:e-1}
	R_{ \beta,  \alpha}\left(X\right) =\inf\left\{\sum_{i=1}^n R_{\beta_i, \alpha_i}(X_i) :X_1+\dots+X_n=X\right\},
\end{equation}
where the infimum is taken over all random variables $X_1,\dots,X_n$.
As a consequence of \eqref{eq:e-1}, we have an RVaR aggregation inequality
\begin{equation}\label{eq:e0}
	R_{ \beta,  \alpha}\left(\sum_{i=1}^n X_i\right) \leq \sum_{i=1}^n R_{\beta_i, \alpha_i}(X_i)
\end{equation}
for all $X_1,\dots,X_n$, provided the right-hand side of \eqref{eq:e0} is well defined (not ``$\infty - \infty$").\footnote{The inequality in \eqref{eq:e0} is essentially Theorem 1 of \cite{ELW18}, which requires a condition on integrability. We slightly generalize this result  to probability measures without finite means, which will be useful for the generality of results offered in this paper; see Lemma \ref{lem:1} in the appendix.  Also note that our parameterization is slightly different from \cite{ELW18}.}
%\com{Here, let us include the boundary cases, but carefully discuss the domains. Please check whether the statements here and later are valid for the boundary cases. If it is too complicated, we can exclude some boundary cases. We can also assume $L^1$ throughout, and comment on the fact that $L^1$ is not important unless RVaR is ES.}
The objective of \cite{ELW18} is the risk sharing problem where the aggregate risk $X$ and the preferences of the agents are known (thus, $\alpha_1,\dots,\alpha_n,\beta_1,\dots,\beta_n$ are given) and one optimizes $ \sum_{i=1}^n R_{ \beta_i , \alpha_i}(X_i) $ over possible allocations $X_1,\dots,X_n$ satisfying $X_1+\dots+X_n=X$.

In this paper, we use the reverse direction of \eqref{eq:e0}: we fix $\boldsymbol \mu =(\mu_1,\dots,\mu_n)\in \M^n$ and $t, s$ with $0 \leq t < t+s \leq 1$, and aim to find the worst-case value of the aggregate risk $R_{t,s}(\nu)$ over $\nu\in \Lambda(\boldsymbol \mu)$ using \eqref{eq:e0}.
% For this purpose, we can rewrite \eqref{eq:e0} as
%\begin{equation}\label{eq:e1}
%R^+_{ \beta, \alpha}(\nu) \leq \sum_{i=1}^n R^+_{\beta_i, \alpha_i}(\mu_i),
%\end{equation}
%where $\nu \in \Lambda(\boldsymbol{\mu})$.
%%We are particularly interested in the worst-case quantiles, that is, the case $s=0$.
%Equation \eqref{eq:e1} will be the foundation on which we build our results in this paper.
For any $0 \le t < t+s \leq 1$, $\beta_0 \in [s, t+s]$, $\nu \in \Lambda(\boldsymbol{\mu})$, noting that $R_{t,s} \leq R_{t+s-\beta_0,\beta_0}$,  \eqref{eq:e0} leads to
\begin{equation}\label{eq:RVaR_ineq}
	R_{t, s}(\nu) \le R_{\sum_{i=1}^n \beta_i, \beta_0}(\nu) \le \sum_{i=1}^n R_{\beta_i, \beta_0}(\mu_i),
\end{equation}
where $\sum_{i=1}^n \beta_i = t+s-\beta_0$. Taking a supremum among all $\nu \in \Lambda(\boldsymbol{\mu})$ and an infimum among all feasible $(\beta_0, \beta_1, \cdots, \beta_n)$ in \eqref{eq:RVaR_ineq}, we get, for any fixed $(t,s)$ with $0 \le t < t+s \leq 1$,
\begin{equation}\label{eq:RVaR_bound}
	\sup_{\nu \in \Lambda(\boldsymbol{\mu})} R_{t, s}(\nu)  \leq \inf_{\substack{\sum_{i=0}^n \beta_i = t+s \\\beta_0\ge s > 0}}\sum_{i=1}^n R_{\beta_i, \beta_0}(\mu_i).
\end{equation}
The right-hand side of \eqref{eq:RVaR_bound} depends only on the marginal distributions $\mu_1,\dots,\mu_n$ and $(t,s)$, and thus we obtain a novel upper bound on the  worst-case RVaR aggregation. We shall refer to the bound  in  \eqref{eq:RVaR_bound} as a \emph{convolution bound}, since it is obtained from the inf-convolution formula in \eqref{eq:e-1}. To simplify notation, for each $n\in \N$, let $$\Delta_n= \left\{(\beta_0,\beta_1,\dots,\beta_{n}) \in (0,1)\times [0,1)^{n}: \sum_{i=0}^{n} \beta_i =1 \right\},$$
which is the set of vectors in the standard $(n+1)$-simplex with positive first component.  %\com{Change discussion later}
In all results, $\boldsymbol \beta$ represents $(\beta_0,\beta_1,\dots,\beta_{n}) $.

We formally present the convolution bound in Theorem \ref{th:qa-4prime} below. More importantly, we show that this bound is indeed sharp under  a few sets of conditions, and hence the convolution bounds are useful in calculating worst-case values in risk aggregation problems.
As far as we are aware of, Theorem \ref{th:qa-4prime} is the only result in the literature on RVaR aggregation with given marginal distributions.
The practically relevant case of quantiles ($s \downarrow 0$) will be discussed in detail in Section \ref{sec:2}.

Throughout, by ``admitting a decreasing density" we mean that the distribution has a left-bounded support  and it has a decreasing probability density function with respect to the Lebesgue measure on its support. The case for ``admitting an increasing density" is analogous.

\begin{theorem}\label{th:qa-4prime}
	Let $\boldsymbol \mu=(\mu_1,\dots,\mu_n)\in \M^n$. For any $t, s$ with $0 \leq t < t+s \leq 1$,
	\begin{equation}\label{eq:prime1}
		\sup_{\nu \in \Lambda(\boldsymbol \mu) }R_{t,s} (\nu) \le
		\inf_{\substack{\boldsymbol \beta\in (t+s)\Delta_n \\\beta_0\ge s > 0}}
		\sum_{i=1}^n R_{\beta_i,\beta_0 }(\mu_i).
		%  R_{\boldsymbol{\beta}}^+ (\boldsymbol{\mu}).
	\end{equation}
	Moreover, \eqref{eq:prime1} holds as an equality in the following cases:
	\begin{enumerate}[(i)]
		\item \label{item:t0} $t=0$;
		%\item[(ii)] $n\le2$;
		\item \label{item:rvar_decr}
		each of $\mu_1,\dots,\mu_n$ admits a decreasing density beyond its $(1-t-s)$-quantile;
		\item \label{item:mutual} $\sum_{i=1}^n \mu_i \left(q^+_{1-t-s}(\mu_i), q^-_{1}(\mu_i)\right] \le t+s$.%$\sum_{i=1}^n \mu_i \left(q^+_{1-t-s}(\mu_i), \infty\right) \le t+s$.
	\end{enumerate}
\end{theorem}

In Theorem \ref{th:qa-4prime}, case \eqref{item:t0} corresponds to the aggregation of ES, which is well known in the literature, e.g., Chapter 8 of \cite{MFE15}. 
Case \eqref{item:rvar_decr} in Theorem \ref{th:qa-4prime} is the most useful as decreasing densities are common in many areas of applications, including but not limited to finance and insurance. 
The proof of this case is quite technical, and it relies on  advanced  results on robust risk aggregation established in \cite{WW16} and \cite{JHW16}.
Case \eqref{item:mutual} corresponds to an assumption which allows for a lower mutually exclusive (see \cite{PW15} and also Definition \ref{def:lme} in Appendix \ref{app:D1}) 
random vector following marginal distributions $\mu_1,\dots,\mu_n$. 
%{\color{red}The definitions are given below. 
%\begin{example}
%	Let $n = 3$, $t + s = 1$ and $X_i \sim \mu_i = \text{Bernoulli}(p_i)$, $i = 1, 2, 3$. Assume $\sum_{i=1}^3 p_i \leq 1$. Hence, $(\mu_1, \cdots, \mu_3)$ satisfies condition (iii) in Theorem \ref{th:qa-4prime}. Hence, we have a lower mutual exclusive random vector $(X_1, \cdots, X_3)$ satisfying Definition \ref{def:lme}. We have
%	$$
%	(X_1, X_2, X_3) = \left\{
%	\begin{aligned}
%		& (0, 0, 0), && \text{with probability } 1-\sum_{i=1}^3 p_i,\\
%		& (1, 0, 0), && \text{with probability } p_1,\\
%		& (0, 1, 0), && \text{with probability } p_2,\\
%		& (0, 0, 1), && \text{with probability } p_3.
%	\end{aligned}
%	\right.
%	$$
%	$$
%	R_{0, 1}(\nu) = \frac{1}{1-t} \int_{0}^{1-t} q_u^-(X_1 + X_2 + X_3) \d u = \left\{
%	\begin{aligned}
%		& \frac{\sum_{i=1}^3 p_i - t}{1-t}, && \text{ if } 1-t > 1-\sum_{i=1}^3 p_i,\\
%		& 0, && \text{ if } 1-t \leq 1-\sum_{i=1}^3 p_i.
%	\end{aligned}
%	\right.
%	$$
%	We only consider the case $1-t > 1-\sum_{i=1}^3 p_i$. 
%	Let $\beta_1 = \min\{ t, p_1 \}$, $\beta_2 = \max\{ 0, \min\{ t-p_1, p_2 \} \}$, $\beta_3 = \max\{0, t- p_1 -p_2\}$. We have
%	$$
%	\sum_{i=1}^n R_{\beta_i,\beta_0 }(\mu_i) = \frac{1}{1-t} \left( \int_{t-p_1}^{1-p_1} + \int_{t-p_2}^{1-p_2} + \int_{p_1 + p_2}^{1-t+p_1+p_2}  \right) = \frac{1}{1-t} \left( 0 + 0 + (1-t+p_1+p_2) - (1-p_3) \right) = \frac{\sum_{i=1}^3 p_i - t}{1-t}.
%	$$
%\end{example}
%}
Such a situation is not common, but it may happen in the context of credit portfolio analysis, where  each $\mu_i$ represents the distribution of loss from a defaultable security which has a small probability of being positive. 
For instance, take $t=s=0.05$, $n=50$ and let $\mu_i$ be Bernoulli distributions with  $\mu_i(\{1\})=0.001$ for $i=1,\dots,n$. In this example, the aggregate risk represents the loss from a portfolio of defaultable bonds with default probability $0.001$,
and the condition in case \eqref{item:mutual} is satisfied because $\sum_{i=1}^n \mu_i \left(q^+_{1-t-s}(\mu_i), q^-_{1}(\mu_i)\right]  =\sum_{i=1}^n \mu_i(\{1\})=0.05\le t+s$. 
%{\color{red} Taking $t + s = 1$ and $q_0^+(\mu_i) = 0$, $i = 1, \cdots, n$ as an example, the case \eqref{item:mutual} means $\mu_1, \cdots, \mu_n$ satisfy lower mutually exclusivity. It means in credit portfolio analysis that there is zero probability to have any two defaultable securities being positive at the same time. 
%$$
%\p(X_i > q_0^+(\mu_i), X_j > q_0^+(\mu_j)) = 0, \;\; \text{ for any } i, j = 1, \cdots, n \text{ and } i \neq j.
%$$
%}   \com{YL: I am not sure about the language.} 
The proof for case \eqref{item:mutual} is based on convenient properties of a  mutually  exclusive random vector. Moreover, we will show in Figure \ref{fig:rvar_inc} (right panel) in Section \ref{sec:8} that the bound \eqref{eq:prime1} is not sharp for marginals with increasing densities, even for homogeneous marginals; however for quantiles (limits of RVaR), the bound becomes sharp for increasing densities (Theorem \ref{th:qa-4}).

Results that are symmetric to the upper convolution  bounds are collected in
Appendix \ref{sec:lower}. For instance,
a lower bound on $\inf_{\nu\in \Lambda(\boldsymbol \mu)} R_{t,s}(\nu)$, which is symmetric to Theorem \ref{th:qa-4prime},
is given in Theorem \ref{th:qa-4primeprime}.

Case \eqref{item:rvar_decr} in Theorem \ref{th:qa-4prime}
involves conditional distributions above a certain quantile. For $\mu \in \M$ and $t\in [0,1)$, let $\mu^{t+}$ be the probability measure given by
$$
\mu^{t+} (-\infty,x] =\max\left\{ \frac{\mu(-\infty,x]-t}{1-t}, 0\right\},~~x\in \R.
$$
The probability measure $\mu^{t+}$ is called the $t$-tail distribution of $\mu$  by \cite{RU02}.
In other words, $\mu^{t+}$ is the distribution measure of the random variable $q_U(\mu)$ where $U$ is a uniform random variable on $[t,1]$.
Equivalently, $\mu^{t+}$ is the distribution measure of $\mu$ restricted beyond its $t$-quantile. %, while $\mu^{t-}$ is the distribution measure of $\mu$ restricted below its $t$-quantile.
For example, the statement in case \eqref{item:rvar_decr} that $\mu$ admits a decreasing density beyond its $(1-t-s)$-quantile is equivalent to the one that $\mu^{(1-t-s)+}$ admits a decreasing density. Moreover, by direct computation, for fixed $\mu \in \M$ and $t \in [0, 1)$, we have
\begin{equation}\label{eq:rvar_scale}
	\begin{aligned}
		 R_{\beta, \alpha} (\mu^{t+}) & = R_{(1-t)\beta, (1-t)\alpha}(\mu), ~ \text{ for all $0 \leq \beta < \beta+\alpha \leq 1$};\\
	 q_u^- (\mu^{t+}) 	& = q_{t+(1-t) u}^- (\mu), ~ \text{ for all $u \in (0,1]$}.
	\end{aligned}
\end{equation}

Using \eqref{eq:rvar_scale}, we obtain Proposition \ref{prop:basic} below  based on Theorem 4.1 of \cite{LW16}. This result is useful in the proof of Theorem \ref{th:qa-4prime}. For $\boldsymbol \mu=(\mu_1,\dots,\mu_n)\in \M^n$ and $t\in [0,1)$, denote by $ \boldsymbol \mu^{t+}= (\mu_1^{t+},\dots,\mu_n^{t+})$.  %\com{Don't use fraction like $\frac{u-(1-d-s)}{d+s}$ in text. Use $({u-(1-d-s)})/({d+s})$.}
%That is the reason why we will mainly focus on these two quantities in the following content.

\begin{proposition}\label{prop:basic}
	For $\boldsymbol \mu=(\mu_1,\dots,\mu_n)\in \M^n$, $t \in [0,1)$ and $s \in (0, 1-t]$, we have
	$$
	\sup_{\nu \in \Lambda(\boldsymbol \mu) }R_{t,s} (\nu) = \sup_{\nu \in \Lambda(\boldsymbol \mu^{(1-t-s)+ } )}\LES_{\frac{s}{t+s}} (\nu)
	$$
	and
	$$
	\sup_{\nu \in \Lambda(\boldsymbol \mu) }q_{t}^+ (\nu) = \sup_{\nu \in \Lambda(\boldsymbol \mu^{t+} )}q_0^+ (\nu).
	$$
\end{proposition}

Proposition \ref{prop:basic} suggests that for the worst-case problems of RVaR aggregation, it suffices to consider the one started from quantile level 0, i.e. the $\LES$ aggregation. In particular, for the worst-case problems of quantile aggregation, it suffices to consider the one at quantile level 0, i.e. the problems $\sup_{\nu \in \Lambda(\boldsymbol \mu^{t+} )}q_0^+ (\nu)$ for generic choices of $\boldsymbol \mu$.
This result will be used repeatedly in our discussions, and it will be the general approach taken in the proof of our main results.

\section{Convolution bounds on quantile aggregation}\label{sec:2}

\subsection{Convolution bounds}

In Theorem \ref{th:qa-4} below we summarize bounds on $\sup_{\nu \in \Lambda(\boldsymbol \mu) }q_t^+ (\nu)$. Most cases can be obtained by sending $s$ to $0$ and replacing $t$ with $(1-t)$ in Theorem \ref{th:qa-4prime}, but a notable difference is that the convolution bounds are  sharp for both decreasing and increasing densities, for $n=2$, and for two types of mutual exclusivity (see Appendix \ref{app:D1}). This is in drastic contrast to the RVaR convolution bounds which are only sharp for decreasing densities or upper mutual exclusivity (see Figure \ref{fig:rvar_inc}).
Results on  lower bounds on $ q^-_t (\nu)$ are put in Appendix \ref{sec:lower}. In particular, Theorem \ref{th:qa-2} is symmetric to Theorem \ref{th:qa-4}.

\begin{theorem}\label{th:qa-4}
	For $\boldsymbol \mu\in \M^n$ and $t\in [0,1)$, we have
	\begin{equation}\label{eq:main1pr}
		\sup_{\nu \in \Lambda(\boldsymbol \mu) }q^+_{t} (\nu) \le \inf_{\boldsymbol \beta\in (1-t)\Delta_n}   \sum_{i=1}^n R_{\beta_i,\beta_0 }(\mu_i).
	\end{equation}
	Moreover, \eqref{eq:main1pr} holds as an equality in the following cases:
	\begin{enumerate}%[(i)]
		\item[(i)] $n \leq 2$;
		\item[(ii)] \label{item_decr}
		each of $\mu_1,\dots,\mu_n$ admits a decreasing density beyond its $t$-quantile;
		\item[(iii)] \label{item_incr}
		each of $\mu_1,\dots,\mu_n$ admits an increasing density beyond its $t$-quantile;
		%\item $\sum_{i=1}^n \mu_i\left((q^-_{t}(\mu_i),\infty)\right) \le 1-t$;
		%\item $\sum_{i=1}^n \left( \mu_i\left((-\infty, q^+_{t}(\mu_i))\right) - t \right) \le 1-t$.
		\item[(iv)] $\sum_{i=1}^n \mu_i\left(q^+_{t}(\mu_i),q^-_{1}(\mu_i)\right] \le 1-t$;%$\sum_{i=1}^n \mu_i\left(q^+_{t}(\mu_i),\infty\right) \le 1-t$;
		\item[(v)] $\sum_{i=1}^n \mu_i  \left[ q^+_{t}(\mu_i), q^-_{1}(\mu_i) \right)  \le 1-t$.
	\end{enumerate}
\end{theorem}

\begin{remark}
	If $\mu_1,\dots,\mu_n$ have positive densities on their supports, then
	$\sup_{\nu \in \Lambda(\boldsymbol \mu) }q^-_t (\nu) = \sup_{\nu \in \Lambda(\boldsymbol \mu) }q^+_t (\nu) $
	%and $\inf_{\nu \in \Lambda(\boldsymbol \mu) }q^-_t (\nu) = \inf_{\nu \in \Lambda(\boldsymbol \mu) }q^+_t (\nu) $
	for all $t\in (0,1)$; see Lemma 4.5 of \cite{BJW14}.
	Hence, using $ q^-_t (\nu) $ or $q^+_t (\nu)$  in Theorem \ref{th:qa-4} is not essential to our discussions.
\end{remark}

\begin{remark}
 The   classic probability bound     $\p( \sum_{i=1}^n X_i \ge \sum_{i=1}^n z_i) \le \sum_{i=1}^n \p(X_i \ge z_i)$ for all $z_1,\dots,z_n\in \R$,
  is a special case of Theorem \ref{th:qa-4} by converting quantile bounds into probability bounds.
To see this, let $\mu_i$ be the distribution of $X_i$ and  $t_i=\p(X_i \ge z_i) $ for $i\in \{1,\dots,n\}$, and let $\nu$ be the distribution of $\sum_{i=1}^n X_i$.
  The bound \eqref{eq:main1pr} gives
  $ q^+_{1-\sum_{i=1}^n t_i} (\nu) \le    \sum_{i=1}^n q^-_{1-t_i  }(\mu_i) \le \sum_{i=1}^n z_i$.
  This implies $\p( \sum_{i=1}^n X_i \ge \sum_{i=1}^n z_i)\le \sum_{i=1}^n t_i$.
\end{remark}

%The right-hand side of \eqref{eq:main1pr} in Theorem \ref{th:qa-4} only involve  the average quantiles of the marginal distributions $\mu_1,\dots,\mu_n$ and they are  easy to calculate without involving dependence structures. %In Sections \ref{sec:7} and \ref{sec:8}, we will compare the proposed bounds with the homogeneous dual bound of \cite{EP06} and rearrangement algorithm (RA) of \cite{EPR13} by numerical examples. The upper bound we provide in \eqref{eq:main1pr} for the worst-case quantile aggregation, together with the lower bound provided by the RA, confines the true value in a narrow scope; see Section \ref{sec:8}.
%We verify the sharpness of the bounds in Theorem \ref{th:qa-4} under some conditions. The bounds in Theorem \ref{th:qa-4} are always sharp for $n=2$. Notably, they are sharp for distributions with decreasing (or increasing) densities on their supports.
In the literature, some sharp bounds on quantile aggregation for decreasing densities are obtained by \cite{WPY13} and \cite{PR13} in the homogeneous case $(\mu_1=\dots=\mu_n)$ and \cite{JHW16} in the heterogeneous case.
For the heterogeneous case, the method of \cite{JHW16} involves solving a system of $(n+1)$-dimensional implicit ODE (equations (E1) and (E2) of \cite{JHW16}), which requires a highly complicated calculation.
In contrast, our result in Theorem \ref{th:qa-4} gives sharp bounds based on the minimum or maximum of an $(n+1)$-dimensional function.

In the homogeneous case $\mu_1=\dots=\mu_n$, as an immediate consequence of Theorem \ref{th:qa-4}, we obtain the following reduced bounds in which one replaces  $\inf_{\boldsymbol \beta\in (1-t)\Delta_n}   \sum_{i=1}^n R_{\beta_i,\beta_0 }(\mu_i)$
by a one-dimensional optimization problem. 
We show that, in some homogeneous case, the sharp result in Theorem \ref{th:qa-4} can be achieved by the reduced bound. A proof of this result follows from  a combination of Theorem \ref{th:qa-4} and Proposition 1 of \cite{EPRWB14}.
In what follows, $\Lambda_n(\mu) = \Lambda(\mu, \cdots, \mu)$ is the set of the aggregate distribution measures with the homogeneous marginal $\mu$.
\begin{proposition}[Reduced convolution bounds]\label{prop:reduced_bound}
	For $\mu\in \M $ and $t\in [0,1)$, we have
	\begin{equation}\label{eq:main1hom}\sup_{\nu \in \Lambda_n(\mu) }q_t^+ (\nu) \le  \inf_{\alpha \in  (0, (1-t)/n ) } n R_{\alpha, 1-t-n\alpha}(\mu) = \inf_{\alpha \in (0, (1-t)/n)} \frac{n}{1-t-n\alpha}\int_{t+(n-1)\alpha}^{1-\alpha} q_u^{-}(\mu) \d u.
	\end{equation}
 Moreover, \eqref{eq:main1hom} holds as an equality if $\mu$ admits a decreasing density beyond its $t$-quantile.
\end{proposition}

First, it is clear that the convolution bound \eqref{eq:main1pr} is better (smaller) than the reduced one \eqref{eq:main1hom}, while the latter is easier to compute. They are not generally equal. Second, %\com{Yang: Yes. The right-hand side of \eqref{eq:main1pr} and that of \eqref{eq:main1hom} are not generally equal. Emphasized.}
in case $\mu$ admits a decreasing density, Proposition 8.32 of \cite{MFE15} (reformulated from  \citet[Theorem 3.4]{WPY13}) gives
$$ \sup_{\nu \in \Lambda_n(\mu) }q_t^+ (\nu)  =  \frac{n}{1-t-n\alpha}\int_{t+(n-1)\alpha}^{1-\alpha} q_u^{-}(\mu) \d u
$$
for some $\alpha \in [0,(1-t)/n)$. Together with \eqref{eq:main1pr}, we get the sharpness of \eqref{eq:main1hom}.

Since quantiles commute with strictly increasing transforms,
Theorem \ref{th:qa-4} leads to a multiplicative version of the convolution bounds,
which can be useful for some applications, in particular, the O-ring theory  in Section \ref{sec:o-ring} and Section \ref{sec:OR}.
Recall that for any random variable $X $ following distribution $\mu$ and any Borel function $f$, the random variable $f(X)$ has distribution $\mu\circ f^{-1}$ where  $ f^{-1}$ is the set-valued inverse of $f$. 
%
%We present a multiplicative result on quantile aggregation in Proposition \ref{th:multiply}. The lower-bound result \eqref{eq:main2pr:multiply} is listed here because we will use it in Section \ref{sec:o-ring}. \com{YL: multiplicative result}%on lower bounds on $ q^-_t (\nu)$ are put in Appendix \ref{sec:lower}. In particular, Theorem \ref{th:qa-2} is symmetric to Theorem \ref{th:qa-4}.
\begin{proposition}\label{th:multiply}
	For $\mu_1, \cdots, \mu_n \in \M$ with support included in $(0, \infty)$, we have
	\begin{equation}\label{eq:main1pr:multiply}
		\sup_{X_i \sim \mu_i, i = 1, \cdots, n } q^+_{t} \left( \prod_{i = 1}^n X_i \right) \le  \exp \left (  \inf_{\boldsymbol \beta\in (1-t)\Delta_n}  \sum_{i=1}^n R_{\beta_i,\beta_0 } \left( \mu_i\circ \exp\right) \right ), ~~ t\in [0,1).
	\end{equation}
	Moreover, \eqref{eq:main1pr:multiply} holds as an equality in the following cases (denote by $f_1, \dots, f_n$ the densities of $\mu_1, \dots, \mu_n$):
	\begin{enumerate}%[(i)]
		\item[(i)] $n \leq 2$;
		\item[(ii)]
		for each $i = 1, \cdots, n$,  $x\mapsto xf_i(x)$  is decreasing beyond the $t$-quantile of $\mu_i$;
		\item[(iii)]
		for each $i = 1, \cdots, n$, $x\mapsto xf_i(x)$ is increasing beyond the $t$-quantile of $\mu_i$;
		\item[(iv)] $\sum_{i=1}^n \mu_i\left(q^+_{t}(\mu_i),q^-_{1}(\mu_i)\right] \le 1-t$;%$\sum_{i=1}^n \mu_i\left(q^+_{t}(\mu_i),\infty\right) \le 1-t$;
		\item[(v)] $\sum_{i=1}^n \mu_i  \left[ q^+_{t}(\mu_i), q^-_{1}(\mu_i) \right)  \le 1-t$;
	\end{enumerate}
\end{proposition}

\subsection{Technical discussions}\label{sec:technical}%: non-sharpness, attainability and truncation}

We do not expect that the formula \eqref{eq:main1pr} always gives sharp bounds, and this is a situation similar to Theorem \ref{th:qa-4prime}. A counter-example of non-sharpness of the bounds in Theorem \ref{th:qa-4} is presented in Section \ref{sec:num_var} with some discrete marginal distributions (see also Example \ref{ex:ex1} in Appendix \ref{app:A}). Nevertheless, in most cases, the bounds in Theorem \ref{th:qa-4} work quite well, as illustrated by the numerical examples later. 
In some special cases, the  reduced bounds in Proposition \ref{prop:reduced_bound} are equivalent to those in Theorem \ref{th:qa-4}. We shall show this does not generally hold (e.g., for some distribution with increasing density) later in Figure \ref{fig:bound_p} (right panel). % and Example \ref{ex:ex0}.

%\subsubsection{Attainability and well-posedness}
In the following proposition, we note that $\sup_{\nu \in \Lambda(\boldsymbol \mu) }q_t^+ (\nu)$ is always attainable as a maximum, which is implied by Lemma 4.2 of \cite{BJW14}.
\begin{proposition}\label{prop:nu}
	For $\boldsymbol \mu\in \M^n$ and $t \in [0,1)$,  there exists $\nu_+ \in \Lambda(\boldsymbol \mu)$ such that
	$$
	\sup_{\nu \in \Lambda(\boldsymbol \mu) }q_t^+ (\nu)= q_t^+ (\nu_+).
	$$
\end{proposition}

We next turn to the right-hand side of \eqref{eq:main1pr}.
Because of the continuity of $R_{\alpha,\beta}$ in $\alpha,\beta\in [0,1]$, the infimum in  $\inf_{\boldsymbol \beta\in (1-t)\Delta_n}   \sum_{i=1}^n R_{\beta_i,\beta_0 }(\mu_i)$ for any $t \in [0,1)$ is attainable in the closure $\overline{\Delta}_n$ of $\Delta_n$; 
%$$
%\begin{aligned}
%	\overline{\Delta}_n = \left\{(\beta_0, \beta_1, \cdots, \beta_n) \in [0,1]^{n+1}: \sum_{i=0}^n \beta_i = 1 \right\};
%\end{aligned}
%$$
see Appendix \ref{app:attain} for details.

 To address  computational efficiency, 
we first focus on the case of monotone densities which are sufficient for (ii) and (iii) in Theorem \ref{th:qa-4} for any $t$. These two assumptions will be used repeatedly later. \begin{enumerate}[(DD)]
		\item[(DD)]
		each of $\mu_1,\dots,\mu_n$ admits a decreasing density;
		\item[(ID)]
		each of $\mu_1,\dots,\mu_n$ admits an increasing density.
	\end{enumerate}
	Under condition  (DD) or (ID),  we can formally argue that the convolution bound is  easy to compute. For an  illustration, consider the infimum problem in \eqref{eq:main1pr} with the condition (DD); here we take $t=0$ without loss of generality due to Proposition \ref{prop:basic}.   
For a fixed $\beta_0\in (0,1)$, note that the mapping 
\begin{align*}
\phi_i: \beta_i\mapsto  \frac{1}{\beta_0} \int_{1-\beta_i-\beta_0}^{1-\beta_i} q_u^-(\mu_i) \d u 
\end{align*}
is convex, because $u\mapsto q_u^-(\mu_i) $ is convex under (DD)
which implies that $\beta_i\mapsto   q_{1-\beta_i-\beta_0}^-(\mu_i)- q_{1-\beta_i}^-(\mu_i) $ is increasing. %\com{Henry: if the problem is non-convex, how do we make sure the numerical result is correct? i.e., perhaps the issue is not about computational speed but more about whether we can have a guaranteed solution or not..}
Therefore,  for fixed $\beta_0$,
$$
(\beta_1,\dots,\beta_n) \mapsto  \sum_{i=1}^n R_{ \beta_i  ,\beta_0 }(\mu_i)  = \sum_{i=1}^n \frac{1}{\beta_0} \int_{1-\beta_i-\beta_0}^{1-\beta_i} q_u^-(\mu_i) \d u 
$$
is 
convex since it is 
the sum of convex functions in each component.  The full optimization can be   converted to an $n$-dimensional convex minimization problem  over $(\beta_1,\dots,\beta_n)$ and 
a one-dimensional problem of optimization over $\beta_0$, which is not necessarily convex. The objective is continuous in $\beta_0$, so that the one-dimensional problem is computable by suitable discrete approximation up to any specified accuracy.
% As such, the computation complexity is relatively low.
In case (ID) holds, the objective is concave in $(\beta_1,\dots,\beta_n)$,
and its solution   always lies on the boundary of the simplex $(1-\beta_0)\overline{\Delta}_{n-1}$. 
When (DD) and (ID) do not hold, the above optimization may be more complicated, but in our numerical experiments in Section \ref{sec:8}, they are always solved quite fast and produce results that are consistent with other methods.
%\textcolor{blue}{
The convolution bound is also compared with a discrete linear programming  formulation in Appendix \ref{sec:R2-1}, showing its advantages in computational time and feasibility in high dimensions.
The next proposition concerns the truncation of the marginal distributions. When calculating the supremum of $q^+_0$ for the aggregation of non-negative risks,
one can safely truncate the marginal distributions at a high threshold. This result is convenient when applying several results in the literature  formulated for  distributions with finite mean or a compact support, including Theorem 1 of \cite{ELW18}.
For a probability measure $\mu \in \M$ and a constant $m \in \R$, let $\mu^{[m]}$ be the distribution of $X\wedge m$ where $X\sim \mu$ and $x\wedge y$ stands for the minimum of two numbers $x$ and $y$. Further denote that $\boldsymbol{\mu}^{[m]} = (\mu_1^{[m]}, \cdots, \mu_n^{[m]})$ for $\boldsymbol{\mu} = (\mu_1, \cdots, \mu_n) \in \M^n$.

\begin{proposition}\label{prop:1end}
	For any distributions  $\mu_1,\dots,\mu_n$   on $[0,\infty]$, $t\in [0,1)$,
	and $m \ge \sum_{i=1}^n q^+_{1-(1-t)/n}(\mu_i)$,
	we have
	\begin{equation}\label{eq:2}
		\sup_{\nu \in \Lambda(\boldsymbol{\mu})} q^+_t(\nu) =  \sup_{\nu \in \Lambda(\boldsymbol{\mu}^{[m]})} q^+_t(\nu).
	\end{equation}
\end{proposition}

\subsection{Quantile aggregation at levels $0$ and $1$}\label{sec:21}

Now we restate the specific cases of quantile aggregation $q^+_0$ and $q^-_1$, where an analogous result to Theorem \ref{th:qa-4} is used; see Appendix \ref{sec:lower}.
\begin{proposition}[Convolution bounds at levels $0$ and $1$] \label{prop:qa-1}
	For $\boldsymbol \mu\in \M^n$, we have
	\begin{equation}\label{eq:main1}\sup_{\nu \in \Lambda(\boldsymbol \mu) }q_0^+ (\nu) \le \inf_{\boldsymbol \beta\in \Delta_n} \sum_{i=1}^n R_{ \beta_i  ,\beta_0 }(\mu_i),\end{equation}
	and
	\begin{equation}\label{eq:main2}\inf_{\nu \in \Lambda(\boldsymbol \mu) } q_1^- (\nu) \ge \sup_{\boldsymbol \beta \in \Delta_n} \sum_{i=1}^n R_{1-\beta_i-\beta_0,\beta_0 }(\mu_i).\end{equation}
	The two bounds are both sharp if $n \le 2$, or  each of $\mu_1,\dots,\mu_n$ admits a decreasing (respectively, increasing) density on its support.
\end{proposition}

If $\mu_1,\dots,\mu_n$ have finite means, the inequalities in \eqref{eq:main1} and \eqref{eq:main2} can be combined into a chain of inequalities.

\begin{proposition}\label{prop:prop1}
	For $\boldsymbol \mu=(\mu_1,\dots,\mu_n)\in \M_1^n$, we have
	\begin{equation}\label{eq:main1p}\inf_{\nu \in \Lambda(\boldsymbol \mu) } q_1^- (\nu) \ge \sup_{\boldsymbol \beta \in \Delta_n} \sum_{i=1}^n R_{1-\beta_i-\beta_0,\beta_0 }(\mu_i) \ge \sum_{i=1}^n R_{0,1}(\mu_i)\ge   \inf_{\boldsymbol \beta\in \Delta_n} \sum_{i=1}^n R_{ \beta_i  ,\beta_0 }(\mu_i) \ge  \sup_{\nu \in \Lambda(\boldsymbol \mu) }q_0^+ (\nu).\end{equation}
\end{proposition}
The tuple of distributions $\boldsymbol \mu \in \M^n$ is said to be \emph{jointly mixable} (JM) if $\delta_C\in  \Lambda(\boldsymbol \mu ) $ for some $C\in \R$; see Appendix \ref{app:3}. %{\color{red}The terminology of ``joint mixability" comes from \cite{WPY13} and is fully investigated in \cite{WW16}. It is defined as a capability of a vector of marginal distributions, such that in some kind of dependence, the sum of the corresponding marginal variables is able to be a constant on the probability space. In this case, the marginal variables are ``mixing" on the probability space, and the variance of the sum variable is zero.} 
Proposition \ref{prop:prop1} implies that \eqref{eq:main1} and \eqref{eq:main2} become sharp if $\boldsymbol \mu \in \M_1^n$ is JM. If $\mu_1,\dots,\mu_n$  do not have finite means, the relationships
%$$\sup_{\boldsymbol \beta \in \Delta_n} R^{-}_{\boldsymbol \beta}(\boldsymbol \mu) \ge   \inf_{\boldsymbol \beta\in \Delta_n}   R^{+}_{\boldsymbol \beta}(\boldsymbol \mu)	\mbox{~~and~~}\inf_{\nu \in \Lambda(\boldsymbol \mu) } q_1^- (\nu) \ge   \sup_{\nu \in \Lambda(\boldsymbol \mu) }q_0^+ (\nu)$$
in \eqref{eq:main1p} may not hold generally, which is illustrated by Example \ref{ex:cauchy} in Appendix \ref{app:A}.

\section{Approximation of the extremal dependence}\label{sec:new6}

A significant  advantage of the convolution bounds on the quantile aggregation problem is that we are able to visualize, in certain cases, the extremal dependence structure corresponding to the convolution bounds. In view of Proposition \ref{prop:basic}, for the problems of worst-case quantile aggregation, it suffices to consider the one at quantile level 0, i.e., $\sup_{\nu \in \Lambda(\boldsymbol \mu) }q_0^+ (\nu)$. Similarly, for the problems of the best-case quantile aggregation, it suffices to consider the one at quantile level 1, i.e., $\inf_{\nu \in \Lambda(\boldsymbol \mu) }q_1^- (\nu)$ as in   Proposition \ref{prop:basic_sym}.
The supremum and the infimum can be replaced by a maximum and a minimum, respectively, as implied by Proposition \ref{prop:nu}.

 We will  describe a   dependence structure, which approximately solves      $\max_{\nu \in \Lambda(\boldsymbol \mu) }q_0^+ (\nu)$ and  $\min_{\nu \in \Lambda(\boldsymbol \mu) } q_1^- (\nu)$ in certain cases.
If the marginal distributions all have decreasing densities as in Theorem \ref{th:qa-4} (ii) and Proposition \ref{prop:qa-1}, then this dependence structure precisely attains both the maximum and the minimum above.

\subsection{Extremal dependence structures: Monotone densities
}\label{sec:description}

We first focus on the case of monotone densities. 
To describe the optimal dependence structure,
we divide the sample space $\Omega$ into $(n+1)$ disjoint events
$A_1,\dots,A_n,B$; in other words,
$\Omega = A_1 \cup \cdots \cup A_n \cup B$. These sets have the following interpretations:
\begin{enumerate}
	\item[$(B)$] ``body": the event that all individual random variables take the ``medium value" of their distributions and the sum of them is a constant;
	\item[$(A_i)$] ``$i$-th right tail": the event that the $i$-th individual random variable takes a ``large value" and the other $(n-1)$ random variables take ``small values".
\end{enumerate}	

Intuitively, the dependence structure is summarized as ``joint mix''  (\cite{WW16}; see Appendix \ref{app:3}) and ``(approximate) mutual exclusivity".
Moreover, $A_i$ is a tail event of the $i$-th random variable $X_i$ by \cite{WZ21}.
The above dependence structure  is not completely specified, as one further needs to
 properly specify what we meant by ``large value", ``medium value" and ``small value", and on each event how the random variables are constructed and dependent.
 Unfortunately, it is in general not possible to provide an analytical description, if the marginal distributions are heterogeneous. In the homogeneous case with a decreasing density, an analytical description is possible, as discussed in \cite{WW11}.
More explicit formulas of this dependence structure will be discussed in Section \ref{sec:6}.

An optimal  structure for the problem of both $\max_{\nu \in \Lambda(\boldsymbol \mu) }q_0^+ (\nu)$ and $\min_{\nu \in \Lambda(\boldsymbol \mu) } q_1^- (\nu)$ admits the above dependence structure  when (DD) holds, as observed by \cite{JHW16}; a formal and more general  result on this observation is   Theorem \ref{th:qa-5} below.
The optimality comes from a result of \cite{JHW16} where it is shown that
  the sum under this dependence structure is the minimum with respect to convex order given marginal distributions.
 Moreover, this dependence structure leads to an approximation to
optimality in many relevant situations that (DD) does not hold; some numerical results will be shown in Section \ref{sec:performance}.

We note that the optimal structure for $\max_{\nu \in \Lambda(\boldsymbol \mu) }q_0^+ (\nu)$ or $\min_{\nu \in \Lambda(\boldsymbol \mu) } q_1^- (\nu)$ is not unique in general, and in this section we only describe one such candidate. In all our follow-up discussions, we will focus on this candidate.

For now, assume that all marginal distributions have decreasing densities, i.e.,  (DD) holds.
Our method of convolution bounds allows us to determine the existence of the above events $A_1,\dots,A_n, B$ from the optimizing  vector $\boldsymbol{\beta}  $. We explain this below.
For $\boldsymbol \mu=(\mu_1,\dots,\mu_n)\in \M^n$ and $\boldsymbol \beta =(\beta_0,\beta_1,\dots,\beta_n)\in \overline{\Delta}_n$, we denote by
\begin{equation}\label{eq:R+}
	R^{+}_{\boldsymbol \beta}(\boldsymbol \mu) =\sum_{i=1}^n R_{ \beta_i  ,\beta_0 }(\mu_i).
\end{equation}
The bound  \eqref{eq:main1pr} is sharp in our setting by Theorem \ref{th:qa-4}.
Suppose that $\boldsymbol \beta \in \overline{\Delta}_n$ is the optimizer to \eqref{eq:main1pr}, that is,
$$\max_{\nu \in \Lambda(\boldsymbol \mu) }q_0^+ (\nu) = \inf_{\boldsymbol \beta'\in \Delta_n}   R^{+}_{\boldsymbol \beta'}(\boldsymbol \mu)= R^{+}_{\boldsymbol \beta}(\boldsymbol \mu).$$
%The main question here is how the knowledge of $\boldsymbol \beta$ in the above equality
%helps use to further identify or approximate the
%corresponding dependence structure attaining $\max_{\nu \in \Lambda(\boldsymbol \mu) }q_0^+ (\nu)$.
 We describe  a  classification on the existence of $A_1,\dots,A_n,B$ based on the obtained value of $\boldsymbol \beta$.
\begin{enumerate}[  1.]
	\item[  1.] If $\beta_1 = \cdots =\beta_n = 0$, then the optimal dependence structure is ``a full joint mix"; that is, the individual random variables add up to a constant on the whole probability space. Only the event $B$ occurs; all   events $A_i$ are of zero probability.
	
	\item[  2.] If $\beta_i \neq 0$ for $i \in I$ where the index set $I$ is a non-empty proper subset of $\{1, \cdots, n\}$, then %the optimal dependence is ``joint mixability'' and ``(approximate) mutual exclusivity". In this case,
	the events of ``body" and ``$i$-th right tail" occur; i.e., the possible events are $\left\{B,  A_i: i \in I \right\}$.
	
	\item[  3.] If $\beta_i \neq 0$ for  all $i = 1, \cdots, n$, then %the optimal dependence is ``joint mixability'' and ``(approximate) mutual exclusivity".	In this case,
	all $(n+1)$ events occur; i.e., the possible events are $\left\{B, A_1, \cdots, A_n \right\}$.
\end{enumerate}
To show the above classification statement, note that  $\beta_i$ indicates the maximum value  $i$-th random variable $X_i$ takes  on the event $B$. More precisely, the largest value of $X_i$ takes on $B$ is $q_{1-\beta_i}(\mu_i)$. Hence, $\beta_i=0$ means that there is no ``large" values of $X_i$ that is  considered as a ``tail", and thus $A_i$  does not occur.
  The quantile function of the corresponding sum is illustrated in an example by the left panel of Figure \ref{fig:quantileplot}.
\begin{figure}[t]
	\centering
	\caption{Quantile functions for the sum.  Left panel:  decreasing densities ($n = 3$, quantile functions are $\frac{6}{5} r(t)$, $\frac{4}{5} r(t)$ and $\frac{4}{5} r(t)$, where $r(t) = -\log(\epsilon + (1-\epsilon) (1-t)), t \in [0, 1]$) and $\epsilon =0.0001$); Right panel: increasing densities ($n = 3$, quantile functions are $-\frac{6}{5} r(1-t)$, $-\frac{4}{5} r(1-t)$ and $-\frac{4}{5} r(1-t)$, $t \in [0, 1]$.).  The events $A_1, \cdots, A_n, B$ are described in Theorem \ref{th:qa-5}.
	}
	\includegraphics[width=0.43\textwidth]{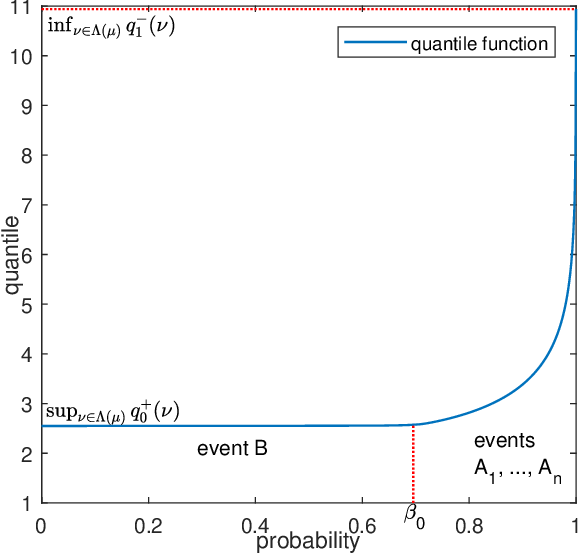}~~
	\includegraphics[width=0.43\textwidth]{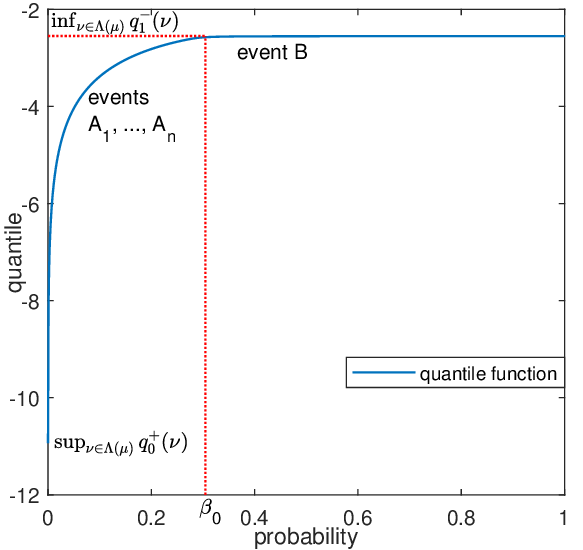}
	\label{fig:quantileplot}
\end{figure}

For a general $t\in (0,1)$, to build a corresponding dependence structure for $\max_{\nu \in \Lambda(\boldsymbol \mu) }q_t^+ (\nu)$, we need to
build the above  events $B$ and $A_1,\dots,A_n$ for the conditional distributions $\mu_1^{t+},\dots,\mu_n^{t+}$. These events will take up probability $1-t$ in total. The remaining event $C$ has probability $t$, which can specified as
\begin{enumerate}[($C$)]
	\item[($C$)] ``feet": the event that all individual random variables take   small values below their $t$-quantile.
	\end{enumerate}
Conditional on the event $C$, the dependence structure of $(X_1,\dots,X_n)$ no longer matters, as it does not contribute to the quantile of the sum.
For the optimality, it suffices to require (DD) to hold for  $\mu_1^{t+},\dots,\mu_n^{t+}$.
 Similarly we can deal with the case of $\min_{\nu \in \Lambda(\boldsymbol \mu) } q_t^- (\nu)$.

The above discussions also apply to the setting with (ID) in place of (DD) by replacing each involved random variable   $X_i$  with its negation $-X_i$. In this case, a similar dependence structure can be constructed based on an event $B$ of ``body" and $n$ events $A_1,\dots,A_n$ of ``left tail". An example is presented in the right panel of Figure \ref{fig:quantileplot}. We omit the details here.

In general, if (DD) and (ID) do not hold, but the densities are approximately increasing or decreasing, then we can still use the above construction, and obtain an approximately optimal structure. This will be discussed next.

\subsection{The general case and an approximation}\label{sec:6}

%In view of Proposition \ref{prop:basic}, for the problems of quantile aggregation, it suffices to consider the one at quantile level 0. Hence,

 Section \ref{sec:description} contains a   description of a class of dependence structures that leads to the optimized value of  $\max_{\nu \in \Lambda(\boldsymbol \mu) }q_t^+ (\nu)$ under the assumption (DD) or (ID).
In this section, we discuss more on this class of dependence structures and show that some further specifications may be used as an approximation for the cases without (DD) and (ID).

In the following, for $0 \le \alpha < \beta \le 1$
and any  probability  measure $\mu$, we  let $\mu^{[\alpha, \beta]}$ be the probability measure given by
$$
\mu^{[\alpha, \beta]}(-\infty,x]=\frac{\left(\min\left\{\mu(-\infty,x], \beta\right\} -\alpha\right)_+}{\beta-\alpha},~~x\in \R.
$$
Equivalently, $\mu^{[\alpha, \beta]}$ is the distribution measure of the random variable $q_V (\mu)$ where $V \sim \mathrm U[\alpha, \beta]$, a uniform random variable on $[\alpha, \beta]$. In particular, $\mu^{[\alpha, 1]} = \mu^{\alpha+}$ is the $\alpha$-tail distribution   of $\mu$ in Section \ref{sec:rvar}.

We say that a random vector $(X^*_1,\dots,X^*_n)$ attains the maximum of   $q_0^+$ for $\boldsymbol \mu=(\mu_1,\dots,\mu_n)\in \M^n$ if $X_1^*\sim \mu_1,\dots,X_n^*\sim \mu_n$ and $q_0^+(X_1^*+\dots+X_n^*)=\max_{\nu \in \Lambda(\boldsymbol \mu) }q_0^+ (\nu)$.
The existence of the maximizer $\nu_+ \in \Lambda(\boldsymbol \mu)$ is guaranteed by Proposition \ref{prop:nu}.

Next, we introduce a special class of dependence structures in a way similar to Section \ref{sec:description}.
Fix $\boldsymbol \beta=(\beta_0,\beta_1,\dots,\beta_n)\in \overline{\Delta}_n$ and  $\boldsymbol \mu=(\mu_1,\dots,\mu_n)\in \M^n$. Let the random vector $(X_1^*,\dots,X_n^*)$ satisfy
\begin{equation}
	\label{eq:opt_struc}
	\left\{\begin{array}{l}
		X_i^* =  Z_i \id_{A_i} + W_i \id_{B^c\backslash A_i} + Y_i \id_{B}, ~ i = 1,\cdots,n,
		\\
		\mbox{where $(A_1,\dots,A_n, B)$ is a partition of $\Omega$ independent of all others, and  $\p(A_i)= \beta_i$ for each $i$,}
		\\\mbox{ $Z_i \sim \mu_i^{[1-\beta_i, 1]}$, $W_i \sim \mu_i^{[0,1-\beta_0-\beta_i]}$, $Y_i \sim \mu_i^{[1-\beta_0-\beta_i, 1-\beta_i]}$, and $\sum_{i=1}^n Y_i = R^{+}_{\boldsymbol \beta}(\boldsymbol \mu)$ almost surely}.
	\end{array}\right.
\end{equation}
The existence of $(X_1^*,\dots,X_n^*)$  satisfying \eqref{eq:opt_struc} requires some conditions, which will be clear from Theorem \ref{th:qa-5} below.
The construction in \eqref{eq:opt_struc}  is not unique.
In particular,   the dependence among $(Z_1,\dots,Z_n,W_1,\dots,W_n)$ is not specified.
A specification may be given by
\begin{equation}
	\label{eq:opt_explic}
	\left\{\begin{array}{l}
		%		Z_i= q_{1-\frac{\beta_i}{1-\beta_0}U}^{-}(\mu_i),~
		%		W_i=q_{\frac{1-\beta_0-\beta_i}{1-\beta_0}U}^{-},~A_i=\{U \in [0, 1-\beta_0),K = i\}
		%		\\
		X_i^* = q_{1-\frac{\beta_i}{1-\beta_0}U}^{-}(\mu_i) \id_{\{U \in [0, 1-\beta_0),K = i\}} + q_{\frac{1-\beta_0-\beta_i}{1-\beta_0}U}^{-}(\mu_i) \id_{\{U \in [0, 1-\beta_0),K \neq i\}} + Y_i \id_{\{U \in [1-\beta_0, 1]\}}\\
		\mbox{for each $i =1,\cdots,n$, }
		\mbox{where $U \sim \mathrm U[0,1]$, $\p(K=i) = \frac{\beta_i}{1-\beta_0}$ for  $i=1,\dots,n$,}\\
		\mbox{the random vector $(Y_1,\dots,Y_n)$ is coupled by \eqref{eq:opt_struc}, and $U,K,(Y_1,\dots,Y_n)$ are independent.}
	\end{array}\right.
\end{equation}
%\com{Yang: Here it was a bit unclear indeed. $Y_1, \dots, Y_n$ are not independent. Perhaps this version is better?}

\begin{theorem}\label{th:qa-5}
	Suppose that $\boldsymbol \mu=(\mu_1,\dots,\mu_n)\in \M^n$
	and
	$
	\max_{\nu \in \Lambda(\boldsymbol \mu) }q_0^+ (\nu) =  R^{+}_{\boldsymbol \beta}(\boldsymbol \mu)
	$
	for some $\boldsymbol \beta \in \overline{\Delta}_n$.
	There exists a random vector $(X_1^*,\dots,X_n^*)$ of the form \eqref{eq:opt_struc} attaining the maximum of $q_0^+$ for $\boldsymbol \mu$. Moreover, if $\beta_0 = 1$, then $\boldsymbol{\mu}$ is jointly mixable; if $\beta_0 \neq 1$, $\beta_1,\cdots,\beta_n > 0$ and  the minimum  of each of the functions $h_i: (0,1-\beta_0]\to \R$,
	\begin{equation}\label{eq:h}
		h_i(u) = q_{1-\frac{\beta_i}{1-\beta_0}u}^{-}(\mu_i) + \sum_{j\ne i} q_{\frac{1-\beta_0-\beta_j}{1-\beta_0}u}^{-}(\mu_j), ~~ i = 1, \cdots, n,
	\end{equation}
	is attained at $u=1-\beta_0$, then $(X_1^*,\dots,X_n^*)$ in \eqref{eq:opt_explic} attains the maximum of $q_0^+$ for $\boldsymbol \mu$. %\com{Always use eqref instead of (ref)}
\end{theorem}

Theorem \ref{th:qa-5} gives  useful information on the worse-case dependence structure attaining $\max_{\nu \in \Lambda(\boldsymbol \mu) } q_0^+ (\nu)$ based on our knowledge of minimizer $\boldsymbol{\beta}$. The form \eqref{eq:opt_explic}  gives a specification of all $n$ ``right tail" events: for each $i = 1, \cdots, n$, on the ``$i$-th right tail" event $A_i$, the $(n-1)$ other random variables are all comonotonic, while they are counter-monotonic to the $i$-th individual random variable.    Theorem \ref{th:qa-5} can be applied to arbitrary quantile levels $t$ by considering the conditional distributions  $\mu_1^{t+},\dots,\mu_n^{t+}$.

In the homogeneous case ($\mu_1=\cdots =\mu_n$), the condition for optimality of the dependence structure \eqref{eq:opt_explic} holds for any distribution with a decreasing density if $\beta_0 \neq 1$ and $\beta_1 = \cdots =\beta_n = \frac{1-\beta_0}{n}$. In this case, $h_1 = \cdots = h_n$ on $(0, 1-\beta_0]$. According to Theorem 3.2 and Proposition 3.4 of \cite{BJW14}, $h_1$ is decreasing on $(0, 1-\beta_0]$. Theorem \ref{th:qa-5} then shows that the corresponding measure $\nu_+$ attains the worst-case quantile aggregation. In the heterogeneous case, we give some numerical examples to show the performance of \eqref{eq:opt_explic} in Section \ref{sec:8}.

The dependence structure \eqref{eq:opt_struc} is motivated by the discussions in Section \ref{sec:description} on the setting (DD), and hence it performs well for distributions with approximately decreasing densities. For  the setting (ID), the optimal $\boldsymbol{\beta}$ in \eqref{eq:main1pr} for   $\max_{\nu \in \Lambda(\boldsymbol{\mu})} q_0^+ (\nu)$ is often given by $\beta_0 = 0$, $\beta_i = 1$ for some $i$, $\beta_j = 0$ for other $j \neq i$. Hence, in \eqref{eq:opt_struc}, we have $\p(A_i) = 1$ for some $i$ and $\p(A^c) = \p(A_j) = 0$ for other $j \neq i$. In this case, although Theorem \ref{th:qa-5} holds true, the dependence in \eqref{eq:opt_struc} is completely unspecified. Nevertheless, for  approximately increasing densities, an alternative explicit dependence structure can be similarly designed based on an event $B$ of ``body" and $n$ events $A_1,\dots,A_n$ of ``left tail". We omit the details.

The only unspecified part in  \eqref{eq:opt_explic} is the design of $Y_1,\dots,Y_n$  which add up to a constant. Such random variables are known to exist under some conditions of joint mixability, but they are not easy to explicitly construct  or to simulate except for some very simple cases such as uniform marginal distributions.
Below, we  give an explicit suboptimal  dependence structure as an approximation of  \eqref{eq:opt_explic} without the vector $(Y_1,\dots,Y_n)$:
\begin{equation}\label{eq:opt_notjm}
	X_i^* = \left\{
	\begin{aligned}
		& q_{1-\frac{\beta_i}{1-\beta_0}U}^{-}(\mu_i) \id_{\{K = i\}} + q_{\frac{1-\beta_0-\beta_i}{1-\beta_0}U}^{-}(\mu_i) \id_{\{K \neq i\}}, ~ \text{ if } \beta_0 \neq 1,\\
		& q_{1-\frac{1}{n}U}^{-}(\mu_i) \id_{\{K = i\}} + q_{\frac{n-1}{n}U}^{-}(\mu_i) \id_{\{K \neq i\}}, ~ \text{ if } \beta_0 = 1,
	\end{aligned}
	\right.
	~~~~ i =1,\cdots,n,
\end{equation}
where $U,K$ are given in \eqref{eq:opt_explic} and we further set $\p(K=i)=\frac{1}{n}$, $i=1,\cdots,n$ in case $\beta_0=1$ (i.e., set $\beta_i/(1-\beta_0) = 1/n$). For $(X_1^*,\dots,X_n^*)$  in  \eqref{eq:opt_notjm}, it is easy to see that $X_i^*\sim \mu_i$  for each $i=1,\dots,n$. If $\beta_0 \neq 1$, using $h_i$ in \eqref{eq:h}, the essential infimum of $\sum_{i=1}^n X_i^*$ is given by
\begin{equation}\label{eq:essinf_sum}
	\min_{1 \leq i \leq n} \min_{0 \leq x \leq 1} h_i(x) = \min_{1 \leq i \leq n} \min_{0 \leq x \leq 1}  \left\{q_{1-\frac{\beta_i}{1-\beta_0}x}^{-}(\mu_i) + \sum_{j\ne i} q_{\frac{1-\beta_0-\beta_j}{1-\beta_0}x}^{-}(\mu_j)\right\}.
\end{equation}
Since $X_1^*,\dots,X_n^*$  are obtained by explicit construction,
the above infimum \eqref{eq:essinf_sum} serves as a lower bound for $\sup_{\nu \in \Lambda(\boldsymbol \mu) }q_0^+ (\nu)$. If $\beta_0 \neq 1$,   the first-order condition in the optimality of $\boldsymbol{\beta}$ gives $h_i(1-\beta_0)= R^{+}_{\boldsymbol \beta}(\boldsymbol \mu)$ for $i=1,\cdots,n$ satisfying $\beta_i \neq 0$;   see \eqref{eq:foc} in Appendix \ref{app:proof}.

%\textcolor{blue}{
The dependence structure in \eqref{eq:opt_notjm} has an explicit formula as soon as $\beta_0$ is computed, so it has at most the same computational complexity as computing the convolution bound; see the explanations of the computational issues in the numerical results in Section \ref{sec:8}.
%}

This construction can be further improved as follows. For any $\boldsymbol \beta' \in \Delta_n$, define
$$
H(\boldsymbol \beta') = \min_{1 \leq i \leq n} \min_{0 \leq x \leq 1} h_i(x; \boldsymbol \beta') = \min_{1 \leq i \leq n} \min_{0 \leq x \leq 1}  \left\{q_{1-\frac{\beta_i'}{1-\beta_0'}x}^{-}(\mu_i) + \sum_{j\ne i} q_{\frac{1-\beta_0'-\beta_j'}{1-\beta_0'}x}^{-}(\mu_j)\right\}.
$$
We then solve another $n$-dimensional optimization problem
\begin{equation}\label{eq:problem_another}
	\sup_{\boldsymbol \beta' \in \Delta_n} H(\boldsymbol \beta').
\end{equation}
The maximum point is denoted by $\boldsymbol \gamma=(\gamma_0,\gamma_1,\dots,\gamma_n) \in \overline{\Delta}_n$. %\com{It has to be a maximum point, otherwise does not exist}
Hence, we have the suboptimal dependence structure
\begin{equation}\label{eq:opt_supsub}
	X_i^* = \left\{
	\begin{aligned}
		& q_{1-\frac{\gamma_i}{1-\gamma_0}U}^{-}(\mu_i) \id_{\{K = i\}} + q_{\frac{1-\gamma_0-\gamma_i}{1-\gamma_0}U}^{-}(\mu_i) \id_{\{K \neq i\}}, ~ \text{ if } \gamma_0 \neq 1,\\
		& q_{1-\frac{1}{n}U}^{-}(\mu_i) \id_{\{K = i\}} + q_{\frac{n-1}{n}U}^{-}(\mu_i) \id_{\{K \neq i\}}, ~ \text{ if } \gamma_0 = 1,
	\end{aligned}
	\right.
	~~~~ i =1,\cdots,n,
\end{equation}
where $U,K$ are independent, $U \sim \mathrm{U}[0,1]$ and $\p(K=i)=\frac{\gamma_i}{1-\gamma_0}$, $i=1,\cdots,n$ if $\gamma_0 \neq 1$ and $\p(K=i)=\frac{1}{n}$, $i=1,\cdots,n$ if $\gamma_0 = 1$. For $(X_1^*,\dots,X_n^*)$  in  \eqref{eq:opt_supsub}, it is easy to see that $X_i^*\sim \mu_i$  for each $i=1,\dots,n$ and the essential infimum of $\sum_{i=1}^n X_i^*$ is $H(\boldsymbol \gamma)$.

It turns out that \eqref{eq:opt_notjm} gives a good approximation for the maximum value of $q_0^+$ in many cases and \eqref{eq:opt_supsub} does even better. The numerical performance will be illustrated in Section \ref{sec:8}. Note that
\begin{equation}\label{eq:approx_interval}
	H(\boldsymbol \beta) \leq H(\boldsymbol \gamma)\leq \sup_{\nu \in \Lambda(\boldsymbol \mu) }q_0^+ (\nu)  \leq R^{+}_{\boldsymbol \beta}(\boldsymbol \mu).
\end{equation}
As a result, we provide two-side approximation intervals $[H(\boldsymbol{\beta}), R_{\boldsymbol{\beta}}^+(\boldsymbol{\mu})]$ or $[H(\boldsymbol{\gamma}), R_{\boldsymbol{\beta}}^+(\boldsymbol{\mu})]$ for true value of $\sup_{\nu \in \Lambda(\boldsymbol \mu) }q_0^+ (\nu)$. If only $\boldsymbol{\beta}$ is provided, the former interval can be adopted to approximate the worst-case quantile aggregation; if  it is convenient to conduct another optimization \eqref{eq:problem_another}, the latter one would be more accurate in approximation.

\section{Dual formulation}\label{sec:7}
In this section, we investigate the dual formulation of the quantile aggregation problem. In Theorem \ref{th:qa-4}, %\com{We should mention Theorem 3 instead}
the convolution bound \eqref{eq:main1pr} is obtained by an $n$-dimensional optimization problem.
%\com{changed ``$n$-dim" to ``$n$-dimensional"}
%\begin{equation}
%\label{eq:strong_bound}
%%\inf_{\boldsymbol \beta\in {\Delta} _n} \sum_{i=1}^n R_{[t+(1-t)(\beta-\beta_i),1-(1-t)\beta_i] }(\mu_i) =
%\mbox{(convolution bound) } ~~~~  \inf_{\boldsymbol \beta\in (1-t)\Delta_n}  \sum_{i=1}^n \frac{1}{\beta_0} \int_{1-\beta_i-\beta_0}^{1-\beta_i}q_u^- (\mu_i)\d u.
%%\inf_{\boldsymbol \beta\in {\Delta} _n}  \sum_{i=1}^n \frac{1}{(1-t)(1-\beta)} \int_{t+(1-t)(\beta-\beta_i)}^{1-(1-t)\beta_i}q_u^- (\mu_i)\d u.
%\end{equation}
The main result in this section is that  under continuity conditions  the convolution bound \eqref{eq:main1pr} %\com{changed ``Equation \eqref{eq:main1}" to ``\eqref{eq:main1}"}
is equal to a dual bound \eqref{eq:n_dim_dual_bound}, with a convenient correspondence between the minimizers of both problems.

%\subsection{The dual bound}
The following proposition gives a dual bound  on  quantile aggregation, which is essentially Theorem 4.17 of \cite{R13} that is expressed in terms of probability instead of quantiles. %The bound of the homogeneous case is derived in \cite{EP06}; also see \eqref{eq:dual_expression} below.
%\com{Perhaps should cite Embrechts and Puccetti 06 here, and even Rushendorf 82}
\begin{proposition}\label{prop:n_dual}
	For $t \in [0,1)$,  it holds that
	\begin{equation}\label{eq:n_dim_dual_bound}
		\mbox{\rm [dual bound]} \quad~~~~ \sup_{\nu \in \Lambda(\boldsymbol \mu) }q_t^+ (\nu) \le D_n^{-1}(1-t),
	\end{equation}
	where
	$
	D_n^{-1}(\alpha) = \inf \{ x \in \R: D_n(x) < \alpha \},~ \alpha \in (0,1]
	$
	and the function $D_n: \R \to \R$ is defined by%a function derived from dual formulation
	\begin{equation}\label{eq:D_n}
		D_n(x) = \inf_{\mathbf{r} \in \Delta_n(x)} \left\{ \sum_{i=1}^n \frac{1}{x-r} \int_{r_i}^{x-r+r_i} \mu_i(y, \infty) \d y \right\}, ~ x \in \R,
	\end{equation}
	where
	$\mathbf{r} = (r_1, \cdots, r_n)$, $r = \sum_{i=1}^n r_i$ and
	%\com{the brackets are too small.}
	$\Delta_n(x) = \{(r_1, \cdots, r_n) \in \R^n: \sum_{i=1}^n r_i < x \}$.
\end{proposition}%\com{I have a general comment here: I would not use ``$n$-dimensional dual bound", since the bound is a $1$-dimensional quantity. It is a $1$-dimensional quantity obtained from an $n$-dimensional optimization problem. Make this specific. }

Below we always write $\mathbf{r} =(r_1,\dots,r_n)$ and $r = \sum_{i=1}^n r_i$.  %\com{Use $\mathbf r$ instead of $\boldsymbol r$ for vectors. Only use boldsymbol when followed by a greek like $\boldsymbol \mu$.}
We find that the dual bound \eqref{eq:n_dim_dual_bound} is equal to  our convolution bound \eqref{eq:main1pr} if the marginal distribution and quantile functions are continuous.
\begin{theorem}\label{th:qa-6}
	For fixed $t \in [0,1)$, let $x = D_n^{-1}(1-t)$. Suppose that each of $\mu_1,\dots,\mu_n$ has continuous distribution and quantile functions.
	The convolution bound \eqref{eq:main1pr} and the  dual bound \eqref{eq:n_dim_dual_bound} share the same value $x$.
	Moreover, the correspondence between the minimizers $\boldsymbol \beta \in \overline{\Delta}_n$  of \eqref{eq:main1pr} and $\mathbf{r}$ in the closure of $\Delta_n(x)$ of \eqref{eq:D_n} is given by:
	\begin{equation}\label{eq:corres}
		\mu_i(-\infty,r_i] = 1- \beta_0-\beta_i,~~ \mu_i(-\infty,x-r+r_i] = 1-\beta_i, ~~ i=1,\cdots,n.
	\end{equation}
	%	
	%	Case 2: $\boldsymbol \beta \in \blacktriangle_n$ with $\beta=1$, $\mathbf{r} = (r_1, \cdots, r_n)$ with $r=s$, and for $i=1,\cdots,n$,
	%	\begin{equation}
	%	\mu_i\big((-\infty,r_i]\big) = t + (1-t)(1-\beta_i) = 1 - (1-t)\beta_i,~~ \mu_i\big((-\infty,s-r+r_i]\big) = 1-(1-t)\beta_i.
	%	\end{equation}
	%	
	%	
\end{theorem}

%Theorem \ref{th:qa-6} shows that the convolution bound and the heterogeneous dual bound share  are equal. In addition, if we obtain the value and the optimizer the convolution bound, then we obtain those of the heterogeneous dual bound from the correspondence \eqref{eq:corres} and vice versa. However,

As far as we are aware of, there are no sharpness results on the dual bound in the  setting of heterogeneous marginals.
Therefore, our main results on convolution bounds also contribute to the literature by establishing the sharpness of the dual bounds in several situations, as  the convolution bound and the dual bound are usually equal.
Moreover, we note that the convolution bound is applicable to RVaR aggregation problems, whereas the dual bound based on probability is specific to quantile aggregation.
On the computational side, as the set $(1-t)\Delta_n$ is bounded and the set $\Delta_n(x)$ is unbounded, optimization of the convolution bound \eqref{eq:main1pr} is often easier than that of the dual bound \eqref{eq:n_dim_dual_bound}. Moreover,  \eqref{eq:n_dim_dual_bound} needs to additionally compute an inverse function from $D_n$.  Further, the equivalence in Theorem \ref{th:qa-6} may fail for discrete distributions; see Table \ref{tab:r1-1} of Section \ref{sec:8}. 

In the homogeneous case $\mu_1=\dots=\mu_n$, \cite{EP06} derived a (reduced) dual bound for the worst-case quantile aggregation based on a one-dimensional optimization problem:
\begin{equation}\label{eq:dual_expression}
	\mbox{\rm [reduced dual bound]} \quad~~~~ D^{-1}(1-t) = \inf \left\{ x \in \R: D(x) < 1-t \right\},
\end{equation}
where
$$
D(x) = \inf_{a < \frac{x}{n}} \frac{n}{x-na} \int_{a}^{x-(n-1)a} \mu\left((y, \infty)\right) \d y,~~ x \in \R.
$$
This dual bound is a special case of \eqref{eq:n_dim_dual_bound} by letting $r_1=\cdots=r_n$ in \eqref{eq:D_n}. Thus, the reduced dual bound \eqref{eq:dual_expression} is larger than or equal to the dual bound \eqref{eq:n_dim_dual_bound},
as well as our convolution bound \eqref{eq:main1pr} by Theorem \ref{th:qa-6}. Similarly to the discussion in Section \ref{sec:2}, the dual bound \eqref{eq:n_dim_dual_bound} and the reduced one \eqref{eq:dual_expression} are not generally equal. %\com{Yang: Added a sentence to emphasize.}

Similarly to Theorem \ref{th:qa-6}, one can show that the reduced dual bound \eqref{eq:dual_expression} is the same as the reduced convolution bound \eqref{eq:main1hom} if the marginal distribution and quantile functions are continuous. In Figure \ref{fig:bound_p} (right panel) of Section \ref{sec:8}, we give out examples that \eqref{eq:main1pr} is strictly smaller than \eqref{eq:dual_expression}.

%\com{I think this is not too interesting at this point. It is better to not present the result, but to comment on it. Just say that in the homogeneous case, the Embrechts-Puccetti bound is not as good as the convolution bound (or the heterogeneous dual bound), because of some examples, and because it is equivalent to our simplified convolution bound. Is the proof similar?}

%\com{Rewrite}

\section{Numerical illustration}\label{sec:8}

In this section, the convolution bounds in Theorems \ref{th:qa-4prime}-\ref{th:qa-4} are computed and compared with the existing bounds by numerical examples, including the dual bound of  \cite{EP06}  and the rearrangement algorithm (RA) of \cite{PR12} and \cite{EPR13}. We give some numerical examples to show the performance of the candidate and suboptimal dependence structures \eqref{eq:opt_explic} and \eqref{eq:opt_notjm} in the heterogeneous case. 

We briefly explain the output of RA. 
%\textcolor{blue}{
If a tuple $ \boldsymbol \mu$ of  marginal distributions is given as quantile functions or distribution functions, then RA involves   discretization of the marginal distributions by $N$ steps, where $N$ is chosen as $10^5$ in our implementations. If the marginal distributions are given as empirical distributions of data, then discretization is not needed.
Running RA on  $\boldsymbol \mu$ returns an interval $[\underline{s}_N, \bar{s}_N]$, whose left and right end-points are close when $N$ is sufficiently large, providing an approximation for $\sup_{\nu \in \Lambda(\boldsymbol \mu) }q_t^+(\nu)$.  The left end-point is always a (numerical) lower bound, whereas the right end-point is not a lower or upper bound. Although RA is a popular algorithm, there is no   guarantee that its produced lower bound converges to the true value of $\sup_{\nu \in \Lambda(\boldsymbol \mu) }q_t^+(\nu)$, and there are no theoretical results on the time for RA to converge.
For the above claims and a detailed explanation on implementing  RA, see \cite{EPR13}.
%} 
Consistently with the literature, we treat  the upper value produced by RA as a good approximation of the true value of the worst-case RVaR, although no convergence result is established. 

%\textcolor{blue}{
Next, we explain how we compute the convolution bound, denoted by $B_{\rm conv}$. 
We use the built-in function \texttt{fmincon} in MATLAB to numerically compute $B_{\rm conv}$, where the input are the marginal quantile functions. 
No discretization is needed for this computation, as long as marginal quantiles can be specified.\footnote{Note that many distributions, including empirical distributions from data, have their quantile functions as built-in functions in most computational softwares. For those that do not have a built-in quantile function, computing a numerical quantile function is a standard and simple task.}
All computations are performed on MATLAB R2017b with Intel(R) Core(TM) i5-8250U CPU @ 1.60GHz. 
 In our implementations, we use the default optimization method, the interior-point algorithm. The convergence criterion is a termination tolerance scalar value $10^{-6}$ on the first-order optimality, which is the default choice, and this  convergence criterion is met in all computations. 
 Note, however, that the optimization problem required in computing $B_{\rm conv}$ is generally non-convex, and so global optimality may not always be attainable by \texttt{fmincon}.
 % finds the global optimum (or a local maximum). 
   % Since computing $B_{\rm conv}$ involves an optimization, the global optimality of the numerical value depends on whether   \texttt{fmincon} finds the global optimum (or a local maximum). 
 Nonetheless, in theory,  as the left  end-point $\underline{s}_N$ of RA is a lower bound and $B_{\rm conv}$ is  an upper bound, we have
$$
\mbox{RA output $\underline{s}_N$} \le \mbox{true~} \sup_{\nu \in \Lambda_n(\mu)} q_{t}^+ (\nu) \le \mbox{true~}   B_{\rm conv} \le \mbox{computation of $B_{\rm conv}$}.
$$
As far as we know, there is no theoretical guarantee that the global optimum in the computation of  $B_{\rm conv}$ is attained.  Nevertheless,  in most results,  $\underline{s}_N$ and computed values of  $B_{\rm conv}$  coincide almost perfectly, and therefore convergence is practically verified in these cases.
%}

%\textcolor{blue}{
With discretization, the quantile aggregation problem can also be   formulated into a linear programming (LP) problem with an exponential number of variables. The LP formulation  approximates  to the true optimal value of the problem, but it is difficult if the dimension $n$ or the number of points in the discretization is high. In a real risk management problem where loss distributions are typically continuous (such as asset prices or insurance losses), a fine discretization is required to ensure good approximation, making LP very slow. In Appendix \ref{sec:R2-1}, we provide a detailed comparison among LP, RA and the convolution bound to compute the quantile aggregation problem.
%}

\subsection{Convolution bounds on RVaR aggregation}\label{sec:nume_rvar}
For any $t, s$ with $0 \leq t < t+s \leq 1$, we numerically compute the RVaR aggregation value $R_{t,s} (\nu)$, $\nu \in \Lambda_n(\mu)$ with different methods in the homogeneous case, where the marginal distribution is identical and denoted by $\mu$. The convolution bound is given by \eqref{eq:prime1} and the true value is approximated by RA.%\footnote{Consistent with the literature, we roughly interpreted the   upper value produced by RA as a good approximation of the true value of the worst-case RVaR, although no convergence result is established; see also Table \ref{table6}.}

We fix $t+s=0.9$ and change $s \in (0, 0.9)$ to simulate values of $\sup_{\nu \in \Lambda_n(\mu)} R_{t,s} (\nu)$. In Figure \ref{fig:rvar_inc} (left panel), we check Theorem \ref{th:qa-4prime} that the convolution bound \eqref{eq:prime1} is sharp for marginals with decreasing densities. In Figure \ref{fig:rvar_inc} (right panel), we see that the convolution bound \eqref{eq:prime1} is not sharp for marginals with increasing densities. Although this bound is not sharp for increasing densities, the difference is small and it performs quite well numerically. Moreover, in Figure \ref{fig:rvar_inc}, the convolution bound \eqref{eq:prime1} is sharp if $t = 0$ (Theorem \ref{th:qa-4prime}) and $s \downarrow 0$ (Theorem \ref{th:qa-4}). %\com{RA not true value}
\begin{figure}[t]
	\centering
	\caption{Bounds for $\sup_{\nu \in \Lambda_3(\mu)} R^+_{0.9-s, s}(\nu)$. Left panel: $\mu = $ Pareto$(1,1/2)$ with a decreasing density $\frac{1}{2}x^{-3/2}, ~ x \in [1, \infty)$. %{\color{red}The optimal value is given by the right-hand side of Equation \eqref{eq:main1pr}.} 
		%$\mu  $ has  cdf $ 1-x^{-1/2},~ x \in [1, \infty)$.
		Right panel:  $\mu$ has an increasing density $\frac{5}{9}(101-x)^{-\frac{3}{2}}, \; x\in[1,100]$.
	}
	\includegraphics[width=0.48\textwidth]{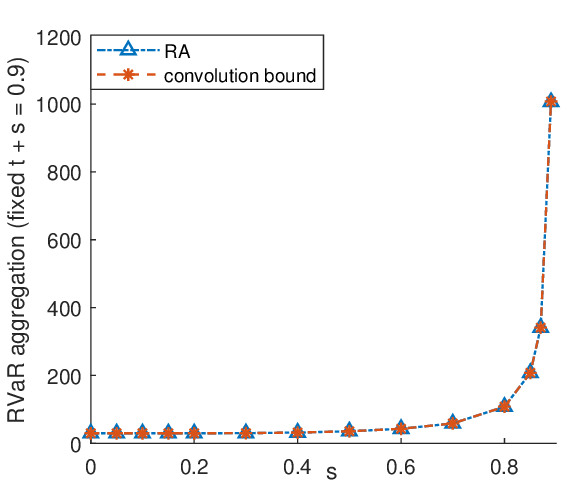}~~
	\includegraphics[width=0.48\textwidth]{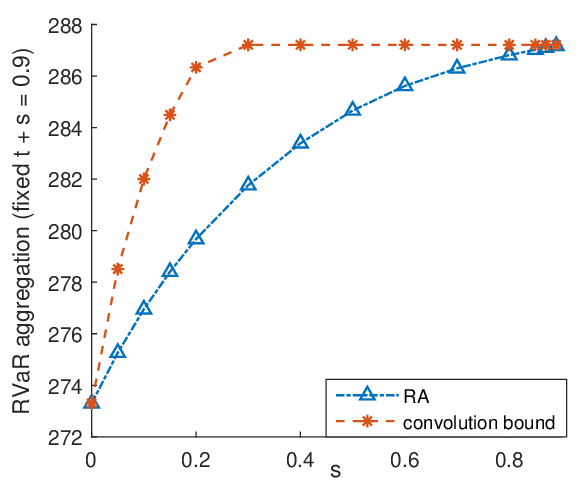}
	\label{fig:rvar_inc}
\end{figure}

%We will look at a straightforward application of our convolution bound to quantitative risk management. In the supervision institute, one needs to report the value-at-risk. We have a sum risk variable
%$$
%S = X_1 + X_2 + X_3,
%$$
%where $X_1$ means the operational risk, $X_2$ means the insurance risk and $X_3$ means the risk. The Basel Committee requires to compute the $\alpha$-quantile of the sum risk (where $\alpha = 0.9997$), we can report the number below.

\subsection{Numerical comparison with existing results}\label{sec:num_var}

For $t \in [0,1)$, we numerically compare the quantile aggregation value $q_t^+ (\nu)$, $\nu \in \Lambda_n(\mu)$ with analytical bounds obtained in the homogeneous case, where the marginal distribution is identical and denoted by $\mu$. Recall that the convolution bound is given by \eqref{eq:main1pr}, the (reduced) dual bound derived in \cite{EP06} is given by \eqref{eq:dual_expression} and the reduced convolution bound is given by \eqref{eq:main1hom}. The standard bound is derived from the lower Fr\'{e}chet-Hoeffding bound (see Remark A.29 of \cite{FS16}). We also give the quantile aggregation value under a comonotonic scenario for comparison. 

\begin{figure}[hbtp]
	\centering
	\caption{Bounds for $\sup_{\nu \in \Lambda_3(\mu) }q_t^+ (\nu)$. Left panel: $\mu = $ Pareto$(1,1/2)$ with a decreasing density $\frac{1}{2}x^{-3/2}, ~ x \in [1, \infty)$.
		Right panel: $\mu$ has an increasing density $\frac{5}{9}(101-x)^{-3/2}, \; x\in[1,100]$.  In the left panel, ``reduced dual bound", ``reduced convolution bound", ``convolution bound" and ``RA" have the same curve, and for better visibility the ``RA" curve is not plotted.  In the right panel, ``RA" and ``convolution bound" have the same curve and ``reduced dual bound" and ``reduced convolution bound" have the same curve.}
	\includegraphics[width=0.48\textwidth]{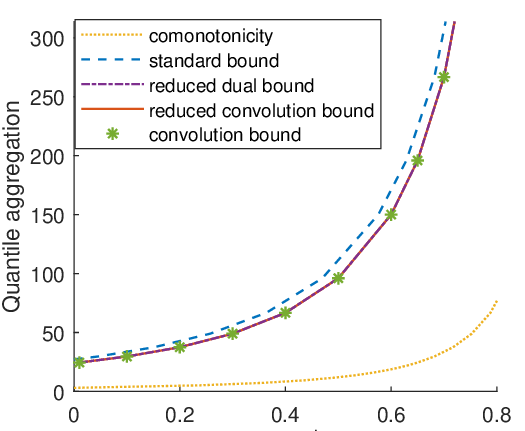}~~
	\includegraphics[width=0.48\textwidth]{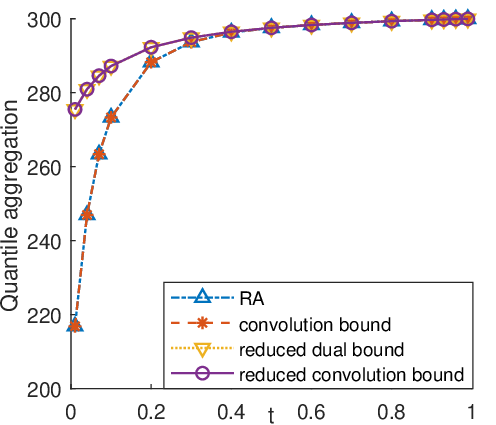}
	\label{fig:bound_p}
\end{figure}

This section also serves as a numerical illustration for the worst-case risk aggregation in Section \ref{sec:qrm}. In Figure \ref{fig:bound_p}, we compute \eqref{eq:intrinsic_prob} in the setting that  losses in a portfolio follow some given marginal distributions. Figure  \ref{fig:bound_p} (left panel) illustrates that the convolution bound \eqref{eq:main1pr}, the reduced convolution bound \eqref{eq:main1hom} and the reduced dual bound \eqref{eq:dual_expression} share the same value of quantile aggregation for a Pareto distribution. The standard bound performs worst as an upper bound for $\sup_{\nu \in \Lambda_n(\mu)} q_t^+ (\nu)$. The comonotonic scenario serves as a lower bound. Results for other distributions such as Lognormal and Gamma distribution are similar and we omit them.

In Figure \ref{fig:bound_p} (right panel),  we plot analytical bounds of the maximum possible quantile aggregation value $\sup_{\nu \in \Lambda_3(\mu) }q_t^+ (\nu)$, where the convolution bound \eqref{eq:main1pr} achieves a strictly smaller value than the reduced dual bound \eqref{eq:dual_expression}. It means that our bound \eqref{eq:main1pr} is an analytically better bound for quantile aggregation. Figure \ref{fig:bound_p} (right panel) further shows that \eqref{eq:main1pr} is better than the reduced convolution bound \eqref{eq:main1hom}; see also Example \ref{ex:ex0}.

%We point out that the examples of distribution are constructed from the idea in Example \ref{ex:ex0} that each of them is not 3-CM. They are transformations of Pareto(1,1/2), with right end-point truncated at 100, which have correspondingly a decreasing or increasing density.

In  Table \ref{table6}, we numerically check the performance of the bound \eqref{eq:main1pr} against RA in more detail. %RA is useful and efficient for heterogeneous or high dimensional cases and for distributions on which analytical results are unavailable. Given a large integer $N$ and a sufficiently small positive number $\epsilon$,
%RA returns an interval $[\underline{s}_N, \bar{s}_N]$, whose left and right end-points are relatively similar values if $N$ is sufficiently large, making it a good approximation for $\sup_{\nu \in \Lambda(\boldsymbol \mu) }q_0^+(\nu)$. Specifically, the left end-point is always a (numerical) lower bound, while the right end-point is close but not always greater than $\sup_{\nu \in \Lambda(\boldsymbol \mu) }q_0^+(\nu)$; see \cite{EPR13}.
\begin{table}[t]\centering
	\captionof{table}{RA (with $N=10^5$) and the convolution bound $B_{\rm conv} $ to compute $\sup_{\nu \in \Lambda(\boldsymbol \mu) }q_0^+(\nu)$ for two different settings of heterogeneous marginal distributions. RA produces an interval whose left-end point is a lower bound, and  the convolution bound $B_{\rm conv} $ is an upper bound which is sharp in the first setting of Pareto distributions. }\label{table6}
	\begin{tabular}{c | c c | c c} 
		& $X_i \sim $ Pareto$(1,\alpha_i)$ & time & $X_i \sim $ Pareto$(1, i+2)$, $i=1,\cdots,20$ & time  \\
		& $\alpha_i = 2+i$,  & $n = 20$  & $X_{20+i} \sim $LogN$(5-i,{( {i}/{2})^2})$, $i=1,\cdots,20$ & $n = 60$ \\
		&$i=1,\dots,20$ & & $X_{40+i} \sim \Gamma(i+1 ,\frac{10}{i})$, , $i=1,\cdots,20$ &  \\
		\hline
		RA & [22.5966, 22.5971] & 111s & [539.5141, 539.6205] & 639s \\
		$B_{\rm conv} $ \eqref{eq:main1pr} & 22.5968 & 46s & 539.5611 & 672s \\
		RA minus \eqref{eq:main1pr} & [$-2.04*10^{-4}$, $3.02*10^{-4}$] & & [-0.0470, 0.0594] &  \\
		\hline
	\end{tabular}
\end{table}

Concerning performance, Figure \ref{fig:bound_p} (right panel) and Table \ref{table6} both indicate that the convolution bound and RA have a similar value for most cases. We discuss three aspects. First, we emphasize again that the true value of $\sup_{\nu \in \Lambda(\boldsymbol \mu) }q_0^+(\nu)$ is generally unavailable. It is available in cases with monotone densities, where the true value equals to the convolution bound according to Theorem \ref{th:qa-4}. It is the case of Figure \ref{fig:bound_p} and  the first model in Table \ref{table6}.   Second,  if the true value of $\sup_{\nu \in \Lambda(\boldsymbol \mu) }q_0^+(\nu)$ is   unknown, then we can use the (upper) convolution bound together with the lower bound provided by the RA to approximately target the true value. As shown in the second model of in Table \ref{table6}, the difference between the two bounds is quite small and we can approximately know the true value. Third, we show that in some cases RA does not perform well while the convolution bound provides a sharp result; see Example \ref{ex:RA_not_sharp}.

\begin{example}\label{ex:RA_not_sharp}
	Let $\mu$ be a triatomic uniform distribution on $\{1,2,3\}$.
	By constructing a random vector uniformly distributed on $\{(1,2,3),(2,3,1), (3,1,2)\}$, we get $\sup_{\nu \in \Lambda( \mu,\mu,\mu) } q_0^+ (\nu) =6$. As a result, \eqref{eq:main1pr} provides a sharp upper bound   $\inf_{\boldsymbol{\beta} \in \Delta_3} R_{\boldsymbol{\beta}}^+ (\mu,\mu,\mu) = 6$ with the optimal $\boldsymbol{\beta}=(1,0,0,0)$. However, the interval provided by the RA is $[5, 5]$; see Example \ref{ex:RA-fails} for details.
\end{example}

Concerning computation time, we find that the convolution bound \eqref{eq:main1pr} is computed quicker than or similarly to RA. %\footnote{The computation is performed on MATLAB R2017b with Intel(R) Core(TM) i5-8250U CPU @ 1.60GHz.}
 In conclusion, \eqref{eq:main1pr} is not only a   good analytical upper bound, but also performs quickly in the numerical calculation for the maximum possible lower end-point $\sup_{\nu \in \Lambda(\boldsymbol \mu) }q_0^+(\nu)$.
 
 Theorem \ref{th:qa-6} assumes continuous distribution   and quantile functions. 
 Generally, equivalence between quantile methods and probability methods can be troublesome when dealing with discrete distributions. 
We illustrate in a simple example that the equivalence in Theorem \ref{th:qa-6} may fail. 
Define $\mu = \text{Bernoulli}(0.5)$ and $n = 3$. We numerically compute the values of RA, the convolution bound and the dual bound on $\sup_{\nu \in \Lambda_3(\mu) }q_t^+ (\nu)$  in Table \ref{tab:r1-1}. The true values are available in this simple setting for comparison. We make two observations. First,  as we mentioned above, the true value of $\sup_{\nu \in \Lambda_3(\mu) }q_t^+ (\nu)$ is bounded from below by the RA left end-point and from above by the convolution bound and is exactly obtained if these two bounds are equal. Second, for $t = 0.3$ and $0.4$, the values of the convolution bound  $B_{\text{conv}}$ and the dual bound $B_{\text{dual}}$ are not the same. In these cases, we observe that $B_{\text{conv}}$ is closer (equal) to the true value  than $B_{\text{dual}}$. 
\begin{table}[t]\centering
	\captionof{table}{RA (with $N = 10^5$), true value,  convolution bound $B_{\text{conv}}$ and dual bound $B_{\text{dual}}$ on $\sup_{\nu \in \Lambda_3(\mu) }q_t^+ (\nu)$.}
	\begin{tabular}{c | c | c | c | c | c | c | c} 
		$t  $ & 0.1 & 0.2 & 0.3 & 0.4 & 0.5 & 0.6 & 0.7\\
		\hline
		RA (left end-point) &  1  &   1  &   2  &   2  &   $\textbf{2}$  &   3  &   3 \\
		\hline
		true value &  1  &   1  &   \textbf{2}  &   \textbf{2}  &   \textbf{3}  &   3  &   3 \\
		\hline
		$B_{\text{conv}}$ &  1.6667  &  1.875 &  $\textbf{2}$  & $\textbf{2}$ &    $\textbf{3}$  &  3  &  3 \\
		\hline
		$B_{\text{dual}}$  &  1.6667  &  1.875  &  $\textbf{2.1429}$ &  $\textbf{2.5}$  &  3  &  3  &  3  \\
		\hline
	\end{tabular}
	\label{tab:r1-1}
\end{table}

\subsection{Performance of extremal dependence structures} \label{sec:performance}

Recall that in Section \ref{sec:6} we propose a candidate dependence structure \eqref{eq:opt_explic} for the worst-case quantile aggregation. We also state a suboptimal structure \eqref{eq:opt_notjm} without involving $Y_1, \cdots, Y_n$. A better suboptimum \eqref{eq:opt_supsub} is obtained by solving another optimization problem and a two-side approximation interval $[H(\boldsymbol{\gamma}), R_{\boldsymbol{\beta}}^+(\boldsymbol{\mu})]$ is established. We now give some numerical examples to compare their corresponding lower end-points in the heterogeneous case with $n=3$. As shown in Theorem \ref{th:qa-5}, possible values of the aggregation variable in \eqref{eq:opt_explic} are those of the functions $h_1, h_2, h_3$ on $[0,1-\beta_0]$, while the corresponding values in \eqref{eq:opt_notjm} are those of $h_1, h_2, h_3$ on $[0,1]$. Thus, the lower end-point derived from \eqref{eq:opt_explic} is attained at the minimal values of all $h_1, h_2, h_3$ on $[0,1-\beta_0]$, while that from \eqref{eq:opt_notjm} is attained at those on $[0,1]$.
\begin{figure}[t]
	\centering
	\caption{Performance of extremal dependence structures with settings in Table \ref{table:subs1}.  In each panel, we plot the function of $h_i$ and the values of    quantile aggregation  provided by the convolution bound, the suboptimum $\boldsymbol{\beta}$ and the suboptimum $\boldsymbol{\gamma}$. }
	\label{fig:h}
	\includegraphics[width=0.9\textwidth]{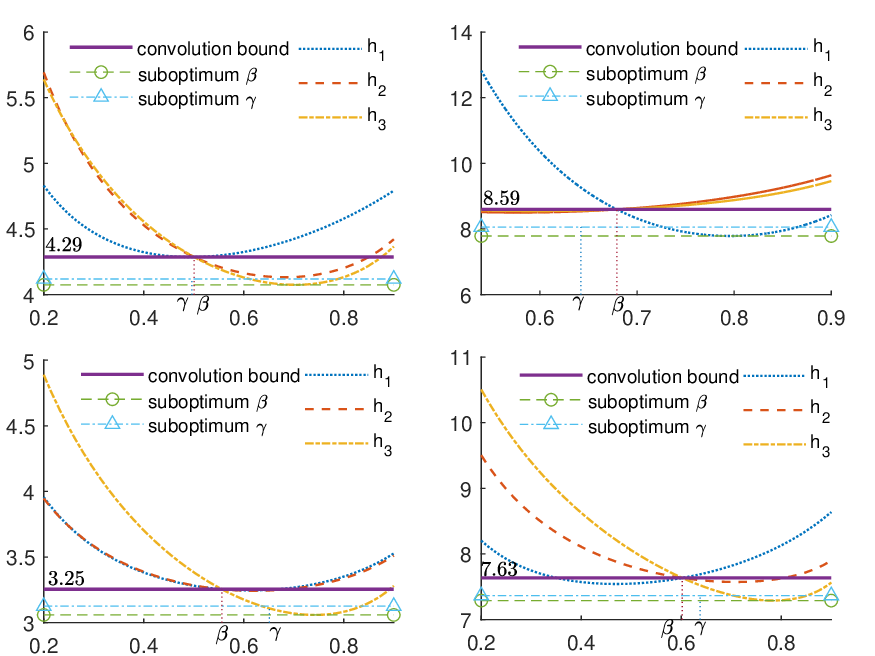}
\end{figure}

\begin{table}[t]\centering
	\small
	\captionof{table}{Numerical values of lower end-points in Figure \ref{fig:h}.}\label{table:subs1}
	\begin{tabular}{c | c | c | c | c} 
		& $\mu_1=$ Pareto(1,3) & $\mu_1=$ Pareto(1,1/3) & $\mu_1=$ Pareto(1,3) & $\mu_1=$ Pareto(1,3)    \\
		& $\mu_2=$ LogN(0,1)  & $\mu_2=$ LogN(0,1) & $\mu_2=$ LogN(-1,1) & $\mu_2=$ LogN(0,1)  \\
		& $\mu_3=$ $\Gamma(1,2)$ & $\mu_3=$ $\Gamma(1,2)$ & $\mu_3=$ $\Gamma(1,2)$ & $\mu_3=$ $\Gamma(3,2)$ \\
		\hline
		Mean & 5.1487 & $\infty$ & 4.1065 & 9.1487 \\
		RA & [4.2856,4.2857] & [8.5933,8.5936] & [3.2545,3.2545] & [7.6338,7.6341] \\
		$[H(\boldsymbol{\gamma}), R_{\boldsymbol{\beta}}^+(\boldsymbol{\mu})]$ & [4.1185, 4.2857] & [8.055, 8.5936] & [3.1254,3.2545] & [7.3653,7.634]\\
		\hline
		Bound \eqref{eq:main1pr} & 4.2857 & 8.5936 & 3.2545 & 7.634 \\
		Candidate \eqref{eq:opt_explic} & 4.2855 & 8.4995 & 3.2545 & 7.5415 \\
		Suboptimum \eqref{eq:opt_notjm} & 4.0739 & 7.7835 & 3.0587 & 7.2889 \\
		Suboptimum \eqref{eq:opt_supsub} & 4.1185 & 8.055 & 3.1254 & 7.3653 \\
		%		Approx Algo &  & &  &  \\
		\hline
	\end{tabular}
\end{table}

In Figure \ref{fig:h}, according to the sufficient condition in Theorem \ref{th:qa-5}, the third subfigure shows that \eqref{eq:opt_explic} gives out the worst-case quantile aggregation. Even in the other subfigures, the essential infimum of $\sum_{i=1}^n X_i^*$ of \eqref{eq:opt_explic}, which is the minimal value of $h_1$ on $[0,1-\beta_0]$, is just slightly lower than the corresponding $q_0^{+}(\nu_+)$.
%It means that \eqref{eq:opt_explic} provides a good approximation to the worse-case aggregation measure $\nu_+$.
We further show the numerical values in Table \ref{table:subs1},
including the convolution bound, the RA results, and values from the suboptimal structures \eqref{eq:opt_notjm} and \eqref{eq:opt_supsub}.
Recall that \eqref{eq:opt_notjm} is based on $\boldsymbol{\beta}$, while \eqref{eq:opt_supsub} requires solving $\boldsymbol{\gamma}$ in another optimization problem.
The suboptimal methods  give explicit random vectors, and hence it is useful in visualizing the worst case of quantile aggregation. In Table \ref{table:subs1}, both \eqref{eq:opt_notjm} and \eqref{eq:opt_supsub} produce numbers close to the convolution bound in many cases, while \eqref{eq:opt_supsub} is always better but requires more computation. 

\section{Two applications}\label{sec:OR}

%In Section \ref{sec:OR}, we revisit the two running examples in Section \ref{sec:ex} and use our convolution bound to illustrate the application.% to illustrate their usefulness.

\subsection{Robust risk management}	

	The convolution bounds can be directly applied to compute the worst-case or best-case risk aggregation  problem   in  Section \ref{sec:qrm} for the risk measure being VaR  or RVaR. 
	In this section, we solve the robust portfolio selection problem   \eqref{eq:r1-robust} also presented in   Section \ref{sec:qrm}. 
	Using a standard Lagrangian technique, for \eqref{eq:r1-robust},   it is equivalent to solve 
	\begin{align}
			\mbox{maximize} ~~~&  \inf_{\mathbf{X} \sim \boldsymbol{\mu}} \E[u(\boldsymbol{\lambda} \cdot \mathbf{X})] - \xi    \sup_{\mathbf{X} \sim \boldsymbol{\mu}} q_t^+(\boldsymbol{\lambda} \cdot (-\mathbf{X})) ~~~\mbox{over $\boldsymbol{\lambda} \in \overline{\Delta}_{n-1}$} ,
			\label{eq:r1-2}
	\end{align}
	where $\xi\ge 0$ is a Lagrangian multiplier. After    obtaining the optimizer $\boldsymbol{\lambda}^*_\xi$ of   \eqref{eq:r1-2} for varying $\xi\ge 0$  we can calibrate $\xi$ with the risk   constraint $\sup_{\mathbf{X} \sim \boldsymbol{\mu}} q_t^+(\boldsymbol{\lambda}^*_\xi \cdot  (-\mathbf{X})) = x$   to solve \eqref{eq:r1-robust} whenever the risk   constraint is binding. In what follows, we will focus on \eqref{eq:r1-2}.
	Recall that $u$ is a strictly concave and increasing function. 
	For $\boldsymbol \lambda =(\lambda_1,\dots,\lambda_n)$, 
	 	the first term in  the problem  \eqref{eq:r1-2} admits a simple formula
		\begin{align} \label{eq:r1-3} 
	\inf_{\mathbf{X} \sim \boldsymbol{\mu}} \E[u(\boldsymbol{\lambda} \cdot \mathbf{X})] = \int_0^1 u\left( \sum_{i=1}^n \lambda_i q_v^+(\mu_i) \right) \d v =: U(\boldsymbol \lambda),
		\end{align}
	because the worst-case portfolio is comonotonic (e.g., Corollary 3.29 of \cite{R13}).
	Since $\boldsymbol \lambda \mapsto u( \sum_{i=1}^n \lambda_i q_v^+(\mu_i) ) $ is concave,  so is $U$. The concavity of $U$ implies that  a maximizer for \eqref{eq:r1-3} may favour some diversification. 
	On the other hand, as shown by Proposition 7.1 of \cite{CLLW22}, a minimizer for  $\sup_{\mathbf{X} \sim \boldsymbol{\mu}} q_t^+(\boldsymbol{\lambda} \cdot (-\mathbf{X}))$  favours no diversification  if the marginal distributions are identical and satisfy  (DD) or (ID). 
	Therefore, intuitively, there is a trade-off between diversification and concentration in \eqref{eq:r1-2}. 
	Applying the convolution bound  in the form of Theorem \ref{th:qa-2} (the symmetric version of Theorem \ref{th:qa-4}) and using the positive homogeneity of RVaR, we have 
	\begin{align*}
	   \inf_{\mathbf{X} \sim \boldsymbol{\mu}} \E[u(\boldsymbol{\lambda} \cdot \mathbf{X})] - \xi  \sup_{\mathbf{X} \sim \boldsymbol{\mu}}  q_t^+(\boldsymbol{\lambda} \cdot (-\mathbf{X}))  
	&=  U(\boldsymbol \lambda) + \xi   \inf_{\mathbf{X} \sim \boldsymbol{\mu}} q_{1-t}^-(\boldsymbol{\lambda} \cdot \mathbf{X})  \\
	&\ge  
	 U(\boldsymbol \lambda)  + \xi   \sup_{\boldsymbol \beta\in (1-t)\Delta_n}  \sum_{i=1}^n R_{1-\beta_0-\beta_i,\beta_0 }(\lambda_i X_i) \\ 
	&=   \sup_{\boldsymbol \beta\in (1-t)\Delta_n} \left\{ U(\boldsymbol \lambda)+ \xi   \sum_{i=1}^n \lambda_i R_{1-\beta_0-\beta_i,\beta_0 }(\mu_i) \right\} .
\end{align*}
The above objective will be maximized over $\boldsymbol \lambda$. 
For a fixed $\boldsymbol \beta\in (1-t)\Delta_n$, the above objective is concave in $\boldsymbol \lambda$, which is easy to maximize. As we discussed in Section \ref{sec:technical}, the optimization of $ \boldsymbol \beta$ is also often simple.
The inequality above becomes an equality when the convolution bound is sharp. 
This is guaranteed if $\mu_1,\dots,\mu_n$ have increasing (or decreasing) tail densities below level $1-t$ (Theorem \ref{th:qa-2}).
Since $t$ is close to $1$, this requirement is very weak and it is satisfied by portfolio models in practice. 
  We will assume that the convolution bound is sharp from now on.

%	The computation difficulty of the sup problem in the last equality is at the same level as that of the convolution bound, which is shown to be accurate and fast. If the marginals are all Pareto, i.e., $q_v^+(\mu_i) = a_i (1-v)^{-\frac{1}{\alpha_i}}$, $v  \in (0, 1)$, $i  = 1, 2, 3$ and $u(x) = \frac{1}{p} x^p$, $x \in (0, \infty)$, we have
%	$$
%	\begin{aligned}
%	& \quad \int_0^1 u\left( \sum_{i=1}^n \lambda_i q_v^+(\mu_i) \right) \d v  + \xi \cdot \sum_{i=1}^n \lambda_i R_{1-\beta_0-\beta_i,\beta_0 }(\mu_i) \\
%	&= 
%	\int_0^1 \frac{1}{p} \left( \lambda_1 a_1 (1-v)^{-\frac{1}{\alpha_1}} + \lambda_2 a_2 (1-v)^{-\frac{1}{\alpha_2}} + \lambda_3 a_3 (1-v)^{-\frac{1}{\alpha_3}} \right)^p \d v \\
%	&\quad + \frac{\xi}{\beta_0} \cdot \Bigg( \frac{-\lambda_1 a_1}{1-\frac{1}{\alpha_1}} \left( (1-\beta_1-\beta_0)^{1-\frac{1}{\alpha_1}} - (1-\beta_1)^{1-\frac{1}{\alpha_1}} \right) \\
%	&\qquad\qquad+  \frac{-\lambda_2 a_2}{1-\frac{1}{\alpha_2}} \left( (1-\beta_2-\beta_0)^{1-\frac{1}{\alpha_2}} - (1-\beta_2)^{1-\frac{1}{\alpha_2}} \right)\\ 
%	&\qquad\qquad+  \frac{-\lambda_3 a_3}{1-\frac{1}{\alpha_3}} \left( (1-\beta_3-\beta_0)^{1-\frac{1}{\alpha_3}} - (1-\beta_3)^{1-\frac{1}{\alpha_3}} \right)  \Bigg).
%	\end{aligned}
%	$$
%	which is %linear in $\lambda_i$, $i = 1, 2, 3$ and 
%	convex in $\beta_i$, $i = 1, 2, 3$; see the same analysis in Section \ref{sec:technical}. Hence, the computation essentially reduces to a one-dimensional search on $\beta_0$, $\lambda_i$, $i = 1, 2, 3$.
	
For a simple illustration, we consider three heterogeneous assets with normal and log-normal distributions and different parameters. The parameters of these distributions are chosen such that the problem is non-trivial in the sense that the three marginal distributions do not dominate each other.
 We take an exponential utility function to characterize the preference of the decision maker. The numerical results of the following setting are given in Table \ref{table:QRM}. 	
%	\begin{table}[htp]
%		\color{red}
%		\centering
%		\begin{tabular}{c|c|c|c|c|c}
%			Lagrangian $\xi$ & $0$ & $0.001$  & $0.01$ & $0.1$ & $1$ \\
%			\hline
%			portfolio $\lambda^*$ & $(0.71, 0.14, 0.15)$ & $(0.65, 0.19, 0.15)$ & $(0.14, 0.69, 0.17)$ & $(0, 0.87, 0.13)$ & $(0, 0.97, 0.03)$ 
%		\end{tabular}
%		\caption{Different optimal portfolios under different coefficients of risk aversion. We use the following setting: %of \cite{AP14}: 
%			$n = 3$, $X_1 \sim \text{N}(4.5, 1)$, $X_2 \sim \text{N}(4, 0.01)$, $X_3 \sim \text{LogN}(1, 3)$, $t = 0.99$, $U(x) = -5 \exp(-x/5)$, $x \in \R$.}
%		\label{table:QRM}
%	\end{table}
	\begin{table}[t] 
		\centering
		\caption{Different optimal portfolios under different risk   constraints, where
			$n = 3$, $t = 0.99$, $X_1 \sim \text{N}(0.9, 1.8^2)$, $X_2 \sim \text{N}(0.1, 0.2^2)$, $X_3=Y-4$ with $Y \sim \text{LogN}(1.5, 1.2^2)$,  and $u(x) = 5(1- \exp(-x/5))$, $x \in \R$}
		\label{table:QRM}
		\def\arraystretch{1.5}
		\begin{tabular}{c|c|c|c|c|c}
			Lagrangian $\xi$ & $0$ & $0.01$  & $0.1$ & $1$ & $10$ \\
			\hline	
	%		Risk constraint $x$ & $\ge 3.0745$ &  2.4885     & 1.6258    & 0.4588  & 0.3653          \\
	%		\hline 	
					risk constraint $x$ & $\ge 3.075$ &  2.489     & 1.626    & 0.459   & 0.365          \\
			\hline  
			 
			optimal portfolio $\boldsymbol \lambda^*$ & $(0.41, 0.23, 0.36)$ & $(0.23, 0.38, 0.39)$ &  $(0, 0.64, 0.36)$ & $(0, 0.97, 0.03) $ & $(0, 1, 0)$  \\\hline
			  utility $U(\boldsymbol \lambda^*)$ & 0.704	 &  0.702    & 0.678        & 0.200  &   0.095 
			    \\
			\hline 
		\end{tabular}
	\end{table}
	If the Lagrangian multiplier $\xi = 0$, then the problem is robust utility maximization under uncertainty  without risk constraint (or, the risk constraint is not binding), and the optimizer is a  diversified portfolio $(0.41, 0.23, 0.36)$. As the Lagrangian multiplier $\xi$ increases, the role of the risk constraint is getting more important, and the optimal portfolio becomes more concentrated. In case $\xi = 10$, the optimal portfolio is to only invest in the second asset, which has the smallest expected return and the smallest variance (thus, the safest choice). This is consistent to our intuition that when the penalty on the worst-case dependence is large, the decision maker prefers a concentrated portfolio, which does not have dependence uncertainty. A similar phenomenon is also observed by  \cite{PP18} and \cite{CLLW22} in different settings without the utility term $U(\boldsymbol \lambda)$.   

\subsection{The O-ring model}

We proceed to analyze the O-ring model presented in Section \ref{sec:o-ring}. Our goal is to find
the minimum value of  \eqref{prob:o-ring} as well as its optimizing dependence structure. This optimizing dependence structure will yield matching patterns in the label market. 
For this, we   use Proposition \ref{th:multiply}  and the arguments in Section \ref{sec:description} on the extremal dependence.

For an illustration, we will use the following simple setting: $\mu_Z=\mathrm{Beta}(5/6, 1)$, $\mu_1=\mu_2= \mathrm{Beta}(5/4, 1)$. That is, there are two workers in each firm, where the product value of the firm and the successful probabilities of the workers follow the Beta distributions.
Note that all Beta distributions of the form $\mathrm{Beta}(\alpha, 1)$ satisfy the condition in Proposition  \ref{th:multiply} (iii); equivalently, (ID) holds for the distributions of $\log (Z)$, $\log(X_1)$ and $\log(X_2)$.
Hence,
the right-hand side of \eqref{eq:main1pr:multiply} is the true minimum value in \eqref{prob:o-ring}.
For values of the threshold $y \in [0, 1)$, we obtain the corresponding minimum probability $t$ as well as the optimal $\boldsymbol \beta/(1-t)$ of \eqref{eq:main1pr:multiply} in Table \ref{table:Oring}.
We will explain this table in more detail below.

\begin{table}[t]
	\centering
	\caption{Different cases of the optimal matching}
	\label{table:Oring}
		\def\arraystretch{1.5}
	\begin{tabular}{c|c|c|c}
		threshold $y$   &  $(0, 0.046)$ & $[0.046, 0.353)$  & $[0.353, 1)$  \\
		\hline
		corresponding probability $t$ & $(0, 0.077)$ & $[0.077, 0.420)$ & $[0.420, 1)$ \\
		\hline
		optimal $\boldsymbol{\beta}/(1-t)$ of \eqref{eq:main1pr:multiply} &  $(0, 1, 0, 0)$ & $(0, 1, 0, 0)$ & $(1, 0, 0, 0)$ \\
		\hline
		optimal $\boldsymbol{\beta}$  of \eqref{eq:lower} &  $(<1,>0,>0,>0)$ & $( <1, >0, 0, 0)$ & $(1, 0, 0, 0)$\\
		\hline
		possible events & $C$, $B$, $A_Z$, $A_X$ & $C$, $B$, $A_Z$ & $C$, $B$ \\\hline
	\end{tabular}
\end{table}

For a given threshold $y \in (0,1)$ and its corresponding probability level $t \in (0, 1)$,
by Proposition \ref{prop:basic},
one needs to  consider an optimal matching of the conditional distributions $\mu_Z^{t+}$, $\mu_1^{t+}$ and $\mu_2^{t+}$ on an event with total probability $1-t$,
and the matching on the remaining event $C$ with probability $t$ can be arranged arbitrarily.
As $\mu_1=\mu_2$, by symmetry, the overall optimal dependence structure
includes four possible events ($\Omega = A_Z \cup A_X \cup B\cup C$):
\begin{enumerate}[$(A_Z)$]%[(i)]
	\item[$(B)$] two medium-skilled workers work together as a team in a medium-value  firm;
	\item[$(A_Z)$] a low-value  firm hires two high-skilled workers as a team;
	\item[$(A_X)$] a low-skilled worker works with a high-skilled coworker in a high-value  firm;
	\item[$(C)$] two very low-skilled workers work  in a very low-value  firm.
\end{enumerate}	
In the above construction of $(Z,X_1,X_2)$, we have
$$
B\cup A_Z \cup A_X = \{Z\ge q^+_t (\mu_Z)\}=\{X_1\ge q^+_t (\mu_1)\}= \{X_2\ge q^+_t (\mu_2)\};$$ 
$$
C=  \{Z < q^+_t (\mu_Z)\}=\{X_1< q^+_t (\mu_1)\}= \{X_2< q^+_t (\mu_2)\}.
$$
Such a structure is called $t$-concentration by \cite{WZ21}.

For a given threshold $y \in (0,1)$ and its corresponding probability level $t \in (0, 1)$,
by the arguments in Section \ref{sec:description}, since  (ID) holds,  to determine the possible events in the dependence structure, one should compute the optimal $\boldsymbol{\beta}$ from the lower convolution bound \eqref{eq:main2pr:multiply}, i.e.,
\begin{equation}\label{eq:lower}
	\sup_{\boldsymbol \beta\in \Delta_n} \sum_{i = Z, 1, 2} R_{1-\beta_i-\beta_0,\beta_0 } \left(  \mu_i^{t+} \circ \exp \right).
\end{equation}
We denote the optimal $\boldsymbol{\beta}$  in Table \ref{table:Oring} by $(1 - \beta_Z - \beta_1 - \beta_2, \beta_Z, \beta_1, \beta_2)$. As  in Section \ref{sec:description}, we have the following classification:
\begin{enumerate}[Case 1.]
	\item[Case 1.] If $y \in [0.353, 1)$, then $\beta_Z = \beta_1 = \beta_2 = 0$, implying that the events $C$ and $B$ occur.
	
	\item[Case 2.] If $y \in [0.046, 0.353)$,  then  $\beta_Z > 0$ and $\beta_1 = \beta_2 = 0$,  implying that the events $C$, $B$ and $A_Z$ occur.
	
	\item[Case 3.] If $y \in (0, 0.046)$,  then $\beta_Z, \beta_1, \beta_2 > 0$, implying that all the events $C$, $B$, $A_Z$ and $A_X$ occur.	
\end{enumerate}

The event $C$ corresponds to the proportion of firms and workers that are given up by the matching problem. %These firms and workers have too low ranks in their relative groups and it is not optimal to match them with highly ranked partners to hope to pass the threshold $y$.
Since our goal is to obtain as many project values above $y$ as possible, some projects have to be left behind, and they are composed of low-value firms and low-skilled workers. 
The event $B$ corresponds to the proportion of medium-value firms which hire medium-skilled workers. This reflects the majority of firms and workers and they are matched together. The event $A_Z$ means that the low-value firm has to hire high-skilled workers to minimize  the global deficiency proportion of production.

The event $A_X$  matches a  high-value  firm and a high-skilled worker with a low-skilled coworker. If the firm value is high enough and the threshold $y$ is low enough, then there is no point for this firm to hire  two   high-skilled workers anymore; in fact, the firm can hire one high-skilled worker and reduce its cost by hiring a low-skilled coworker  if the goal is only to bypass the threshold $y = 0.046$. This may be realistic in settings where robots or automated machines are cheaper and less effective than human workers, but  they are sufficient to pass a threshold of interest (e.g., quality control) for the firm, so the firm would use robots or automated machines. However, this situation does not happen if the threshold is high enough.

The optimal matching, featured with events $A_Z$ and $A_X$, for   problem \eqref{prob:o-ring} is quite different from the classic result in \cite{K93}, where high-skilled workers are always matched with high-value firms. Certainly, the objectives in the two settings are different.
To explain this from the perspective of dependence, the product function $(z,x_1,\dots,x_n)\mapsto z  \prod_{i=1}^n x_i $ is a supermodular function, and its expected value is maximized by positive matching, that is, comonotonicity; see e.g., \citet[Section 2]{PW15}. On the other hand, $(z, x_1, \dots, x_n)\mapsto \id_{\{z  \prod_{i=1}^n x_i \le y \}}$
is neither supermodular nor submodular, and its minimization (or maximization) is highly complicated and involves both positive and negative matching; see \citet[Section 3]{PW15} for extremal negative dependence. Translating this into the O-ring theory, to minimize the percentage of production values under a threshold,  one needs to assign  high-skilled workers and  high-value firms to assist  less-performed workers or firms. Such a matching policy is quite common in socially relevant real situations, e.g., team tournaments, help groups, and financial assignments, to name a few. The appearance of negative matching is getting increasing attention in various economic contexts; see e.g., the recent work of \cite{BTZ21, BTWZ23}.

\section{Conclusion}\label{sec:9}
Using the RVaR convolution result  of \cite{ELW18}, we establish new (semi-analytical) bounds for the problem of quantile aggregation, and show that these bounds are sharp in many cases with analytical formulas in the literature. We  can interpret the corresponding worst-case dependence structure and give  explicit  construction for the complicated optimization problem.
The convolution bounds cover all existing theoretical results on quantile aggregation. Moreover,  the proposed bound has advantages in its tractability, interpretability, and computation.
% such as estimating the extremal quantile/risk aggregation,  robust statistical hypothesis testing, visualization and realization of the corresponding dependence structure, simulation calibration, and solution of scheduling problems.

The level of theoretical difficulty in quantile aggregation leaves ample room for future adventures and challenges. For instance, the sharpness of convolution bounds under general conditions, other than those in Theorems \ref{th:qa-4prime},  \ref{th:qa-4}, \ref{th:qa-4primeprime} and \ref{th:qa-2},  is an open question.
For the interested reader, we connect our results to the theory of joint mixability in Appendix \ref{app:3}, where many questions remain to be open.
Additional information on the dependence structure, other than the marginal distributions, can be incorporated in the quantile aggregation problem, and it usually leads to highly challenging questions; see e.g., \cite{BRV17, BRVW17} and \cite{BKLP22}.
In view of the broad appearance of quantile aggregation,
its application domain includes many problems in economics, finance, risk management, statistics, and scheduling, in addition to the two applications discussed Section \ref{sec:ex}.
We mention some applications in Appendix \ref{app:applications}, on which many relevant questions warrant thorough future investigation.

%\vspace{-4mm}

\appendix
\setcounter{lemma}{0}
\renewcommand{\thelemma}{A.\arabic{lemma}}
\setcounter{proposition}{0}
\renewcommand{\theproposition}{A.\arabic{proposition}}
\setcounter{theorem}{0}
\renewcommand{\thetheorem}{A.\arabic{theorem}}
\setcounter{definition}{0}
\renewcommand{\thedefinition}{A.\arabic{definition}}
\setcounter{corollary}{0}
\renewcommand{\thecorollary}{A.\arabic{corollary}}
\setcounter{example}{0}
\renewcommand{\theexample}{A.\arabic{example}}
\setcounter{equation}{0}
\renewcommand{\theequation}{A.\arabic{equation}}

\section{Lower convolution bounds}\label{sec:lower}

In this appendix, we quickly collect results on lower convolution bounds for $\inf_{\nu \in \Lambda(\boldsymbol \mu) } R_{t,s} (\nu)$
and $\inf_{\nu \in \Lambda(\boldsymbol \mu) } q^-_t (\nu)$,
and some related results.
The proofs of these results are symmetric to those on the upper convolution bounds, and they are omitted.

%\subsection{Symmetric results of lower convolution bounds}
%The derivation of lower convolution bounds is based on symmetry with the upper ones. For any random variables $X_1,\cdots,X_n$ satisfying $X_i \sim \mu_i$, $i=1,\cdots,n$, we denote by $\tilde{\mu}_i$ the distribution measure of $-X_i$, $i=1,\cdots,n$. Substituting $\tilde{\mu}_1, \cdots, \tilde{\mu}_n$ into \eqref{eq:prime1} and noting that $R_{d,s}(\sum_{i=1}^n X_i) = -R_{1-d-s,s}(-\sum_{i=1}^n X_i)$ and $R_{1-\beta_i-\beta_0,\beta_0}(X_i) = -R_{\beta_i,\beta_0}(-X_i)$, we prove the lower convolution bound \eqref{eq:prime2}.
\begin{theorem}[RVaR aggregation]\label{th:qa-4primeprime}
	Let $\boldsymbol \mu=(\mu_1,\dots,\mu_n)\in \M^n$.
	For any $t, s$ with $0 \leq t < t+s \leq 1$,
	\begin{equation}\label{eq:prime2}
		\inf_{\nu \in \Lambda(\boldsymbol \mu) } R_{t,s} (\nu) \ge \sup_{\substack{ \boldsymbol \beta \in (1-t)\Delta_n \\\beta_0\ge s > 0} }
		\sum_{i=1}^n R_{1-\beta_i-\beta_0,\beta_0 }(\mu_i).
		%R_{\boldsymbol{\beta}}^- (\boldsymbol{\mu}).
	\end{equation}
	Moreover,   \eqref{eq:prime2} holds as an equality in the following cases:
	\begin{enumerate}[(i)]
		\item $t+s=1$;
		%\item $n\le2$;
		%		\item \label{item:rvar_decr2}		each of $\mu_1,\dots,\mu_n$ admits an decreasing density below its $t$-quantile;
		\item \label{item:rvar_incr2}
		each of $\mu_1,\dots,\mu_n$ admits an increasing density below its $(1-t)$-quantile;
		\item $\sum_{i=1}^n \mu_i \left[q_0^+(\mu_i), q_{1-t}^-(\mu_i)\right) \le  1-t$.%$\sum_{i=1}^n \mu_i\left((-\infty, q_{1-t}^-(\mu_i))\right)\le  1-t$.
	\end{enumerate}
\end{theorem}

Let $\mu^{t-}$ be the probability measure given by
$$
\mu^{t-} \left(-\infty, x\right]= \min\left\{ \frac{\mu\left(-\infty, x\right]}{t}, 1\right\},~~x\in \R.
$$
That is, $\mu^{t-}$ is the distribution measure of the random variable $q_V(\mu)$ where $V$ is a uniform random variable on $[0,t]$. In the case of Theorem \ref{th:qa-4primeprime} (iii), it equivalently means that each of $\mu_1^{(1-t)-},\dots,\mu_n^{(1-t)-}$ admits an increasing density. We denote by $\boldsymbol \mu^{t- }= (\mu_1^{t-},\dots,\mu_n^{t-})$. Proposition \ref{prop:basic_sym} (symmetric to Proposition \ref{prop:basic}) shows  relevant results.

\begin{proposition}\label{prop:basic_sym}
	For $\boldsymbol \mu=(\mu_1,\dots,\mu_n)\in \M^n$, for $0 \leq t < t+s \leq 1$,
	we have
	$$
	\inf_{\nu \in \Lambda(\boldsymbol \mu) } R_{t,s} (\nu) = \inf_{\nu \in \Lambda(\boldsymbol \mu^{(1-t)-}) } \ES_{s/(1-t)} (\nu)
	$$
	and
	$$	\inf_{\nu \in \Lambda(\boldsymbol{\mu})} q_t^- (\nu) = \inf_{\nu \in \Lambda(\boldsymbol \mu^{t- } )}q_1^- (\nu).$$
\end{proposition}
Similarly to the worst-case values,  for the best-case values of RVaR aggregation, it suffices to consider the one ended at quantile level 1, i.e. the $\ES$ aggregation. In particular, for the worst-case problems of quantile aggregation, it suffices to consider the one at quantile level 1, i.e. the problems $\inf_{\nu \in \Lambda(\boldsymbol \mu^{t-} )}q_1^+ (\nu)$ for generic choices of $\boldsymbol \mu$.

%In Theorem \ref{th:qa-2} below we summarize bounds on $\inf_{\nu \in \Lambda(\boldsymbol \mu) } q_t^- (\nu)$. Most cases can be obtained by sending $s$ to $0$ and replacing $t$ with $1-t$ in Theorem \ref{th:qa-4primeprime}.

\begin{theorem}[Quantile aggregation]\label{th:qa-2}
	For $\boldsymbol \mu\in \M^n$, for $t\in (0,1]$, we have
	\begin{equation}\label{eq:main2pr}\inf_{\nu \in \Lambda(\boldsymbol \mu) } q^-_t (\nu) \ge \sup_{\boldsymbol \beta \in t \Delta_n} 	\sum_{i=1}^n R_{1-\beta_i-\beta_0,\beta_0 }(\mu_i).\end{equation}
	Moreover, \eqref{eq:main2pr} holds as an equality in the following cases:
	\begin{enumerate}[(i)]
		\item $n \leq 2$;
		\item \label{item2_incr}
		each of $\mu_1,\dots,\mu_n$ admits an increasing density below its $t$-quantile; %i.e., each of $\mu_1^{t-},\dots,\mu_n^{t-}$ admits an increasing density on its support,
		\item \label{item2_decr}
		each of $\mu_1,\dots,\mu_n$ admits a decreasing density below its $t$-quantile; %i.e., each of $\mu_1^{t-},\dots,\mu_n^{t-}$ admits a decreasing density on its support,
		\item $\sum_{i=1}^n \mu_i  \left[q_0^+(\mu_i), q_{t}^-(\mu_i)\right)  \leq t$;%$\sum_{i=1}^n \mu_i  \left(-\infty, q_{t}^-(\mu_i)\right)  \leq t$;
		\item $\sum_{i=1}^n \mu_i   \left(q_0^+(\mu_i), q_{t}^-(\mu_i)\right]    \leq t$.%$\sum_{i=1}^n \mu_i^{t-}   \left(q_0^+(\mu_i), \infty\right) \leq 1.$
		
	\end{enumerate}
\end{theorem}

Proposition \ref{prop:reduced_bound_sym} (symmetric to Proposition \ref{prop:reduced_bound}) concerns a reduced lower convolution bound.
\begin{proposition} \label{prop:reduced_bound_sym}
	For $\mu\in \M $ and $t\in (0,1]$, we have
	\begin{equation}\label{eq:main2hom}\inf_{\nu \in \Lambda_n(\mu) } q_t^- (\nu) \ge \sup_{\alpha \in (0,t/n)} n R_{1-t+(n-1)\alpha, t-n\alpha}(\mu) = \sup_{\alpha \in (0,t/n)} \frac{n}{t-n\alpha} \int_{\alpha}^{t-(n-1)\alpha} q_s^{-}(\mu) \d s.
	\end{equation}
	Moreover, \eqref{eq:main2hom} holds as an equality if $\mu$ admits an increasing density below its $t$-quantile.
\end{proposition}
Proposition \ref{prop:nu_sym} (symmetric to Proposition \ref{prop:nu})  shows that $\inf_{\nu \in \Lambda(\boldsymbol \mu) } q_1^- (\nu) $ is always attainable and the infimum can be replaced by a minimum.
\begin{proposition}\label{prop:nu_sym}
	For $\boldsymbol \mu\in \M^n$ and $t \in (0,1]$,  there exists $\nu_- \in \Lambda(\boldsymbol \mu)$ such that
	$
	\inf_{\nu \in \Lambda(\boldsymbol \mu) } q_t^- (\nu) =  q_t^- (\nu_-).
	$
\end{proposition}

Proposition \ref{th:multiply_lower} (symmetric to Proposition \ref{th:multiply}) presents a lower convolution bound for multiplicative risks.

\begin{proposition}\label{th:multiply_lower}
	For $\mu_1, \cdots, \mu_n \in \M$ with support included in $(0, \infty)$, we have
	\begin{equation}\label{eq:main2pr:multiply}
		\inf_{X_i \sim \mu_i, i = 1, \cdots, n } q^-_{t} \left( \prod_{i = 1}^n X_i \right) \ge  \exp \left\{ \sup_{\boldsymbol \beta\in t\Delta_n}  \sum_{i=1}^n R_{1-\beta_i-\beta_0,\beta_0 } \left(  \mu_i\circ \exp \right) \right\}, ~~ t\in (0,1].
	\end{equation}
	Moreover, \eqref{eq:main2pr:multiply} holds as an equality in the following cases (denote by $f_1, \dots, f_n$ the densities of $X_1, \dots, X_n$):
	\begin{enumerate}%[(i)]
		\item[(i)] $n \leq 2$;
		\item[(ii)]
		for each $i = 1, \cdots, n$,  $x\mapsto xf_i(x)$  is decreasing beyond the $t$-quantile of $\mu_i$;
		\item[(iii)]
		for each $i = 1, \cdots, n$, $x\mapsto xf_i(x)$ is increasing beyond the $t$-quantile of $\mu_i$.
		%\item $\sum_{i=1}^n \mu_i\left((q^-_{t}(\mu_i),\infty)\right) \le 1-t$;
		%\item $\sum_{i=1}^n \left( \mu_i\left((-\infty, q^+_{t}(\mu_i))\right) - t \right) \le 1-t$.
		\item[(vi)] $\sum_{i=1}^n \mu_i  \left[q_0^+(\mu_i), q_{t}^-(\mu_i)\right)  \leq t$;%$\sum_{i=1}^n \mu_i  \left(-\infty, q_{t}^-(\mu_i)\right)  \leq t$;
		\item[(vii)] $\sum_{i=1}^n \mu_i   \left(q_0^+(\mu_i), q_{t}^-(\mu_i)\right]    \leq t$.%$\sum_{i=1}^n \mu_i^{t-}   \left(q_0^+(\mu_i), \infty\right) \leq 1.$
	\end{enumerate}
\end{proposition}

\subsection*{Acknowledgements}
We thank an Editor, an Associate Editor, three anonymous referees,  Peng Liu, and Giovanni Puccetti for many helpful comments that have greatly improved the paper. JB and YL gratefully acknowledge financial support from the Air Force Office of Scientific Research under award number FA9550-20-1-0397, and additional support is gratefully acknowledged from NSF 1915967, 2118199, 2229011. YL acknowledges support from The Chinese University of Hong Kong, Shenzhen research startup fund (No. UDF01003336) and Shenzhen Excellent Science and Technology Innovation Talents Development Plan (No. RCBS20231211090814028) and is partly supported by the Guangdong Provincial Key Laboratory of Mathematical Foundations for Artificial Intelligence (Grant No. 2023B1212010001). HL acknowledges support from the National Science Foundation under grants CAREER CMMI-1834710 and IIS-1849280. RW acknowledges financial support from the Natural Sciences and Engineering Research Council of Canada (RGPIN-2024-03728, CRC-2022-00141).

\newpage

\setcounter{lemma}{0}
\renewcommand{\thelemma}{EC.\arabic{lemma}}
\setcounter{proposition}{0}
\renewcommand{\theproposition}{EC.\arabic{proposition}}
\setcounter{theorem}{0}
\renewcommand{\thetheorem}{EC.\arabic{theorem}}
\setcounter{definition}{0}
\renewcommand{\thedefinition}{EC.\arabic{definition}}
\setcounter{corollary}{0}
\renewcommand{\thecorollary}{EC.\arabic{corollary}}
\setcounter{example}{0}
\renewcommand{\theexample}{EC.\arabic{example}}
\setcounter{equation}{0}
\renewcommand{\theequation}{EC.\arabic{equation}}
\setcounter{remark}{0}
\renewcommand{\theremark}{EC.\arabic{remark}}

\begin{center}
	\Large Proposed E-Companion: Technical Appendices  \ref{app:proof}-\ref{app:applications}
\end{center}

\normalsize

\section{Proofs of main results}\label{app:proof}
%\subsection{Proofs in Section \ref{sec:1}}\label{app:sec:1}

\subsection{Proofs in Section \ref{sec:rvar}}
\label{app:D1}
We first present a lemma slightly generalizing the RVaR inequalities in Theorem 1 of \cite{ELW18} to include distributions possibly with no finite mean.

\begin{lemma}\label{lem:1}
	Let  $\alpha_1,\dots,\alpha_n,\beta_1,\dots,\beta_n\in [0,1]$. Denote by $b = \sum_{i=1}^n \beta_i$ and $a = \bigvee_{i=1}^n \alpha_i$. If $b + a \le 1$, then for all $\boldsymbol{\mu} = (\mu_1,\dots,\mu_n)\in \M^n$ and $\nu \in \Lambda(\boldsymbol{\mu})$,
	\begin{equation}\label{eq:lem1-1}
		R_{b, a}(\nu) \leq \sum_{i=1}^n R_{\beta_i, \alpha_i}(\mu_i),
	\end{equation}
	provided the right-hand side of \eqref{eq:lem1-1} is well-defined (no  ``$\infty - \infty$").
\end{lemma}

\begin{proof}[Proof of Lemma \ref{lem:1}]
	%\proof{Proof of Lemma \ref{lem:1}.}
	Theorem 1 of \cite{ELW18} with the notation $\mathrm{RVaR}_{\beta,\alpha}(\mu)=R_{\beta,\alpha}(\mu)$ for   $\alpha,\beta\geq0$, $\alpha+\beta\le 1$
	gives \eqref{eq:lem1-1} if $\mu_1,\dots,\mu_n \in \M_1$.
	For $\mu_1,\dots,\mu_n$ that do not necessarily have finite means, we always assume that the right-hand side of \eqref{eq:lem1-1} is well-defined (no ``$\infty - \infty$").
	
	If there exists some $i$ such that $R_{\beta_i,\alpha_i} (\mu_i) = \infty$,  \eqref{eq:lem1-1} holds trivially. Now we assume $R_{\beta_i, \alpha_i} (\mu_i) < \infty$, $i = 1,\cdots,n$. There are four cases:
	\begin{enumerate}
		\item Suppose $b + a<1$ and $b>0$.  In this case, $R_{b, a}$ and $R_{\beta_i, \alpha_i}$ are continuous with respect to weak convergence on $\M$ (see e.g.~\cite{CDS10}). For $\mu\in\Gamma(\mu_1,\dots,\mu_n)$ such that $\nu=\lambda_\mu$,
		we can find a sequence $\mu^{(k)}$, $k\in \N$ such that all one-dimensional margins of $\mu^{(k)}$ are in $\M_1$,
		and $\mu^{(k)}\to \mu$ weakly as $k\to \infty$.
		As a consequence, all one-dimensional margins of $\mu^{(k)}$, as well as its projection $\lambda_{\mu_k}$, converge weakly.
		Since  \eqref{eq:lem1-1} holds for probability measures in $\M_1$, using the continuity of $R_{b, a}$ and $R_{\beta_i, \alpha_i}$, we know \eqref{eq:lem1-1} holds in this case.
		\item   Suppose $b + a=1$ and $b>0$.
		If
		$R_{1-a, a}(\nu)=-\infty$, \eqref{eq:lem1-1} holds trivially.
		If
		$R_{1-a, a}(\nu)>-\infty$,
		then
		$$
		\lim_{\epsilon \downarrow 0} R_{1-a, a-\epsilon}(\nu) = R_{1-a, a}(\nu)
		$$
		since $R_{1-a, a-\epsilon}(\nu)$ is monotone for $\epsilon\in (0,a)$.
		In the first case, we have shown, for $\epsilon\in (0,\bigwedge_{i=1}^n \alpha_i)$,
		$$
		R_{1-a,a-\epsilon}(\nu)\le \sum_{i=1}^n R_{\beta_i, \alpha_i-\epsilon}(\mu_i).
		$$
		Taking a limit as $\epsilon\downarrow 0$ establishes \eqref{eq:lem1-1}.
		\item Suppose $b + a < 1$ and $b=0$. It implies that $\beta_1=\cdots=\beta_n=0$. Because $R_{0, \alpha_i} (\mu_i) < \infty$, $i=1,\cdots,n$, we have
		$$
		\lim_{\epsilon \downarrow 0} R_{\epsilon, \alpha_i}(\mu_i) = R_{0, \alpha_i}(\mu_i),
		$$
		since $R_{\epsilon, \alpha_i}(\mu_i)$ is monotone for $\epsilon\in (0, 1-\alpha_i)$, $i = 1,\cdots, n$. In the first case, we have shown, for $\epsilon\in (0, 1-a)$,
		$$
		R_{n \epsilon, a}(\nu)\le \sum_{i=1}^n R_{\epsilon, \alpha_i}(\mu_i).
		$$
		Taking a limit as $\epsilon\downarrow 0$ establishes \eqref{eq:lem1-1}.
		\item Suppose $b + a = 1$ and $b=0$. It implies that $\beta_1=\cdots=\beta_n=0$. Because $R_{0, \alpha_i} (\mu_i) < \infty$, $i=1,\cdots,n$, we know $\sum_{i=1}^n R_{0, 1}(\mu_i)$ is well defined. By the linearity of $R_{0,1}$, we have
		$$
		R_{0, 1}(\nu) = \sum_{i=1}^n R_{0, 1}(\mu_i) \leq \sum_{i=1}^n R_{0, \alpha_i}(\mu_i),
		$$
		which establishes \eqref{eq:lem1-1}.%\Halmos
		 \qedhere
	\end{enumerate}
	%\endproof
\end{proof}

\begin{proof}[Proof of Theorem \ref{th:qa-4prime}]
	%\proof{Proof of Theorem \ref{th:qa-4prime}.}
	The inequality \eqref{eq:prime1} is shown in the text above Theorem \ref{th:qa-4prime}. We proceed to prove the sharpness under the following cases.
	\begin{enumerate}[(i)]
		\item
		If $t = 0$, then $R_{t, s} = \ES_{s}$ and $\{\boldsymbol{\beta} \in (1-t) \Delta_n: \beta_0 \geq s\} = \{(s, 0, \cdots, 0)\}$.
		It is well known (e.g., \cite{K01}) that $\ES_s$ is subadditive and comonotonic additive, which gives
		$$
		\sup_{\nu \in \Lambda (\boldsymbol{\mu})} R_{0, s}(\nu) = \sup_{\nu \in \Lambda (\boldsymbol{\mu})}\ES_s(\nu) = \sum_{i=1}^n \ES_{s} (\mu_i) = \inf_{\substack{\boldsymbol \beta\in (0+s)\Delta_n \\\beta_0 \geq s}} \sum_{i=1}^n R_{\beta_i,\beta_0}(\mu_i).
		$$ 		
		%	\item If $n=1$, sharpness of \eqref{eq:prime1} holds trivially. If $n=2$, we take  $\nu \in \Lambda(\boldsymbol{\mu})$ obtained from
		%	random variables $X\sim \mu_1$ and $Y\sim \mu_2$ where
		%	$X$ and $Y$ share the same $(1-t-s)$-tail event  $A$ (i.e., an event of probability $t+s$ on which both $X$ and $Y$ take their largest values; see \cite{WZ21}) and $X$ and $Y$ are counter-monotonic conditional on $A$.
		%	For this specific construction,
		%	$$\begin{aligned}
		%		\sup_{\nu \in \Lambda(\boldsymbol \mu) }R_{t,s} (\nu)
		%		& \geq  \frac{1}{s} \int_{1-t-s}^{1-t} q_{2-t-s-u}^{-}(\mu_1) + q_u^{-}(\mu_2) \d u\\
		%		&  =  \frac{1}{s} \int_{1-s}^{1} q_u^{-}(\mu_1)  \d u + \frac{1}{s} \int_{1-t-s}^{1-t} q_u^{-}(\mu_2) \d u\\
		%		& \geq  \inf_{\beta_1 \in (0, t)} \left\{\frac{1}{s} \int_{1-s-\beta_1}^{1-\beta_1} q_u^{-}(\mu_1)  \d u + \frac{1}{s} \int_{1-t -s+ \beta_1}^{1-t + \beta_1} q_u^{-}(\mu_2) \d u \right\}\\
		%		& = \inf_{\beta_1 \in (0, t)} \left\{ R_{\beta_1, s}(\mu_1) + R_{t-\beta_1, s}(\mu_2) \right\}  \geq \inf_{\substack{\boldsymbol \beta\in (t+s)\Delta_n \\\beta_0 \geq s}} \sum_{i=1}^n R_{\beta_i,\beta_0}(\mu_i).
		%	\end{aligned}
		%	$$
		%	Thus, the bound \eqref{eq:prime1} is sharp.
		
		\item
		%The following proof is built on the main results of \cite{JHW16} on convex order of risk aggregation.
		%		Note that the two cases \eqref{eq:main1sharp}  and  \eqref{eq:main2sharp} are not symmetric. In particular, $\inf_{\nu \in \Lambda(\boldsymbol \mu) } q_1^- (\nu)$
		%		has an analytical formula obtained by \cite{JHW16}, but $\sup_{\nu \in \Lambda(\boldsymbol \mu) }q_0^+ (\nu) $ does not.
		%		
		%		First, we show \eqref{eq:main2sharp}. Note that a distribution that has a decreasing density below its $t$-quantile is supported in either a finite interval $[a,b]$ or
		%		an half real line $[a,\infty)$ for some $a,b\in \R$.
		%		Hence, without loss of generality, we can assume $q_0^+(\mu_i)=0$, $i=1,\dots,n$.
		%		
		%	
		
		%	This proof relies on Proposition \ref{prop:basic}.%, which allows us to assume $t+s=1$.
		%We proceed to prove that the bound \eqref{eq:prime1} is sharp, i.e., ``$\geq$" holds for \eqref{eq:prime1} in this case.
		
		\textbf{Step 1}: Using Proposition \ref{prop:basic} (which will be shown later), we   have
		\begin{equation}\label{eq:rvar_problem_at_0}
			\sup_{\nu \in \Lambda(\boldsymbol \mu) } R_{t,s} (\nu) = \sup_{\nu \in \Lambda(\boldsymbol \mu^{(1-t-s)+}) } \LES_{\frac{s}{t+s}} (\nu).
		\end{equation}
		Hence, it suffices to consider the problem of the right-hand side of \eqref{eq:rvar_problem_at_0}.
		
		\textbf{Step 2}: Since each of $\mu_1,\dots,\mu_n$ admits a decreasing density beyond its $(1-t-s)$-quantile,  each of the measures $\mu_1^{(1-t-s)+}, \dots, \mu_n^{(1-t-s)+}$ admits a decreasing density on its support. We can define an aggregate random variable $T_{s_n}$ by (see Equation (3.4) of \cite{JHW16})
		$$
		T_{s_n}=h(U)\id_{\{U\in (0,s_n)\}} + d(s_n) \id_{\{ U\in[s_n,1]\}},
		$$
		which will be explained below.
		\begin{enumerate}[(a)]
			\item We can write $T_{s_n} = \sum_{i=1}^n X_i$ where $X_i\sim \mu_i^{(1-t-s)+}$, $i=1,\dots,n$. Let $\nu_0$ be the distribution measure of $T_{s_n}$. Lemma 3.4 (c) of \cite{JHW16} gives $\nu_0\in \Lambda(\boldsymbol \mu^{t+})$.
			\item $U$ is a uniform random variable on $[0,1]$, $h,d:[0,1]\to \R$ are functions and $s_n \in [0,1]$ is a constant. They are given by:
			$$
			\begin{aligned}
				h(x) &= \sum_{i=1}^n y_i(x) - (n-1)y(x), ~~ x \in (0, 1),\\
				d(x) &= \frac{1}{1-x} \sum_{i=1}^n \E\left[X_i \id_{\{ y_i(x) - y(x) \leq X_i \leq y_i(x) \}}\right], ~~ x \in (0,1),\\
				s_n &= \inf\{x\in (0,1): h(x) \leq d(x) \},
			\end{aligned}
			$$
			where $y,y_1,\dots,y_n$ are functions on $(0,1)$ satisfying  (see Equations (E1)-(E2) of \cite{JHW16})
			$$
			\begin{aligned}
				&(\mbox{E1}): ~~ \sum_{i=1}^n \p(X_i>y_i(x))=x,\\
				&(\mbox{E2}): ~~ \p(y_i(x)-y(x)< X_i\le y_i(x))=1-x,~~~i=1,\dots,n.
			\end{aligned}
			$$
			\item
			According to Lemma 3.2 of \cite{JHW16}, $h$ is a decreasing function on $(0,s_n)$. Hence, for all $u \in (0, s_n)$, we have $h(u) \geq d(s_n)$, and further $d(s_n) = q_0^+ (\nu_0)$.
		\end{enumerate}
		\textbf{Step 3}: Denote by $a = \min\{\frac{t}{t+s}, s_n\}$. We proceed to show
		\begin{equation}\label{eq:rvar_da}
			\LES_{\frac{s}{t+s}}(\nu_0) = d(a).
		\end{equation}
		We verify this by direct computation. If $t/(t+s) \geq s_n$, then
		\begin{equation*}
			\LES_{\frac{s}{t+s}}(\nu_0) = \frac{1}{\frac{s}{t+s}}\E\left[T_{s_n} \id_{\{U \in [\frac{t}{t+s}, 1]\}}\right] = d(s_n);
		\end{equation*}
		if $t/(t+s) < s_n$, then
		\begin{equation*}
			\begin{aligned}
				\LES_{\frac{s}{t+s}}(\nu_0) &= \frac{1}{\frac{s}{t+s}}\E\left[T_{s_n} \id_{\{U \in [\frac{t}{t+s}, 1]\}}\right]\\
				&= \frac{1}{\frac{s}{t+s}} \left(\E\left[T_{s_n}\right] - \E\left[h(U) \id_{\{ U \in (0, \frac{t}{t+s}) \}}\right]\right)\\
				&= \frac{1}{\frac{s}{t+s}} \left(\sum_{i=1}^n \E\left[X_i\right] - \sum_{i=1}^n \E\left[X_i (\id_{\left\{ X_i > y_i(\frac{t}{t+s}) \right\}} + \id_{\left\{ X_i < y_i(\frac{t}{t+s}) - y(\frac{t}{t+s}) \right\}})\right]\right)\\
				&= \frac{1}{\frac{s}{t+s}} \sum_{i=1}^n \E\left[X_i \id_{\left\{ y_i(\frac{t}{t+s}) - y(\frac{t}{t+s}) \leq X_i \leq y_i(\frac{t}{t+s}) \right\}}\right] = d\left(\frac{t}{t+s}\right),
			\end{aligned}
		\end{equation*}
		where the third equality is due to Lemma 3.3 of \cite{JHW16}.
		
		\textbf{Step 4}:
		We now show
		\begin{equation}\label{eq:da_beta}
			d(a) = R_{\boldsymbol{\beta}}^+(\boldsymbol{\mu}),
		\end{equation}
		for some $\boldsymbol{\beta} \in (t+s)\Delta_n$ satisfying $\beta_0 \geq s$, which is defined by
		$$
		\begin{aligned}
			&\beta_0 = (t+s) (1-a), ~~\beta_i = (t+s) \mu_i^{(1-t-s)+}(y_i(a), \infty) = \mu_i(y_i(a), \infty), ~i=1,\cdots,n.
		\end{aligned}
		$$
		
		%\com{never write $+\infty$}
		According to (E1), $\sum_{i=1}^{n}\beta_i=(t+s)a$. We have $(\beta_0, \beta_1, \cdots, \beta_n) \in (t+s)\Delta_n$, and $\beta_0 \geq s$. Hence,
		\begin{align*}
			d(a) & = \sum_{i=1}^n \frac{1}{1-a} \E\left [ X_i \id_{\{y_i(a)-y(a)\le X_i \le y_i(a)\} }\right]
			\\& = \sum_{i=1}^n \frac{1}{1-a} \int_{y_i(a) - y(a)}^{y_i(a)} x \mu_i^{(1-t-s)+} (\d x)\\
			& = \sum_{i=1}^n \frac{1}{1-a} \int_{a-\frac{\beta_i}{t+s}}^{1-\frac{\beta_i}{t+s}} q_u^- (\mu_i^{(1-t-s)+}) \d u\\
			%		& = \sum_{i=1}^n \frac{1}{(1-a)(t+s)} \int_{(1-t-s)+(t+s)a-\beta_i}^{(1-t-s)+(t+s)-\beta_i} q_{\frac{v-(1-t-s)}{t+s}}^- (\mu_i^{(1-t-s)+}) \d v\\
			& = \sum_{i=1}^n \frac{1}{1-a} \int_{a-\frac{\beta_i}{t+s}}^{1-\frac{\beta_i}{t+s}} q_{1-t-s+(t+s)u}^- (\mu_i) \d u\\
			& = \sum_{i=1}^n \frac{1}{(1-a)(t+s)} \int_{1-\beta_i-\beta_0}^{1-\beta_i} q_{v}^- (\mu_i) \d v\\
			& = \sum_{i=1}^n \frac{1}{\beta_0} \int_{1-\beta_i-\beta_0}^{1-\beta_i} q_v^- (\mu_i) \d v
			= \sum_{i=1}^n R_{\beta_i,\beta_0}(\mu_i),
		\end{align*}
		where the third equality is due to the fact that $\mu_i^{(1-t-s)+} (-\infty, y_i(a) - y(a) ] = a -\frac{\beta_i}{t+s}$ derived from (E2).% and the fourth equality is due to the fact that $q_v^- (\mu_i) = q_{\frac{v-(1-t-s)}{t+s}}^- (\mu_i^{(1-t-s)+})$, $v \in [1-t-s,1]$.

		\textbf{Step 5}: Combining \eqref{eq:rvar_problem_at_0}, \eqref{eq:rvar_da} and \eqref{eq:da_beta}, we have
		$$
		\begin{aligned}
			\sup_{\nu \in \Lambda(\boldsymbol \mu) } R_{t,s} (\nu)
			= \sup_{\nu \in \Lambda(\boldsymbol \mu^{(1-t-s)+}) } \LES_{\frac{s}{t+s}} (\nu)
			& \ge \LES_{\frac{s}{t+s}} (\nu_0)
			\\ &=d(a) =  \sum_{i=1}^n R_{\beta_i,\beta_0}(\mu_i) \ge \inf_{\substack{\boldsymbol \beta'\in (t+s)\Delta_n \\\beta'_0\ge s}} \sum_{i=1}^n R_{\beta'_i,\beta'_0}(\mu_i).
		\end{aligned}
		$$
		Thus, the bound \eqref{eq:prime1} is sharp.
		\item  It suffices to prove the statement for $t=1-s$.
		The assumption %$\sum_{i=1}^n \mu_i (q_{0}^+(\mu_i), \infty)\le 1 $
		$\sum_{i=1}^n \mu_i (q_{0}^+(\mu_i),q^-_{1}(\mu_i)]\le 1 $
		allows for the existence of a lower mutually exclusive (see Definition \ref{def:lme} below) random vector $(X_1,\dots,X_n)$ where $X_i\sim \mu_i$, $i=1,\dots,n$.
		Hence, the desired result follows from Lemma \ref{lem:me} below by checking that the bound \eqref{eq:prime1} is attained by such a vector. %\Halmos
		\qedhere
	\end{enumerate}
	%\endproof
\end{proof}

\begin{definition}[Mutually exclusivity]\label{def:lme}
	We say that a random vector  $(X_1,\dots,X_n)$ where
	$X_i\sim \mu_i$, $i=1,\dots,n$
	is \emph{lower mutually exclusive} if
	$\p(X_i>q_0^+(\mu_i),X_j>q_0^+(\mu_j))=0$ for all $i \ne j$
	and it 
	is \emph{upper mutually exclusive} if
	$\p(X_i<q_1^-(\mu_i),X_j<q_1^-(\mu_j))=0$ for all $i \ne j$.%\Halmos
\end{definition}

\begin{lemma}\label{lem:me}
	If  random variables $X_1,\dots,X_n$ are lower mutually exclusive and bounded from below, then for $\alpha \in (0,1)$,
	\begin{equation}\label{eq:lem:me}
		R_{1-\alpha,\alpha}\left(\sum_{i=1}^n X_i\right) =   \sum_{i=1}^n R_{\beta_i,\alpha}(X_i),
	\end{equation}
	for some $\beta_1,\dots,\beta_n\in [0,1)$ with $\sum_{i=1}^n \beta_i=1-\alpha$.
\end{lemma}

\begin{proof}
	%\proof{Proof of Lemma \ref{lem:me}.}	
	Without loss of generality, we assume $q_0^+(\mu_i) =0$ for each $i$.
	If $q^+_\alpha \left(\sum_{i=1}^n X_i\right)=0$,
	then $\sum_{i=1}^n \p( X_i>0) = \p\left( \sum_{i=1}^n X_i > 0\right)\le 1-\alpha$.
	Hence, we can choose $\beta_i \ge \p(X_i>0)$ for each $i$,
	and
	both sides of \eqref{eq:lem:me} are $0$. Below we assume $q^+_\alpha(\sum_{i=1}^n X_i)>0$.
	
	%We look at $\alpha\in(0,1)$ such that
	%$\p(\sum_{i=1}^n X_i> q^+_\alpha(\sum_{i=1}^n X_i))=1-\alpha$. Note that such $\alpha$ is dense in the set $\{\alpha'\in (0,1): q^+_{\alpha'}(\sum_{i=1}^n X_i)>0\}$.
	
	First, we assume that the distribution $\mu_i$ of $X_i$ is continuous on $\left\{X_i>0\right\}$ for each $i=1,\dots,n$, and so is the conditional distribution of $\sum_{i=1}^n X_i$ on $\left\{\sum_{i=1}^n X_i>0\right\}$.
	
	Let $y= q^+_\alpha\left(\sum_{i=1}^n X_i\right)$ and $A= \left\{\sum_{i=1}^n X_i\le y\right\}$. We have $\p(A)=\alpha$.
	%$B=\{\sum_{i=1}^n X_i =0\}$ and $b=\p(B)/n$.
	For each $i=1,\dots,n$, let
	$\alpha_i=\p(A\cap \{X_i>0\}) =\p(0<X_i\le y)  $
	and $t_i=\p(X_i>0) $.
	%Let $t_i$ be any number satisfying
	%$\sum_{i=1}^n t_i=n-1$ and $t_i\le \p(X_i=0)$.
	%Because $X_1,\dots,X_n$ are mutually exclusive,  $\sum_{i=1}^n t_i\le 1$.
	By direct calculation
	\begin{align*}
		R_{1-\alpha,\alpha}\left(\sum_{i=1}^n X_i\right)=  \E\left[\sum_{i=1}^n X_i \mid  A\right]
		&
		=  \sum_{i=1}^n \E\left[ X_i \mid A\right]
		\\&
		=  \frac 1 \alpha \sum_{i=1}^n\left( 0 + \E\left[ X_i  \id_{ A \cap \{X_i>0\}} \right]\right)
		\\&
		=  \frac 1 \alpha \sum_{i=1}^n\left( \int_{1-t_i+\alpha_i-\alpha}^{1-t_i} q_u^{+}(\mu_i) \d u + \int_{1-t_i }^{1-t_i+\alpha_i} q_u^{+}(\mu_i) \d u \right)
		\\&= \sum_{i=1}^nR_{ t_i-\alpha_i,\alpha}(X_i).
	\end{align*}
	We can check, by lower mutual exclusivity and the continuity assumption, that
	$$\sum_{i=1}^n (t_i-\alpha_i) = \sum_{i=1}^n\left( \p(X_i>0) - \p(0<X_i\le y)\right)   = \sum_{i=1}^n \p(X_i>y) =\p\left( \sum_{i=1}^n X_i>y\right) = 1-\alpha.
	$$
	By \eqref{eq:prime1}, we have
	$$ R_{1-\alpha,\alpha}\left(\sum_{i=1}^n X_i\right) \le \sum_{i=1}^nR_{ t_i-\alpha_i,\alpha}(X_i) .$$
	Therefore, \eqref{eq:lem:me} holds by choosing $\beta_i=t_i-\alpha_i$, $i=1,\dots,n$.
	In case the conditional distributions of $X_1,\dots,X_n$ are positive and not continuous, we can approximate (by convergence in distribution) $X_1,\dots,X_n$ by conditionally continuous distributions while fixing $\p(X_i>0)$ for each $i$. The compactness of the set $(1-\alpha)\overline{\Delta}_{n-1}$ on which $(\beta_1,\dots,\beta_n)$ takes values and the continuity of $R_{\beta,\alpha}$ with respect to weak convergence (e.g., \cite{CDS10}) yields the desirable result.  %\Halmos
	%\endproof
\end{proof}

\begin{proof}[Proof of Proposition \ref{prop:basic}]
	%\proof{Proof of Proposition \ref{prop:basic}.}
	The first equality is a direct consequence of Theorem 4.1 and Example 6.3 of \cite{LW16}, and the second equality follows from Remark 4.1 of \cite{LW16}.%\Halmos
	%\endproof
\end{proof}

\subsection{Proofs in Section \ref{sec:2}}
\label{app:D2}

\begin{proof}[Proof of Theorem \ref{th:qa-4}]
	%\proof{Proof of Theorem \ref{th:qa-4}.}
	The convolution bound \eqref{eq:main1pr} is obtained by taking a limit of \eqref{eq:prime1} in Theorem \ref{th:qa-4} using \eqref{eq:quantileconvergence}.
	Similarly, based on Theorem \ref{th:qa-4} and the fact that $R_{\beta_i,\beta_0}$ is continuous in $\beta_0$,
	this limit argument also gives sharpness in (i), (ii) and (iv).
	Next we proceed to show  sharpness in (iii) and (v).
	\begin{enumerate}[(i)]

		\item[(iii)]
		First, we note that
		\begin{equation}\label{eq:sharp_convert}
			\sup_{\nu \in \Lambda(\boldsymbol \mu) }q^+_{t} (\nu) = -\inf_{\tilde{\nu} \in \Lambda(\tilde{\boldsymbol \mu}) }q^-_{1-t} (\tilde{\nu}),%\le \inf_{\boldsymbol \beta\in (1-t)\Delta_n}   \sum_{i=1}^n R_{\beta_i,\beta_0 }(\mu_i).
		\end{equation}
		where $\tilde{\mu}_i$ is the distribution measure of the random variable $-X_i$ with $X_i \sim \mu_i$, $i=1,\cdots,n$ and $\tilde{\boldsymbol \mu} = (\tilde{\mu}_1, \cdots, \tilde{\mu}_n)$. The fact that each of $\mu_1,\dots,\mu_n$ admits an increasing density beyond its $t$-quantile implies that each of $\tilde{\mu}_1,\dots,\tilde{\mu}_n$ admits an decreasing density below its $(1-t)$-quantile. Note that a distribution that has a decreasing density below its $(1-t)$-quantile is supported in either a finite interval $[a,b]$ or a half real line $[a,\infty)$ for some $a,b\in \R$. Hence, without loss of generality, we can assume $q_0^+(\tilde{\mu}_i)=0$, $i=1,\dots,n$.
		
		For   sharpness of \eqref{eq:main1pr},
		we need to show $$
		\sup_{\nu \in \Lambda(\boldsymbol \mu) }q^+_{t} (\nu)  \ge \inf_{\boldsymbol \beta \in (1-t)\Delta_n} \sum_{i=1}^n R_{\beta_i,\beta_0}(\mu_i).
		$$
		By \eqref{eq:sharp_convert} and the definition of $R_{\beta,\alpha}$, it suffices to show
		\begin{equation}\label{eq:main2sharp-pf}
			\inf_{\tilde{\nu} \in \Lambda(\tilde{\boldsymbol \mu}) } q_{1-t}^- (\tilde{\nu})\le  \sup_{\boldsymbol \beta \in (1-t)\Delta_n} \sum_{i=1}^n R_{1-\beta_i-\beta_0,\beta_0}(\tilde{\mu}_i).
		\end{equation}
		
		Fix $j\in \{1,\dots,n\}$ and $\beta_j \in (0,1-t)$. By taking $\beta_i=0$ for $i \in \{1,\cdots,n\} \backslash \{j\}$ and $\beta_0 = 1-t-\beta_j$, we get
		$$
		\sup_{\boldsymbol \beta \in (1-t)\Delta_n} \sum_{i=1}^n R_{1-\beta_i-\beta_0,\beta_0}(\tilde{\mu}_i)
		\ge  R_{t, 1-t-\beta_j}(\tilde{\mu}_j)  + \sum_{i\ne j} R_{t+\beta_j, 1-t-\beta_j}(\tilde{\mu}_i)
		\ge R_{t, 1-t-\beta_j}(\tilde{\mu}_j).
		$$
		Taking a supremum over $\beta_j\in (0,1-t)$ and $j\in \{1,\dots,n\}$   yields
		\begin{equation}\label{eq:main2sharp-pf2}
			\sup_{\boldsymbol \beta \in (1-t)\Delta_n} \sum_{i=1}^n R_{1-\beta_i-\beta_0,\beta_0}(\tilde{\mu}_i)\ge \bigvee_{j=1}^n\sup_{\beta_j \in (0,1-t)} R_{t, 1-t-\beta_j}(\tilde{\mu}_j) = \bigvee_{j=1}^n q_{1-t}^-(\tilde{\mu}_j).
		\end{equation}

		If $\bigvee_{j=1}^n q_{1-t}^-(\tilde{\mu}_j)=\infty$, then the right-hand side of  \eqref{eq:main2sharp-pf} is $\infty,$ which holds automatically.
		%\eqref{eq:main2sharp-pf2} implies
		%$$\sup_{\boldsymbol \beta \in (1-t)\Delta_n} \sum_{i=1}^n R_{1-\beta_i-\beta_0,\beta_0}(\tilde{\mu}_i)=\infty,$$ and
		%hence \eqref{eq:main2sharp-pf} holds.
		If $\bigvee_{j=1}^n q_{1-t}^-(\tilde{\mu}_j)<\infty$, we can apply Corollary 4.7 of \cite{JHW16}, using the condition that each of $\mu_1,\dots,\mu_n$ admits a decreasing density below its $(1-t)$-quantile. This gives
		\begin{equation}\label{eq:main2sharp-pf3}
			\inf_{\tilde{\nu} \in \Lambda(\tilde{\boldsymbol \mu}) } q_{1-t}^+ (\tilde{\nu})= \max\left\{\bigvee_{i=1}^{n} q_{1-t}^-(\tilde{\mu}_i), \sum_{i=1}^n R_{t,1-t}(\tilde{\mu}_i)\right\}.
		\end{equation}
		Also note that in this case, $R_{t,1-t}(\tilde{\mu}_i) < \infty$, $i=1,\cdots,n$, %$\mu_1^{t-},\dots,\mu_n^{t-} \in \M_1$,
		and hence
		\begin{equation}\label{eq:main2sharp-pf4}  \sup_{\boldsymbol \beta \in (1-t)\Delta_n} \sum_{i=1}^n R_{1-\beta_i-\beta_0,\beta_0}(\tilde{\mu}_i) \ge \sum_{i=1}^n R_{t,1-t}(\tilde{\mu}_i).
		\end{equation}
		Combining \eqref{eq:main2sharp-pf2}-\eqref{eq:main2sharp-pf4}, we get
		\eqref{eq:main2sharp-pf}.
		%		
		%		Further for any $\boldsymbol \beta \in \Delta_n$, we have
		%		\begin{equation}\label{eq:sharp_convertback}
		%		-\sum_{i=1}^n R_{1-\beta_i-\beta_0,\beta_0}(\tilde{\mu}_i) =
		%		\sum_{i=1}^n R_{\beta_i,\beta_0}(\mu_i).
		%		\end{equation}
		%		
		%		In view of \eqref{eq:sharp_convert}, \eqref{eq:main2sharp-pf} and \eqref{eq:sharp_convertback}, we have
		%		\begin{equation*}
		%		\sup_{\nu \in \Lambda(\boldsymbol \mu) }q^+_{t} (\nu) = -\!\!\inf_{\tilde{\nu} \in \Lambda(\tilde{\boldsymbol \mu}) }q^-_{1-t} (\tilde{\nu}) \ge   -\!\!\sup_{\boldsymbol \beta \in (1-t)\Delta_n} \sum_{i=1}^n R_{1-\beta_i-\beta_0,\beta_0}(\tilde{\mu}_i) = \!\!\inf_{\boldsymbol \beta \in (1-t)\Delta_n} \sum_{i=1}^n R_{\beta_i,\beta_0}(\mu_i),
		%		\end{equation*}
		%		which proves the sharpness together with \eqref{th:qa-4}.
		%		\item It suffices to prove the case $t=0$. The assumption $\sum_{i=1}^n \mu_i\left(\left(q_{0}^+(\mu_i),q^-_{1}(\mu_i)\right]\right)\le 1 $ allows for the existence of a lower mutually exclusive (see Definition \ref{def:lme}) random vector $(X_1,\dots,X_n)$ where $X_i\sim \mu_i$, $i=1,\dots,n$.
		%		Letting $\beta_i = \mu_i\left((q_0^+(\mu_i), +\infty)\right)$, $i=1,\cdots,n$ and $\beta_0 = 1-\sum_{i=1}^n \beta_i$, we have
		%		$R_{\beta_i, \beta_0} (\mu_i) = q_0^+(\mu_i)$ and hence
		%		$$
		%		q_0^{+}\left( \sum_{i=1}^n X_i \right) \geq \sum_{i=1}^n q_0^+(\mu_i) = \sum_{i=1}^n R_{\beta_i, \beta_0} (\mu_i).
		%		$$
		%		Hence, the desired result follows as the bound \eqref{eq:main1pr} is attained by such a vector.
		
		\item[(v)] It suffices to prove the case $t=0$. The assumption $\sum_{i=1}^n \mu_i [q_{0}^+(\mu_i),q^-_{1}(\mu_i))\le 1 $
		allows for the existence of an upper mutually exclusive (see Definition \ref{def:lme}) random vector $(X_1,\dots,X_n)$ where $X_i\sim \mu_i$, $i=1,\dots,n$. Hence, we have
		$$
		q_0^+ \left( \sum_{i=1}^n X_i \right) = \min_{1 \le i \le n} \left( q_0^+ (\mu_i) + \sum_{j\ne i} q_1^- (\mu_j) \right) \geq \inf_{\boldsymbol \beta \in \Delta_n}  \sum_{i=1}^n R_{\beta_i, \beta_0}(\boldsymbol{\mu}).
		$$
		Hence, the desired result follows as the bound \eqref{eq:main1pr} is attained by such a vector.\qedhere
	\end{enumerate}
%	\begin{definition}[Upper mutually exclusivity]\label{def:ume}
%		We say that a random vector  $(X_1,\dots,X_n)$ where 
%		$X_i\sim \mu_i$, $i=1,\dots,n$
%		is \emph{upper mutually exclusive} if 
%		$\p(X_i<q_1^-(\mu_i),X_j<q_1^-(\mu_j))=0$ for all $i \ne j$.
%	\end{definition}
	%\endproof
\end{proof}

\begin{proof}[Proof of Proposition \ref{prop:reduced_bound}]
	%\proof{Proof of Proposition \ref{prop:reduced_bound}.}
	Letting $\beta_1=\cdots=\beta_n=\alpha$ in Theorems \ref{th:qa-4} and \ref{th:qa-2}, we immediately get \eqref{eq:main1hom}.
	We show \eqref{eq:main1hom} holds as an equality in this case of decreasing density. %, as \eqref{eq:main2hom} is symmetric.
	Note that the second equality in \eqref{eq:main1hom} is simply the definition.
	By Proposition 1 of \cite{EPRWB14},
	$  \sup_{\nu \in \Lambda_n(\mu) }q_t^+ (\nu) $ is equal to $n$ times the conditional mean of $\mu$ on an interval $[t+(n-1)\alpha, 1-\alpha]$ for some $\alpha \in [0,\frac{1-t}{n}]$.
	Therefore,
	$$
	\sup_{\nu \in \Lambda_n(\mu) }q_t^+ (\nu) \ge   \inf_{\alpha \in (0, \frac{1-t}{n})} n R_{\alpha, 1-t-n\alpha}(\mu).% = \inf_{\alpha \in  (0,\frac1n  ) }  n R_{[t+(1-t)(n-1)\alpha ,1-(1-t)\alpha  ] }(\mu).
	$$
	Also note that the ``$\le$" sign in \eqref{eq:main1hom} is implied by Proposition \ref{prop:reduced_bound}. Hence, \eqref{eq:main1hom} holds as an equality in this case.%\Halmos
	%\endproof
\end{proof}

\begin{proof}[Proof of Proposition \ref{th:multiply}]
	For any fixed $X_i \sim \mu_i$, $i = 1, \cdots, n$, we define $Y_i = \log \left(X_i\right)$ and hence have
	$$
	q_t^{+} \left( \prod_{i = 1}^n X_i \right) = q_t^{+} \left( \exp \left\{ \sum_{i=1}^n Y_i \right\} \right) = \exp \left\{ q_t^{+} \left(  \sum_{i=1}^n Y_i \right) \right\}.
	$$
	We obtain the desired results by investigating the corresponding quantile problem $q_t^{+} \left(  \sum_{i=1}^n Y_i \right)$. We only prove the cases (ii)-(iii). Denote by $f_i$ the density of $X_i$. The density of $Y_i$ at $y \in \R$ is
	$
	\exp(y) \cdot f_i (\exp(y)).
	$
	According to Theorem \ref{th:qa-4}, \eqref{eq:main1pr:multiply} is sharp if $\exp(y) \cdot f_i (\exp(y))$ are all decreasing (resp. increasing) beyond the $t$-quantile of $Y_i$ for all $i = 1, \cdots, n$. With a change of variables (and the fact that $\log$ is strictly increasing), the condition is translated into that the functions $x \cdot f_i(x)$  are all decreasing (resp. increasing) beyond the $t$-quantile of $X_i$ for all $i = 1, \cdots, n$.
\end{proof}

\begin{proof}[Proof of Proposition \ref{prop:nu}]
	%\proof{Proof of Proposition \ref{prop:nu}}
	The statement for $q_t^+$, $t\in (0,1)$, is shown in Lemma 4.2 of \cite{BJW14}. %$q_t^\pm$
	The case of $t=0$ %or $t=1$
	follows from the same argument by noting the upper semicontinuity of $q_0^+$. %\Halmos% and the lower semicontinuity of $q_1^-$.
	%\endproof
\end{proof}

\begin{proof}[Proof of Proposition \ref{prop:1end}]
	%\proof{Proof of Proposition \ref{prop:1end}.}
	% Note that $F_i \le F_i^{[m]}$, $i=1,\dots,n$. Hence,  by a simple stochastic dominance argument,
	To show  the ``$\ge $" direction of \eqref{eq:2}, we note that
	for any $\nu_Y \in \Lambda(\boldsymbol{\mu^{[m]}})$ such that $Y_1 \sim \mu^{[m]}_1,\dots, Y_n\sim \mu^{[m]}_n$ and $\sum_{i=1}^n Y_i \sim \nu_Y$, by letting $X_i= q^{-}_{U_{Y_i}}(\mu_i)$, $i=1,\dots,n$,
	we get $X_i\sim \mu_i$ and $X_i \ge Y_i$ for each $i=1,\dots,n$. Denote by $\nu_X$ the distribution measure of $\sum_{i=1}^n X_i$. It follows that $\nu_X \in \Lambda(\boldsymbol{\mu})$ and
	$
	q^+_t(\nu_X) \ge q^+_t(\nu_Y),
	$
	which gives
	$$
	\sup_{\nu \in \Lambda(\boldsymbol{\mu})} q^+_t(\nu) \geq  \sup_{\nu \in \Lambda(\boldsymbol{\mu^{[m]}})} q^+_t(\nu).
	$$
	To show  the ``$\le$" direction of \eqref{eq:2}, for any $\nu_X \in \Lambda(\boldsymbol{\mu})$ such that random variables $X_1 \sim \mu_1, \dots, X_n \sim \mu_n$ and $\sum_{i=1}^n X_i \sim \nu_X$, let $Y_i=X_i\wedge m$, $i=1,\dots,n$. Write $S_X= \sum_{i=1}^n X_i$ and $S_Y= \sum_{i=1}^n Y_i$. Denote by $\nu_Y$ the distribution measure of $S_Y$. We have $\nu_Y \in \Lambda(\boldsymbol{\mu^{[m]}})$.
	By Corollary 1 of \citet{ELW18}, we have, for $\epsilon >0$,
	$$
	q^{-}_{t+\epsilon} (\nu_X)   \le \sum_{i=1}^n q^{-}_{1-(1-t-\epsilon)/n} (\mu_i).
	$$
	Taking a limit of the above equation as $\epsilon \downarrow 0$, we obtain
	\begin{equation}
		\label{eq:appendix1}
		q^+_t(\nu_X)   \le \sum_{i=1}^n q^+_{1-(1-t)/n} (\mu_i) \le  m.
	\end{equation}
	
	It is clear that $(S_X\wedge m) \le S_Y$ because the real function $x\mapsto x\wedge m$ is subadditive. Denote by $\tilde{\nu}$ the distribution of $S_X \wedge m$. Hence, \eqref{eq:appendix1} implies
	\begin{equation}\label{eq:appendix2}
		q^+_t(\nu_X)  =  q_t^+(\tilde{\nu})  \le  q_t^+(\nu_Y).
	\end{equation}
	Taking a supremum of \eqref{eq:appendix2} over all possible choices of $\nu_X \in \Lambda(\boldsymbol{\mu})$, we get
	$$
	\sup_{\nu \in \Lambda(\boldsymbol{\mu})} q^+_t(\nu) \leq  \sup_{\nu \in \Lambda(\boldsymbol{\mu^{[m]}})} q^+_t(\nu).%\Halmos
	$$
	This completes the proof.
	%\endproof
\end{proof}

\begin{proof}[Proof of Proposition \ref{prop:qa-1}]
	%\proof{Proof of Proposition \ref{prop:qa-1}.}
	It is a direct corollary by letting $t \downarrow 0$ in Theorem \ref{th:qa-4} and $t \uparrow 1$ in Theorem \ref{th:qa-2} respectively. %\Halmos
	%\endproof
\end{proof}
\begin{proof}[Proof of Proposition \ref{prop:prop1}]
	%\proof{Proof of Proposition \ref{prop:prop1}.}	
	Note that
	\begin{align*}\sup_{\boldsymbol \beta \in \Delta_n} \sum_{i=1}^n R_{1-\beta_i-\beta_0,\beta_0 }(\mu_i)
		&\ge \lim_{\epsilon\downarrow0 }  \sum_{i=1}^n R_{(n-1)\epsilon ,1- n\epsilon }(\mu_i) \\& =\sum_{i=1}^n R_{0,1}(\mu_i)
		= \lim_{\epsilon\downarrow0 }  \sum_{i=1}^n R_{\epsilon ,1- n\epsilon }(\mu_i)
		\ge \inf_{\boldsymbol \beta\in \Delta_n}  \sum_{i=1}^n R_{\beta_i, \beta_0}(\mu_i).
	\end{align*}
	Hence, the second and the third inequalities in \eqref{eq:main1p} hold.
	The first and the last inequalities are due to Proposition \ref{prop:qa-1}.%\Halmos
	%\endproof
\end{proof}

\subsection{Proofs in Section \ref{sec:6}}\label{app:D3}

\begin{proof}[Proof of Theorem \ref{th:qa-5}]
	%\proof{Proof of Theorem \ref{th:qa-5}.}	
	%	The proof is divided by two parts, where the first is for results related to \eqref{eq:opt_struc} and the second is for \eqref{eq:opt_explic}. \com{please revise the proof according to the new formulation.}
	
	By assumption and  Proposition \ref{prop:nu}, there exists $\nu_+ \in \Lambda(\boldsymbol \mu)$  such that  	
	$$
	q_0^+ (\nu_+) = \sup_{\nu \in \Lambda(\boldsymbol \mu) }q_0^+ (\nu) = R^{+}_{\boldsymbol \beta}(\boldsymbol \mu).% = \sum_{i=1}^n R_{\beta_i ,\beta_0}(\mu_i).
	$$
	Take  $X_i^*\sim \mu_i$, $i=1,\cdots,n$ such that $\sum_{i=1}^n X_i^*\sim\nu_+$ and
	$
	\sum_{i=1}^n X_i^* \geq q_0^+ (\nu_+)
	$ almost surely.
	We divide the proof into several steps. We first prove the properties of $X_i^*$ in \eqref{eq:opt_struc} in Steps 1-3 and the feasibility of \eqref{eq:opt_explic} and its optimality given the above sufficient condition in Steps 4-5. In Steps 1-3, we will show that the probability space $\Omega$ is divided into
	$\Omega = A_1 \cup \cdots \cup A_n \cup  B$ where $A_i$ is defined by $A_i=\{X_i^* > q_{1-\beta_i}^- (\mu_i) \}$ (``right-tail" parts of $X_1,\cdots, X_n$) and  $B^c = \bigcup_{i=1}^n A_i$, with the following properties:	
	\begin{enumerate}%[(a)]
		\item[(a)] on the set $B$, $X_i^* \sim \mu_i^{[1-\beta_0-\beta_i, 1-\beta_i]}$ for all $i=1,\cdots,n$ and $\sum_{i=1}^n X_i^* = q_0^+(\nu_+)$ almost surely;
		
		\item[(b)]  for any fixed $i=1,\cdots,n$, on the set $A_i$, $X_i^* \sim \mu_i^{(1-\beta_i, 1]}$ and $X_j^* \sim \mu_j^{[0,1-\beta_0-\beta_j)}$ for all $j\neq i$.
		%\item  $\sum_{i=1}^n X_i^* = q_0^{+}(\nu_+)$ on $A^c$ almost surely.
	\end{enumerate}
	\noindent
	\textbf{Step 1}: We show that the set $\left\{ \sum_{i=1}^n X_i^* =q_0^+ (\nu_+) \right\}$ has probability no less than $\beta_0$. By \eqref{eq:RVaR_ineq}, we have
	\begin{equation}\label{eq:thm1_sharp}
		q_0^+(\nu_+) \le R_{1-\beta_0, \beta_0}(\nu_+) \le
		R^{+}_{\boldsymbol \beta}(\boldsymbol \mu) = q_0^+(\nu_+),
	\end{equation}
	and hence  all inequalities in \eqref{eq:thm1_sharp} are equalities. The fact that $q_0^+(\nu_+) = R_{1-\beta_0, \beta_0}(\nu_+)$ implies that
	$q_t^- (\nu_+) = q_0^+ (\nu_+)$ for all $t \in (0, \beta_0]$ and the set $\left\{ \sum_{i=1}^n X_i^* =q_0^+ (\nu_+) \right\}$ has probability no less than $\beta_0$.
	
	\noindent
	\textbf{Step 2}: We proceed to show that the events (``body" parts of $X_1,\cdots,X_n$)
	\begin{equation}\label{eq:set_A^c}
		\{q_{1-\beta_0-\beta_i}^{-}(\mu_i) \leq X_i^* \leq q_{1-\beta_i}^{-}(\mu_i)\},\; i=1,\cdots, n,
	\end{equation}
	are identical and   $\sum_{i=1}^n X_i^* = q_0^+(\nu_+)$ almost surely on this set.
	
	As the events $A_i = \{X_i^* > q_{1-\beta_i}^- (\mu_i) \}$, $i=1,\cdots,n$, and $B^c= \cup_{i=1}^n A_i$, we have $\p(B^c)\leq \p(A_1)+\cdots+\p(A_n) = \sum_{i=1}^n \beta_i =1-\beta_0$. Denote by $\kappa_i \in \M$ the distribution measure of
	$
	T_i = X_i^* \id_{A_i^c} + m \id_{A_i}, ~i=1,\cdots,n,
	$
	where $m$ is a real number satisfying that $m < \min_{1 \le i \le n}  q_{1-\beta_0-\beta_i}^{-} (\mu_i) $. Denote by $\tau$ the distribution measure of the sum variable $\sum_{i=1}^n T_i$.
	It is verified that $\kappa_i$ has a finite mean and $R_{\beta_i, \beta_0}(\mu_i) = \ES_{\beta_0}(\kappa_i)$, $i=1,\cdots,n$. We first prove that
	\begin{equation}
		\label{eq:q_sharp}
		q_t^{-} (\tau) \geq q_{t-1+\beta_0}^{-} (\nu_+),~~ t \in (1-\beta_0, 1].
	\end{equation}
	%If $t = 1-\beta_0$, it trivially holds because the right hand side equals $-\infty$. Then we fix $t \in (1-\beta_0, 1]$.
	Fix $t \in (1-\beta_0, 1]$. We have
	$
	\sum_{i=1}^n T_i \id_{B} = \sum_{i=1}^n X_i^* \id_{B}
	$
	and for any $x \in \R$,
	\begin{equation*}
		\label{eq:q_proof_sharp}
		\begin{aligned}
			\tau(x, \infty) = \p\left(\sum_{i=1}^n T_i > x\right) \geq  \p\left(\sum_{i=1}^n X_i^* > x, B\right)
			&\geq \p\left(\sum_{i=1}^n X_i^* > x \right) - \p(B^c )\\&
			\geq \p\left(\sum_{i=1}^n X_i^* > x \right) - 1+\beta_0 = \nu_+(x, \infty)-1+\beta_0.
		\end{aligned}
	\end{equation*}
	For any $x < q_{t-1+\beta_0}^{-} (\nu_+)$, we have
	$\nu_+(-\infty, x] < t-1+\beta_0$, and $\tau(-\infty, x] \leq \nu_+(-\infty, x] +1-\beta_0 < t$ and then $x < q_{t}^{-} (\tau)$. Hence we have
	$q_t^{-} (\tau) \geq q_{t-1+\beta_0}^{-} (\nu_+)$ and prove \eqref{eq:q_sharp}.	
	Thus, it follows from sharpness of \eqref{eq:thm1_sharp} that
	\begin{equation}
		\label{eq:RVaR_sharp}
		\begin{aligned}
			\sum_{i=1}^n R_{\beta_i, \beta_0}(\mu_i) = \sum_{i=1}^n \ES_{\beta_0} (\kappa_i) &\geq \ES_{\beta_0} (\tau)\\
			&= \frac{1}{\beta_0} \int_{1-\beta_0}^{1} q_t^{-} (\tau) \d t\\
			&\geq \frac{1}{\beta_0} \int_{1-\beta_0}^{1} q_{t-1+\beta_0}^{-} (\nu_+) \d t\\
			&= \frac{1}{\beta_0} \int_{0}^{\beta_0} q_{t}^{-} (\nu_+) \d t
			= R_{1-\beta_0, \beta_0} (\nu_+)
			= \sum_{i=1}^n R_{\beta_i, \beta_0}(\mu_i),
		\end{aligned}
	\end{equation}
	where the first inequality is the well-known  subadditivity of $\ES_{\beta_0}$. Thus, all inequalities in \eqref{eq:RVaR_sharp} are sharp.
	
	The fact that the first inequality in \eqref{eq:RVaR_sharp} is sharp implies that $T_i$, $i=1,\cdots,n$ share the same tail event with probability $\beta_0$ according to Theorem 5 in \cite{WZ21}, i.e., the sets
	$
	\{ T_i \geq q_{1-\beta_0}^- (\kappa_i) \} =
	\{ q_{1-\beta_0-\beta_i}^{-}(\mu_i) \leq X_i^* \leq q_{1-\beta_i}^{-}(\mu_i) \}, ~ i =1, \cdots, n,
	$ (also in \eqref{eq:set_A^c}) are identical and have probability $\beta_0$.
	We denote this set by $B'$. Furthermore, $B'$ does not intersect any $A_i$, $i = 1,\cdots,n$.
	
	We write $Y_i = X_i^*|_{B'}$. Hence $Y_i \sim \mu^{[1-\beta_0-\beta_i, 1-\beta_i]}$ and $\sum_{i=1}^n Y_i = \sum_{i=1}^n X_i^*$ on the set $B'$ and $q_0^+ \left(\sum_{i=1}^n {Y_i}\right) = \E\left[\sum_{i=1}^n  Y_i\right] =R^{+}_{\boldsymbol \beta}(\boldsymbol \mu)$. Thus, $\sum_{i=1}^n Y_i = R^{+}_{\boldsymbol \beta}(\boldsymbol \mu)$ almost surely.
	
	\noindent
	\textbf{Step 3:} We proceed to show that the events $A_i$, $i=1,\cdots,n$, are mutually disjoint. We can calculate
	$$
	%\begin{align*}
	\frac{\partial }{\partial \beta'_i } R^{+}_{\boldsymbol \beta'}(\boldsymbol \mu)= \frac{1}{\beta_0' } \left(R^{+}_{\boldsymbol \beta'}(\boldsymbol \mu)  - q^-_{1-\beta'_i} (\mu_i) - \sum_{j\ne i}  q^-_{1-\beta_0' -\beta'_j}(\mu_j)\right), ~ \boldsymbol{\beta'} \in \Delta_n.
	%&=\frac{\partial }{\partial \beta_i }   \sum_{j=1}^n \frac{1}{1-\beta }\int_{\beta-\beta_j} ^{1-\beta_j } q^-_t(\mu_i)\d t
	%\\&= \frac{1}{(1-\beta)^2}  \sum_{j=1}^n \int_{\beta-\beta_j} ^{1-\beta_j } q^-_t(\mu_i)\d t + \frac{1}{1-\beta } \left( \frac{\partial }{\partial \beta_i }  \sum_{j=1}^n\int_{\beta-\beta_j} ^{1-\beta_j } q^-_t(\mu_i)\d t \right)
	%\\&= \frac{1}{1-\beta } \bigg(R^{+}_{\boldsymbol \beta}(\mu_1,\dots,\mu_n)  - q^-_{1-\beta_i} (\mu_i) - \sum_{j\ne i}  q^-_{\beta -\beta_j}(\mu_j)\bigg).
	%\\&= \frac{1}{1-\beta } \bigg(R^{+}_{\boldsymbol \beta}(\mu_1,\dots,\mu_n)  -  \sum_{j=1}^n   q^-_{\beta-\beta_j}(\mu_j) - q^-_{1-\beta_i} (\mu_i) + q^-_{\beta-\beta_i}(\mu_i) \bigg).
	%\end{align*}
	$$
	The first-order condition from the optimality of $\boldsymbol{\beta}$ reads as
	\begin{equation}
		\left\{
		\begin{aligned}
			\label{eq:foc}
			& R_{\boldsymbol \beta}^+ (\boldsymbol \mu)  - q^-_{1-\beta_i} (\mu_i) - \sum_{j\ne i}  q^-_{1-\beta_0 -\beta_j}(\mu_j) = 0,~ \mbox{if $\beta_0 > 0$ and $i \in \{1,\cdots,n\}$ satisfying $\beta_i \neq 0$;}\\
			& R_{\boldsymbol \beta}^+ (\boldsymbol \mu)  - q^-_{1} (\mu_i) - \sum_{j\ne i}  q^-_{1-\beta_0 -\beta_j}(\mu_j) \ge 0,~ \mbox{if $\beta_0 > 0$ and $i \in \{1,\cdots,n\}$ satisfying $\beta_i=0$;}\\
			& R_{\boldsymbol \beta}^+ (\boldsymbol \mu) - \sum_{j=1}^n q_{1-\beta_j}^- (\mu_j) = 0, ~ \mbox{if $\beta_0 = 0$.}
		\end{aligned}
		\right.
	\end{equation}
	Denote the sets (the ``left-tail" parts of $X_1, \cdots, X_n$) by
	$
	C_i = \{ X_i < q_{1-\beta_0-\beta_i}^-(\mu_i) \}, ~ i = 1,\cdots, n.
	$
	We have a partition $\Omega = A_i \cup B' \cup C_i$ and $\p(C_i)=1-\beta_0-\beta_i$, $i=1,\cdots,n$. \eqref{eq:foc} shows that $\p\left(\cap_{j=1}^n C_j\right)=0$ because for any $\omega \in \cap_{j=1}^n C_j$, for any fixed $i \in \{1,\cdots, n\}$,
	$$
	\sum_{j=1}^n X_j^*(\omega) <  q^-_{1-\beta_i} (\mu_j) + \sum_{j\ne i}  q^-_{1-\beta_0 -\beta_j}(\mu_j) \leq R_{\boldsymbol \beta}^+ (\boldsymbol \mu) = q_0^+\left(\sum_{j=1}^n X_j^*\right).
	$$
	Arguing by contradiction that there exists $1 \leq k < l \le n$ such that $\p(A_k \cap A_l) > 0$. For any fixed $i \in \{1,\cdots,n\} \backslash \{k,l\}$, we have
	\begin{align*}
		\p(C_i)&= \p\left(\cap_{j=1}^n C_j\right) + \p\left( \cup_{j \neq i} (C_i \cap A_j) \right)\\
		&= \p\left(\cap_{j=1}^n C_j\right) + \p\left( \cup_{j \neq i,k,l} (C_i \cap A_j) \cup (C_i \cap A_k \cap A_l^c) \cup (C_i \cap A_k^c \cap A_l) \cup (C_i \cap A_k \cap A_l) \right)\\
		&\leq \p\left(\cap_{j=1}^n C_j\right) + \sum_{j \neq i,k,l}\p(A_j) + \p(A_k \cap A_l^c) + \p(A_k \cap A_l^c) + \p(A_k \cap A_l)	\\
		& = \p\left(\cap_{j=1}^n C_j\right) + \sum_{j \neq i,k,l}\p(A_j) + \p(A_k) + \p(A_l) - \p(A_k \cap A_l) \\
		& =  \p\left(\cap_{j=1}^n C_j\right) + 1-\beta_0-\beta_i-\p(A_k \cap A_l) \\
		& = \p\left(\cap_{j=1}^n C_j\right) + \p(C_i) -\p(A_k \cap A_l).
	\end{align*}
	Hence $\p\left(\cap_{i=1}^n C_i\right) \geq \p(A_k \cap A_l) >0$, which leads to a contradiction.
	Thus, $A_1, \ldots, A_n$ are mutually disjoint and $\p(B^c) = \p(\cup_{i=1}^n A_i) = \sum_{i=1}^n \p(A_i) = 1-\beta_0$. As the set $B'$ does not intersect $B^c$ and $\p(B') = \beta_0$, we know $B'= B$ and thus the partition $\Omega = A_1 \cup \cdots \cup A_n \cup B$. This completes the first statement in the theorem on the  properties  of $X_i^*$ in \eqref{eq:opt_struc}.

	\noindent
	\textbf{Step 4}: If $\beta_0 = 1$, we have that $\boldsymbol{\mu}$ is jointly mixable. If $\beta_0 < 1$, we check that the corresponding $X_i^*$ given by \eqref{eq:opt_explic} has distribution $\mu_i$, $i=1,\cdots,n$. For each $i=1,\cdots,n$, if $x < q_{1-\beta_0-\beta_i}^{-}(\mu_i)$, we have
	$$
	\begin{aligned}
		\p(X_i^* \leq x) &= \p(U<1-\beta_0, K \neq i, q_{\frac{1-\beta_0-\beta_i}{1-\beta_0}U}^{-}(\mu_i)\leq x)
		\\&
		= \p(K \neq i) \p(q_{\frac{1-\beta_0-\beta_i}{1-\beta_0}U}^{-}(\mu_i)\leq x)\\
		&= \frac{1-\beta_0-\beta_i}{1-\beta_0}  \frac{1-\beta_0}{1-\beta_0-\beta_i} \mu_i(-\infty, x]
		= \mu_i(-\infty, x].
	\end{aligned}
	$$
	One can similarly check that $\p(X_i^* > x) = \mu_i(x, \infty)$ if $x > q_{1-\beta_i}^{-}(\mu_i)$ and $\p(X_i^* \leq x) =\mu_i(-\infty, x]$ if $q_{1-\beta_0-\beta_i}^{-}(\mu_i) \leq x \leq \mu_{1-\beta_i}^{-}(\mu_i)$.
	Hence $X_i^* \sim \mu_i$, $i=1,\cdots,n$.

	\noindent
	\textbf{Step 5}: We finally show that if $\beta_0 < 1$, $\beta_1,\cdots,\beta_n > 0$   and the minimum of each of the functions $h_1, \cdots, h_n$ is attained at $x=1-\beta_0$, then $(X_1^*,\dots,X_n^*)$ in \eqref{eq:opt_explic} attains the maximum of $q_0^+$ for $\boldsymbol \mu$.
	According to the first-order condition \eqref{eq:foc}, we have
	$
	h_1(1-\beta_0)=\cdots=h_n(1-\beta_0)=R_{\boldsymbol{\beta}}^+(\boldsymbol{\mu}).
	$
	For all $i = 1, \cdots, n$, we have $h_i(x) \geq R_{\boldsymbol{\beta}}^+(\boldsymbol{\mu})$ for all $x \in (0, 1-\beta_0]$, i.e.,
	$
	\sum_{i=1}^n X_i^* \geq R_{\boldsymbol{\beta}}^+(\boldsymbol{\mu})$ almost surely on $\{U \in [0,1-\beta_0)\}.
	$
	Since
	$\sum_{i=1}^n X_i^* = \sum_{i=1}^n Y_i = R_{\boldsymbol{\beta}}^+(\boldsymbol{\mu})$ on $\{ U \in [1-\beta_0, 1] \}$, we have
	$q_0^+ \left(\sum_{i=1}^n X_i^* \right) = \max_{\nu \in \Lambda(\boldsymbol \mu) }q_0^+ (\nu) = R^{+}_{\boldsymbol \beta}(\boldsymbol \mu).$%\Halmos
	%\endproof
\end{proof}

\subsection{Proofs in Section \ref{sec:7}}
\begin{proof}[Proof of Proposition \ref{prop:n_dual}]
	%\proof{Proof of Proposition \ref{prop:n_dual}.}	
	Theorem 4.17 of \cite{R13} gives
	\begin{eqnarray}
		\label{eq:heter_dual}
		\inf_{\nu  \in \Lambda(\boldsymbol \mu)} \nu(-\infty, s]  \geq 1 - D_n(s).
	\end{eqnarray}
	Standard argument inverting \eqref{eq:heter_dual} gives \eqref{eq:n_dim_dual_bound}. %\Halmos
\end{proof}

\begin{proof}[Proof of Theorem \ref{th:qa-6}]
	%\proof{Proof of Theorem \ref{th:qa-6}.}
	\quad
	\begin{enumerate}
		\item For fixed $t \in [0,1)$, denote by $x_1 = R^+_{\boldsymbol{\beta}}(\boldsymbol{\mu})$ the right-hand side \eqref{eq:main1pr}. We proceed to show $D_n^{-1}(1-t) \leq x_1$ and thus the dual bound \eqref{eq:n_dim_dual_bound} is not greater than the convolution bound.
		
		\textbf{Case 1}: if the infimum in \eqref{eq:main1pr} is attained at $\boldsymbol \beta = (\beta_0,\beta_1, \cdots ,\beta_n) \in (1-t){\Delta} _n$ with $\beta_0, \cdots, \beta_n >0$,  the first-order condition reads as the first equation in \eqref{eq:foc}.
		%, for any $i = 1, \cdots, n$,$$q_{1-\beta_i}^- (\mu_i) + \sum_{j\ne i} q_{1-\beta_0 - \beta_j}^- (\mu_j) = \frac{1}{\beta_0}\sum_{j=1}^n \int_{1-\beta_0-\beta_j}^{1-\beta_j}q_u^- (\mu_j) \d u= x_1.$$
		Because $\beta_0 >0$ and $1-\beta_0-\beta_i < 1-\beta_i$, we have $q_{1-\beta_0-\beta_i}^- (\mu_i) < q_{1-\beta_i}^- (\mu_i)$. Define $r_i = q_{1-\beta_0-\beta_i}^- (\mu_i)$ for $i = 1,\cdots,n$. One can check from the first-order condition that $r_i = q_{1-\beta_i}^- (\mu_i)+r-x_1$, and hence $r < x_1 $ and $\mathbf{r} \in \Delta_n(x_1)$. We have
		$$
		\small
		\begin{aligned}
			x_1 = \frac{1}{\beta_0}\sum_{j=1}^n \int_{1-\beta_0-\beta_j}^{1-\beta_j}q_u^- (\mu_j) \d u
			&= \frac{1}{\beta_0}\sum_{j=1}^n \int_{q_{1-\beta_0-\beta_j}^- (\mu_j)}^{q_{1-\beta_j}^- (\mu_j)} y \mu_j(\d y)\\
			&= \frac{1}{\beta_0}\sum_{j=1}^n \int_{r_j}^{x_1-r+r_j} y  \mu_j(\d y) \\
			&= \frac{1}{\beta_0}\sum_{j=1}^n \left( (x_1-r+r_j)(1-\beta_j) - r_j (1-\beta_0-\beta_j) - \int_{r_j}^{x_1-r+r_j} \mu_j(-\infty, y] \d y  \right) \\
			&= \frac{1}{\beta_0}\sum_{j=1}^n \left((x_1-r)(1-\beta_j) + r_j \beta_0  -(x_1-r) + \int_{r_j}^{x_1-r+r_j} \mu_j( y, \infty) \d y  \right) \\
			&= \frac{1}{\beta_0} \left( x_1 \beta_0 - (1-t)(x_1-r) + \sum_{j=1}^n \int_{r_j}^{x_1-r+r_j} \mu_j(y, \infty) \d y \right).
		\end{aligned}
		$$
		It follows that
		$
		1-t = \sum_{j=1}^n \frac{1}{x_1 - r} \int_{r_j}^{x_1-r+r_j} \mu_j (y, \infty) \d y \geq D_n(x_1).
		$
		
		\textbf{Case 2}: Suppose that the infimum in \eqref{eq:main1pr} is attained at $\boldsymbol \beta$ with some $\beta_i= 0$ for some $i \in \{1,\cdots,n\}$ and $\beta_0=1-t-\sum_{i=1}^{n}\beta_i > 0$. %then the first-order condition reads as %: for $i \in \{1,\cdots,n\}$ satisfying $\beta_i \neq 0$,
		%	$$q_{1-\beta_i}^- (\mu_i) + \sum_{j\ne i} q_{1-\beta_0-\beta_i}^- (\mu_j) = \frac{1}{\beta_0}\sum_{j=1}^n \int_{1-\beta_0-\beta_j}^{1-\beta_j}q_u^- (\mu_j) \d u= x_1;$$
		%	for $i \in \{1,\cdots,n\}$ satisfying $\beta_i=0$,
		%	$$
		%	q_{1}^- (\mu_i) + \sum_{j\ne i} q_{1-\beta_0-\beta_j}^- (\mu_j) \le  \frac{1}{\beta_0}\sum_{j=1}^n \int_{1-\beta_0-\beta_j}^{1-\beta_j}q_u^- (\mu_j) \d u= x_1.
		%	$$
		For $i=1,\cdots,n$, we have $q_{1}^- (\mu_i) > q_{t}^- (\mu_i)$ because $t < 1$ and define $r_i = q_{1-\beta_0-\beta_i}^- (\mu_i) < q_1^- (\mu_i)$.
		For $i \in \{1,\cdots,n\}$ satisfying $\beta_i \neq 0$, the first-order condition reads as the first equation in \eqref{eq:foc} and gives that $r_i = q_{1-\beta_i}^- (\mu_i)+r-x_1$.
		For $i$ satisfying $\beta_i=0$, the first-order condition reads as the second equation in \eqref{eq:foc} and gives $q_{1}^- (\mu_i) \leq x_1 -r + r_i$ and $r \leq x_1 - (q_{1}^- (\mu_i) - r_i) < x_1$, which implies $\mu_i(-\infty, x_1- r + r_i]=1$ and $\mathbf{r} \in \Delta_n(x_1)$.
		Similarly,
		%	$$
		%	\begin{aligned}
		%	s_1 &= \frac{1}{1-t}\sum_{j=1}^n \int_{t}^{1}q_u^- (\mu_j) \d u
		%	= \frac{1}{1-t}\sum_{j=1}^n \int_{q_{t}^- (\mu_j)}^{q_{1}^- (\mu_j)} x \d \mu_j((-\infty, x]) \\
		%	&= \frac{1}{1-t}\sum_{j=1}^n \int_{r_j}^{q_1^- (\mu_j)} x \d \mu_j((-\infty, x]) = \frac{1}{1-t}\sum_{j=1}^n \int_{r_j}^{s_1-r+r_j} x \d \mu_j((-\infty, x]) \\
		%%	&= \frac{1}{1-t}\sum_{j=1}^n \big( (s_1-r+r_j) - r_j t - \int_{r_j}^{s_1-r+r_j} \mu_j((-\infty, x]) \d x  \big) \\
		%%	&= \frac{1}{1-t}\sum_{j=1}^n \big( (1-t)r_j + \int_{r_j}^{s_1-r+r_j} \mu_j((x, \infty) \d x  \big) \\
		%	&= \frac{1}{1-t} \big( (1-t)r + \sum_{j=1}^n \int_{r_j}^{s_1-r+r_j} \mu_j((x, \infty) \d x  \big).
		%	\end{aligned}
		%	$$
		$$
		x_1 = \frac{1}{\beta_0} \left(x_1 \beta_0 - (1-t)(x_1-r) + \sum_{j=1}^n \int_{r_j}^{x_1-r+r_j} \mu_j(y, \infty) \d y  \right).
		$$
		Therefore,
		$$
		\begin{aligned}
			1-t = \frac{1}{x_1 - r} \sum_{j=1}^n \int_{r_j}^{x_1-r+r_j} \mu_j(y, \infty) \d y \geq D_n(x_1).
		\end{aligned}
		$$
		
		\textbf{Case 3}: If the infimum in \eqref{eq:main1pr} is attained at some $\boldsymbol \beta$ with $\beta_0=1-t-\sum_{i=1}^n \beta_i=0$, then from \eqref{eq:main1pr} we have the third equation in \eqref{eq:foc}.
		%	$$
		%	\sum_{j=1}^n q_{1-\beta_j}^- (\mu_j) = x_1.
		%	$$
		Define $r_i = q_{1-\beta_0-\beta_i}^- (\mu_i)$, $i=1,\cdots,n$. Then $r=\sum_{i=1}^n r_i = x_1$ and
		\begin{align*}
			1-t = \sum_{i=1}^n \beta_i \!=\! \sum_{i=1}^n \mu_i (r_i, +\infty)=\lim_{\begin{subarray}{c}\mathbf{r'} \in \Delta_n(x_1)\\ \mathbf{r'} \rightarrow \mathbf{r}\end{subarray}} \frac{1}{x_1-r'} \sum_{i=1}^n \int_{r_i'}^{x_1-r'+r_i'}  \mu_i(y,\infty) \d y  \geq  D_n(x_1).
		\end{align*}
		In all three cases,
		$1-t \geq  D_n(x_1)$.
		Since $D_n $ is decreasing,  $D_n^{-1}(1-t) \leq x_1$, and thus the dual bound is not greater than the convolution bound.
		\item For fixed $t \in [0,1)$, we proceed to show that the dual bound $D_n^{-1}(1-t)$ is not smaller than the convolution bound.
		We first claim that if quantile functions of $\mu_1,\cdots,\mu_n$ are continuous, then $D_n $ is strictly decreasing on $\left(-\infty, \sum_{i=1}^n q_{1}^- (\mu_j)\right)$ and is constant 0 on $\left[\sum_{i=1}^n q_{1}^- (\mu_j),\infty\right)$.
		Indeed, for any $x_1 < x_2$, we have
		\begin{align*}
			\small
			D_n(x_1) &= \inf\limits_{\mathbf{r} \in \Delta_n(x_1)} \left\{ \sum_{i=1}^n \frac{1}{x_1-r} \int_{r_i}^{x_1-r+r_i} \mu_i(y, \infty) \d y \right\}\\
			& \geq \inf_{\mathbf{r} \in \Delta_n(x_2)} \left\{ \sum_{i=1}^n \frac{1}{x_1-r} \int_{r_i}^{x_1-r+r_i} \mu_i(y, \infty) \d y \right\} \\
			& \geq \inf_{\mathbf{r} \in \Delta_n(x_2)} \left\{ \sum_{i=1}^n \frac{1}{x_2-r} \int_{r_i}^{x_2-r+r_i} \mu_i(y, \infty) \d y \right\} = D_n(x_2).
		\end{align*}
		We prove that if ``$=$" holds, it must be $D_n(x_1) = D_n(x_2) = 0$. Since $D_n(x_1) = D_n(x_2) \in [0,n]$ is bounded,   the infimum is attained at some $\mathbf{r}$ with $r \leq x_1$. Because $\mu_j(\cdot,\infty)$ is decreasing, we have for $j=1,\cdots,n$, $\mu_j(\cdot,\infty)$ is a constant on $[r_j, x_2-r+r_j]$. The fact that quantile functions of $\mu_1,\cdots,\mu_n$ are continuous implies that these constants can only be 0 or 1, and they cannot be 1 since it is an infimum. Hence for $j=1,\cdots,n$, $\mu_j(y,\infty) \equiv 0$ for $y\in[r_j, x_2-r+r_j]$, which implies $D_n(x_1) = D_n(x_2) = 0$. It is straightforward to check that $D_n(x) = 0$ implies $x \geq \sum_{i=1}^n q_{1}^- (\mu_j)$. Thus we prove the claim. One can further verify that $D_n $ is continuous.

		Now we continue to prove the main result. For fixed $t \in [0,1)$, we have $D_n^{-1}(1-t) < \sum_{i=1}^n q_{1}^- (\mu_j)$ and $D_n(D_n^{-1}(1-t)) >0$.
		As $D_n $ is strictly decreasing and continuous on $\left(-\infty, \sum_{i=1}^n q_{1}^- (\mu_j)\right)$, we have $D_n(D_n^{-1}(1-t))=1-t$. Denote  by $x_2= D_n^{-1}(1-t)$   the value of the dual bound.
		
		\textbf{Case 1}: Suppose that the infimum of $D_n(x_2)$ is attained at $\mathbf{r} = (r_1, \cdots, r_n) \in \Delta_n(x_2)$. Its first-order condition reads as, for any $i=1,\cdots,n$,
		$$
		\mu_i(r_i, \infty) + \sum_{j\ne i} \mu_j(x_2-r+r_j, \infty) = \frac{1}{x_2 - r} \sum_{j=1}^n \int_{r_j}^{x_2-r+r_j} \mu_j(y, \infty) \d y = D_n(x_2) = 1-t.
		$$
		Define $\beta_i = \mu_i (x_2-r+r_i, \infty)$, $i=1,\cdots,n$ and $\beta_0 = 1-t-\sum_{i=1}^n \beta_i$. One can check $\beta_i = 1-\beta_0 - \mu_i(-\infty, r_i]$ and $\boldsymbol{\beta} \in (1-t)\Delta_n$ because $r < x_2$. We have
		\begin{align*}
			1-t &= \frac{1}{x_2-r} \sum_{j=1}^n \int_{r_j}^{x_2-r+r_j} \mu_j(y, \infty) \d y
			\\&= \frac{1}{x_2 - r} \sum_{j=1}^n \int_{1-\beta_0-\beta_j}^{1-\beta_j} (1-u) \d q_u^- (\mu_j)\\
			&= \frac{1}{x_2 - r} \sum_{j=1}^n \left( (x_2-r+r_j) \beta_j - r_j (\beta_0+\beta_j) + \int_{1-\beta_0-\beta_j}^{1-\beta_j} q_u^- (\mu_j) \d u \right)\\
			&= \frac{1}{x_2 - r} \left( (x_2-r)(1-t)-x_2\beta_0 + \sum_{j=1}^n  \int_{1-\beta_0-\beta_j}^{1-\beta_j} q_u^- (\mu_j) \d u \right).
		\end{align*}
		Therefore,
		$$
		x_2 = \frac{1}{\beta_0} \sum_{j=1}^n \int_{1-\beta_0-\beta_j}^{1-\beta_j} q_u^- (\mu_j) \d u,
		$$
		which implies that the value of the dual bound $x_2$ is not smaller than that of the convolution bound.
		
		\textbf{Case 2}: If the infimum of $D_n(x_2)$ is attained at some $\mathbf{r}$ with $r = x_2$, then
		$$
		D_n(x_2) = \lim_{\begin{subarray}{c}\mathbf{r'} \in \Delta_n(x_2)\\ \mathbf{r'} \rightarrow \mathbf{r}\end{subarray}} \sum_{i=1}^n \frac{1}{x_2-r'}\int_{r_i'}^{x_2-r'+r_i'} \mu_i (y, \infty)  \d y = \sum_{i=1}^n \mu_i  (r_i, \infty)  = 1-t.
		$$
		Define $\beta_i = \mu_i  (r_i, \infty) $, $i=1,\cdots,n$ and $\beta_0 = 0$. We have $r_i = q_{1-\beta_i}^- (\mu_i)$, and
		$$
		\begin{aligned}
			x_2 = \sum_{i=1}^n r_i = \sum_{i=1}^n q_{1-\beta_i}^- (\mu_i) &= \lim_{\begin{subarray}{c} \boldsymbol{\beta'} \in \Delta_n\\ \boldsymbol{\beta'} \rightarrow \boldsymbol{\beta}\end{subarray}} \sum_{i=1}^n \frac{1}{\beta_0'} \int_{1-\beta_0'-\beta_i'}^{1-\beta_i'}q_u^- (\mu_i)\d u\
			\geq \inf_{\boldsymbol \beta' \in \Delta_n} \sum_{i=1}^n \frac{1}{\beta_0'} \int_{1-\beta_0'-\beta_i'}^{1-\beta_i'}q_u^- (\mu_i)\d u,
		\end{aligned}
		$$
		which implies that the value of the dual bound $x_2$ is not smaller than that of the convolution bound.  The statement on the correspondence  is shown in the above steps.  %\Halmos
		\qedhere
	\end{enumerate}
	%\endproof
\end{proof}

%
%\subsection{Proofs in Section \ref{sec:lower}}
%
%\begin{proof}[Proof of Theorem \ref{th:qa-4primeprime}]
%	The proof is symmetric to that of Theorem \ref{th:qa-4prime}.
%	%We have shown \eqref{eq:prime1} and its sharpness results. For those of \eqref{eq:prime2}, the proof is based on symmetry between the two inequalities in the following. For any random variables $X_1,\cdots,X_n$ satisfying $X_i \sim \mu_i$, $i=1,\cdots,n$ and $\sum_{i=1}^n X_i \sim \nu$, denote by $\tilde{\mu}_i$ the distribution measure of $-X_i$, $i=1,\cdots,n$ and by $\tilde{\nu}$ that of $-\sum_{i=1}^n X_i$. Substituting $\tilde{\nu}, \tilde{\mu}_1, \cdots, \tilde{\mu}_n$ into \eqref{eq:prime1} and noting that $R_{d,s}(\nu) = -R_{1-d-s,s}(\tilde{\nu})$ and $R_{1-\beta_i-\beta_0,\beta_0}(\mu_i) = -R_{\beta_i,\beta_0}(\tilde{\mu}_i)$, we prove \eqref{eq:prime2}. The sharpness results of \eqref{eq:prime2} can be proved similarly to Theorem \ref{th:qa-4primeprime}.
%\end{proof}
%\begin{proof}[Proof of Proposition \ref{prop:basic_sym}]
%	The proof is symmetric to that of Proposition \ref{prop:basic}.
%\end{proof}
%\begin{proof}[Proof of Theorem \ref{th:qa-2}]
%	The proof is symmetric to that of Theorem \ref{th:qa-4}.
%\end{proof}
%
%\begin{proof}[Proof of Proposition \ref{prop:reduced_bound_sym}]
%	The proof is symmetric to that of Proposition \ref{prop:reduced_bound}.
%\end{proof}
%\begin{proof}[Proof of Proposition \ref{prop:nu_sym}]
%	The proof is symmetric to that of Proposition \ref{prop:nu}.
%\end{proof}
%

\section{Counter-examples}
\label{app:A}

\normalsize

\begin{example}[Non-sharpness in Theorem \ref{th:qa-4}] \label{ex:ex1}
	Without loss of generality, we consider the case $t=0$. Let $\mu$ be a bi-atomic uniform distribution on $\{-1,1\}$.
	It is easy to see that
	$\sup_{\nu \in \Lambda_3(\mu) } q_0^+ (\nu) =-1$ since any $\nu\in\Lambda_3(\mu) $ is supported in $\{-3,-1,1,3\}$ with mean $0$.
	On the other hand,   for $(\beta_0,\beta_1,\beta_2,\beta_3)\in \Delta_3$ with $\beta_1\ge \beta_2\ge \beta_3$, by symmetry, and the fact that $R_{1-\beta,\beta-\alpha}$ is increasing in $\alpha$ and increasing in $\beta$,
	we have
	$$  R_{\beta_1, \beta_0}(\mu )= -R_{1-\beta_0-\beta_1, \beta_0}(\mu ) = -R_{\beta_2+\beta_3, \beta_0}(\mu) \ge -R_{\beta_2, \beta_0+\beta_3}(\mu),
	$$
	$$ R_{\beta_2, \beta_0}(\mu ) \ge R_{\beta_2, \beta_0+\beta_3}(\mu ),$$
	and
	$$ R_{\beta_3, \beta_0}(\mu ) \ge R_{\beta_3, \beta_0+\beta_2}(\mu ) \ge R_{\beta_3, 1-2\beta_3}(\mu) = 0.$$
	Combining the above three inequalities, we have
	$
	\sum_{i=1}^3 R_{\beta_i, \beta_0}(\mu )\ge 0.
	$
	Hence,
	$$\sup_{\nu \in \Lambda_3(\mu) } q_0^+ (\nu) =-1 < 0 \le   \inf_{\boldsymbol \beta \in \Delta_n} \sum_{i=1}^3 R_{\beta_i, \beta_0}(\mu ),$$
	showing that \eqref{eq:main1} is not an equality.
\end{example}

\begin{example}[\eqref{eq:main1hom} does not hold as an equality for an increasing density] \label{ex:ex0}
	Without loss of generality, we consider the case $t=0$. Suppose that $\mu\in \M$ has an increasing density on its support.
	Then, the cdf of $\mu$ is convex, and hence the left quantile $q^-_u(\mu)$ is a concave function of $u\in (0,1)$.
	For the concave and increasing function $q^-_u(\mu)$, we have
	$$
	\frac{1}{1-n\alpha} \int_{(n-1)\alpha}^{1-\alpha} q^-_u(\mu)\d u \ge \frac{1}{1-2\alpha} \int_{ \alpha}^{1- \alpha} q^-_u(\mu)\d u
	\ge   \int_{ 0}^1 q^-_u(\mu)\d u.
	$$
	Therefore,
	$$\inf_{\alpha \in (0,\frac1n )} n  R_{\alpha ,1-n \alpha}(\mu)=n R_{0,1}(\mu). $$
	Note that if \eqref{eq:main1hom} holds as an equality, then
	$\inf_{\nu \in \Lambda_n(\mu) } q_1^- (\nu) =n R_{0,1}(\mu)$, which, by Proposition \ref{lem:lem2} below, implies that $\mu$ is $n$-CM.
	There are distributions $\mu$ with a decreasing density that are not $n$-CM, and an equivalent condition is obtained by \cite{WW11}; see Appendix \ref{app:3} for further explanation.
	Therefore, \eqref{eq:main1hom} does not hold as an equality for some distributions with an increasing density. A specific example is shown in Figure \ref{fig:bound_p} (right panel).
\end{example}

\begin{example}[\eqref{eq:main1p} does not hold without a finite mean]\label{ex:cauchy}
	By Theorem 4.2 of \cite{PRWW18}, for standard Cauchy probability measures $\mu_1,\dots,\mu_n$,
	there exists $\nu_1,\nu_2\in \Gamma(\boldsymbol \mu)$ such that
	$$q_0^+(\nu_1)=q_1^-(\nu_1)= \sup_{\boldsymbol \beta \in \Delta_n} \sum_{i=1}^n R_{1-\beta_0-\beta_i, \beta_0}(\mu_i)= -\frac{n\log (n-1)}{\pi}$$
	and
	$$q_0^+(\nu_2)=q_1^-(\nu_2)=  \inf_{\boldsymbol \beta\in \Delta_n}   \sum_{i=1}^n R_{\beta_i, \beta_0}(\mu_i) = \frac{n\log (n-1)}{\pi}.$$
	Hence, we have
	$$
	\begin{aligned}
		\inf_{\nu \in \Lambda(\boldsymbol \mu) } q_1^- (\nu)  = \sup_{\boldsymbol \beta \in \Delta_n} \sum_{i=1}^n R_{1-\beta_0-\beta_i, \beta_0}(\mu_i)
		& = -\frac{n\log (n-1)}{\pi} \\
		& <  \frac{n\log (n-1)}{\pi}
		=  \inf_{\boldsymbol \beta\in \Delta_n}   \sum_{i=1}^n R_{\beta_i, \beta_0}(\mu_i)=  \sup_{\nu \in \Lambda(\boldsymbol \mu) }q_0^+ (\nu).
	\end{aligned}
	$$
\end{example}

\begin{example}[RA fails] \label{ex:RA-fails}
In   Example \ref{ex:RA_not_sharp},
  $\mu$ is a triatomic uniform distribution on $\{1,2,3\}$,
  and the convolution bound gives $\sup_{\nu \in \Lambda( \mu,\mu,\mu) } q_0^+ (\nu) \le 6$ which is attainable. 
  With the initial matrix below, we see that RA does not provide 6, but instead it gives the interval [5,5].
 $$
	\mbox{RA: }~~
	\begin{pmatrix}
		1 & 1 & 1\\
		2 & 2 & 2\\
		3 & 3 & 3
	\end{pmatrix}
 \quad \Longrightarrow \quad 
	\begin{pmatrix}
		3 & 1 & 1\\
		2 & 2 & 2\\
		1 & 3 & 3
	\end{pmatrix}	\quad  \Longrightarrow \quad \mbox{termination}
 $$ 
 For how RA runs, see \cite{EPR13}. 
\end{example}

\section{Well-posedness}\label{app:attain}
Similarly to \eqref{eq:R+}, for $\boldsymbol \mu=(\mu_1,\dots,\mu_n)\in \M^n$ and $\boldsymbol \beta =(\beta_0,\beta_1,\dots,\beta_n)\in \Delta_n$, we denote by
\begin{equation}\label{eq:R-}
	R^{-}_{\boldsymbol \beta}(\boldsymbol \mu) =\sum_{i=1}^n R_{1-\beta_i-\beta_0,\beta_0 }(\mu_i).
\end{equation}
We discuss the attainability of the infimum in $\inf_{\boldsymbol \beta\in \Delta_n} R_{\boldsymbol \beta}^+(\boldsymbol \mu)$ and the supremum in $\sup_{\boldsymbol \beta \in \Delta_n} R^{-}_{\boldsymbol \beta}(\boldsymbol \mu)$. Note that $R_{\boldsymbol \beta}^+(\boldsymbol \mu)$ and $R_{\boldsymbol \beta}^-(\boldsymbol \mu)$ are well defined for $\boldsymbol \beta \in \Delta_n$. Now we discuss cases with $\beta_i$ taking boundary values of $0, 1$. We discuss whether $R^{+}_{\boldsymbol \beta}(\boldsymbol \mu)$ and $R^{-}_{\boldsymbol \beta}(\boldsymbol \mu)$ are well defined on the closure $\overline{\Delta}_n$.% for simplicity. The general cases $R^{+}_{\boldsymbol \beta,t}(\boldsymbol \mu)$ and $R^{-}_{\boldsymbol \beta,t}(\boldsymbol \mu)$ of $t \in (0,1)$ can be checked in a similar manner.

\begin{enumerate}
	\item
	For $\boldsymbol{\beta} \in \Delta_n \subset \overline{\Delta}_n$, there is no undefined form ``$\infty - \infty$" in $R^{+}_{\boldsymbol \beta}(\boldsymbol \mu)$ and $R^{-}_{\boldsymbol \beta}(\boldsymbol \mu)$, which are hence always well defined.
	\item
	For $\boldsymbol{\beta} \in \overline{\Delta}_n$ with $\beta_i=0$ for some $i \in \{1,\cdots,n\}$ and $\beta_0 \in (0,1]$, we define $R^{-}_{\boldsymbol \beta}$ and $R^{+}_{\boldsymbol \beta}$ similarly:
	$$
	R^{+}_{\boldsymbol \beta}(\boldsymbol \mu) = \sum_{j \neq i} R_{\beta_j, \beta_0 }(\mu_j) \! + \! R_{0, \beta_0}(\mu_i),~~ R^{-}_{\boldsymbol \beta}(\boldsymbol \mu) = \sum_{j \neq i} R_{1-\beta_j-\beta_0, \beta_0}(\mu_j) \! + \! R_{1-\beta_0, \beta_0}(\mu_i),
	$$
	except ``$\infty - \infty$" cases that the integral of $q_t^{-}(\mu_i)$ at the neighbour of $0$ is negative infinite and that of $q_t^{-}(\mu_j)$ at the neighbour of $1$ is infinite for some $i,j \in \{1,\cdots,n\}$, i.e., $R_{0, \epsilon}(\mu_j) = \infty$ and $R_{1-\epsilon, \epsilon}(\mu_i) = -\infty$ for some $\epsilon \in (0,1)$.  $R^{+}_{\boldsymbol \beta}$ and $R^{-}_{\boldsymbol \beta}$ are always well defined if $\boldsymbol \mu \in \M_1^n$.
	\item
	For $\boldsymbol{\beta} \in \overline{\Delta}_n$ with $\beta_0 = 0$, we define %\com{better check carefully}
	$$
	R^{+}_{\boldsymbol{\beta}}(\boldsymbol{\mu}) = \sum_{i=1}^n q_{1-\beta_i}^{-}(\mu_i), ~~~~ R^{-}_{\boldsymbol{\beta}}(\boldsymbol{\mu}) = \sum_{i=1}^n q_{\beta_i}^{+}(\mu_i),
	$$
	except ``$\infty - \infty$" cases that $q_1^{-}(\mu_i) = \infty$ and $q_0^{-}(\mu_j)=-\infty$ for some $i,j \in \{1,\cdots,n\}$ and $i \neq j$. They are always well defined if $\mu_1, \cdots, \mu_n$ are all bounded from the positive or negative side.
\end{enumerate}	

Because of the continuity of $R_{\beta,\alpha}$ in $\beta,\alpha\in [0,1]$, it can be proved that the infimum of $\inf_{\boldsymbol \beta\in (1-t)\Delta_n} R^{+}_{\boldsymbol \beta}(\boldsymbol \mu)$ of cases $t \in [0,1)$ and the supremum of $\sup_{\boldsymbol \beta \in t\Delta_n} R^{-}_{\boldsymbol \beta}(\boldsymbol \mu)$ of cases $t \in (0,1]$ are attained in the well-defined part of $\overline{\Delta}_n$.

\section{Connection to joint mixability}\label{app:3}

Joint mixability is closely related to quantile aggregation.
The tuple of distributions $\boldsymbol \mu \in \M^n$ is said to be \emph{jointly mixable} (JM, \cite{WPY13})
if $\delta_C\in  \Lambda(\boldsymbol \mu ) $ for some $C\in \R$.
Such $C$ is called a center of $\boldsymbol \mu $.
The name JM means that the marginal distributions $(\mu_1,\dots,\mu_n)$ is \emph{able} to support a \emph{joint mix} dependence (i.e., a random vector with a constant sum). 
Similarly, a probability measure $\mu$  on $\R$  is $n$-\emph{completely mixable} ($n$-CM, \cite{WW11}) if the $n$-tuple $ (\mu,\dots,\mu) $ is JM.
Obviously, if $\boldsymbol \mu \in \M_1^n$ is JM, then its center is unique and equal to the sum of the means of its components.
If $\boldsymbol \mu \in \M^n$ is JM but it is not in $\M_1^n$, then its center may not be unique (\cite{PRWW18}). The determination of joint mixability for a given $\boldsymbol \mu \in \M^n$ is well known to be a challenging problem and analytical results are limited. The main results of this appendix are a sufficient condition on the sharpness of convolution bounds and some conditions on the determination of JM.

%\subsection{Sharpness of convolution bounds}
We first see that JM is a sufficient condition for the bounds in Proposition \ref{prop:qa-1} to be sharp for probability measures with finite means.

\begin{proposition}\label{prop:qa-3}
	If $\boldsymbol \mu \in \M_1^n$ is JM, then the bounds in Proposition \ref{prop:qa-1} are sharp, and their values are equal to the unique center of $\boldsymbol \mu$.
\end{proposition}
\begin{proof}%[Proof of Proposition \ref{prop:qa-3}]
	%\proof{Proof of Proposition \ref{prop:qa-3}.} %See the alphabet and the tebahpla.\Halmos	
	Note that since $\boldsymbol \mu=(\mu_1,\dots,\mu_n)$ is JM, we know $\delta_C\in \Lambda(\boldsymbol \mu) $ where $C=\sum_{i=1}^n R_{0,1}(\mu_i)$.
	Hence, by Proposition \ref{prop:qa-1},
	$$  \inf_{\boldsymbol \beta\in \Delta_n}   R^{+}_{\boldsymbol \beta}(\boldsymbol \mu)\ge \sup_{\nu \in \Lambda(\boldsymbol \mu) }q_0^+ (\nu) \ge q_0^+ (\delta_C)  \ge C \ge  \inf_{\boldsymbol \beta\in \Delta_n}   R^{+}_{\boldsymbol \beta}(\boldsymbol \mu).$$
	The case for \eqref{eq:main2} is similar.%\Halmos
	%\endproof
\end{proof}
Proposition \ref{prop:qa-3} supports Proposition \ref{prop:qa-1} by giving further conditions for the bounds in Proposition \ref{prop:qa-1} to be sharp, which can be checked through existing results on joint mixability in \cite{WW16}. However, unlike Theorem \ref{th:qa-2}, Proposition \ref{prop:qa-3} itself does not offer new ways to calculate quantile aggregation, since the convolution bounds in \eqref{eq:main1} and \eqref{eq:main2} are all trivially equal to the center if  we know $(\mu_1,\dots,\mu_n)\in \M^n_1$ is JM.

%\subsection{Determination of JM}
Next, we look in the converse direction: implications of Theorem \ref{th:qa-4} and Theorem \ref{th:qa-2} on conditions for JM.
Proposition \ref{prop:qa-1} directly implies the following necessary condition for JM, which is also noted by Proposition 3.3  of \cite{PRWW18} with a similar argument.
If $\boldsymbol \mu $ is JM with center $C$, then  $C=q_0^+(\nu_0) =q_1^-(\nu_0)$ for some $\nu_0\in \Lambda(\boldsymbol \mu )$. Hence, $$\inf_{\nu \in \Lambda(\boldsymbol \mu) } q_1^- (\nu)\le C \le\sup_{\nu \in \Lambda(\boldsymbol \mu) } q_0^+ (\nu).$$
Using Proposition \ref{prop:qa-1}, we arrive at (where $R^-_{\boldsymbol{\beta}}(\boldsymbol{\mu})$ is defined at \eqref{eq:R-})
$$
\sup_{\boldsymbol \beta \in \Delta_n} R^{-}_{\boldsymbol \beta}(\boldsymbol \mu ) \le C \le \inf_{\boldsymbol \beta\in \Delta_n}   R^{+}_{\boldsymbol \beta}(\boldsymbol \mu ).
$$
If the means of $\mu_1,\dots,\mu_n$ are finite, then by Proposition \ref{prop:prop1}, we have
$\sup_{\boldsymbol \beta \in \Delta_n} R^{-}_{\boldsymbol \beta}(\boldsymbol \mu )
\ge \inf_{\boldsymbol \beta\in \Delta_n}   R^{+}_{\boldsymbol \beta}(\boldsymbol \mu ).
$
Therefore, a necessary condition for  $\boldsymbol \mu \in \M_1^n$ to be JM  is
$$
\sup_{\boldsymbol \beta \in \Delta_n} R^{-}_{\boldsymbol \beta}(\boldsymbol \mu ) = \inf_{\boldsymbol \beta\in \Delta_n}   R^{+}_{\boldsymbol \beta}(\boldsymbol \mu ).
$$
We summarize the above simple findings in the following proposition.
We use the convention that the closed interval $[a,b]$ is empty if $a>b$.
\begin{proposition}\label{prop:2}
	The possible center $C$  of $\boldsymbol \mu =(\mu_1,\dots,\mu_n)\in \M^n$  satisfies
	\begin{equation}\label{eq:main5}C\in \left[ \sup_{\boldsymbol \beta \in \Delta_n} R^{-}_{\boldsymbol \beta}(\boldsymbol \mu ) ,  \inf_{\boldsymbol \beta\in \Delta_n}   R^{+}_{\boldsymbol \beta}(\boldsymbol \mu )\right].\end{equation}
	In particular, if   $\boldsymbol \mu $ is JM, then
	\begin{equation}\label{eq:main3}
		\sup_{\boldsymbol \beta \in \Delta_n} R^{-}_{\boldsymbol \beta}(\boldsymbol \mu ) \le  \inf_{\boldsymbol \beta\in \Delta_n}   R^{+}_{\boldsymbol \beta}(\boldsymbol \mu ),
	\end{equation}
	and further if $\boldsymbol \mu \in \M_1^n$, then
	\begin{equation}\label{eq:main4}
		\sup_{\boldsymbol \beta \in \Delta_n} R^{-}_{\boldsymbol \beta}(\boldsymbol \mu ) =\sum_{i=1}^n R_{0,1}(\mu_i)= \inf_{\boldsymbol \beta\in \Delta_n}   R^{+}_{\boldsymbol \beta}(\boldsymbol \mu ).
	\end{equation}
\end{proposition}

The set $\Delta_n$ appeared in Proposition \ref{prop:2} may be replaced by $\overline{\Delta}_n$  if $(\mu_1,\dots,\mu_n)\in \M_1^n$.
We next verify that for many classes distributions known in the literature,
\eqref{eq:main3}-\eqref{eq:main4} actually are sufficient for JM, and all centers are identified with Proposition \ref{prop:2}.
We first present a convenient result which is useful for the determination of JM for distributions with finite means.
\begin{proposition} \label{lem:lem2}
	For  $\boldsymbol \mu = (\mu_1,\dots,\mu_n)\in \M_1^n$,
	the following statements are equivalent.
	\begin{enumerate}[(i)]
		\item $\boldsymbol \mu$   is JM.
		\item   $\sup_{\nu \in \Lambda(\boldsymbol \mu) }q_0^+ (\nu) = \sum_{i=1}^n R_{0,1}(\mu_i) .$
		\item $\inf_{\nu \in \Lambda(\boldsymbol \mu) }q_1^- (\nu) = \sum_{i=1}^n R_{0,1}(\mu_i) .$
		\item $\sup_{\nu \in \Lambda(\boldsymbol \mu) }q_0^+ (\nu) = \inf_{\nu \in \Lambda(\boldsymbol \mu) }q_1^- (\nu) .$
	\end{enumerate}
\end{proposition}

\begin{proof}%[Proof of Proposition \ref{lem:lem2}]
	%\proof{Proof of Proposition \ref{lem:lem2}.}
	Let $C=\sum_{i=1}^n R_{0,1}(\mu_i)$. By Proposition \ref{prop:prop1},
	\begin{equation}
		\label{eq:lem2pf1}\sup_{\nu \in \Lambda(\boldsymbol \mu) }q_0^+ (\nu) \le C \le \inf_{\nu \in \Lambda(\boldsymbol \mu) }q_1^- (\nu).
	\end{equation}
	As a consequence, (iv)$\Rightarrow$(ii)-(iii).
	
	If $\boldsymbol \mu$   is JM, then there exists $\nu_0 \in  \Lambda(\boldsymbol \mu)$ such that $q_0^+ (\nu_0) =  \sum_{i=1}^n R_{0,1}(\mu_i) =q_1^- (\nu_0)$.
	This, together with \eqref{eq:lem2pf1},
	shows the implication (i)$\Rightarrow$(ii)-(iv).
	
	If $\sup_{\nu \in \Lambda(\boldsymbol \mu) }q_0^+ (\nu) =C,$
	then, noting that
	$R_{0,1}(\nu)= C$  for all $\nu \in \Lambda(\boldsymbol \mu)$,
	we have  $\delta_C\in \Lambda(\boldsymbol \mu)$, since $ \Lambda(\boldsymbol \mu)$ is closed under weak convergence (Theorem 2.1 of \cite{BJW14}). This shows (ii)$\Rightarrow$(i). Similarly, (iii)$\Rightarrow$(i).%\Halmos
	%\endproof
\end{proof}

Next, in view of Theorem \ref{th:qa-4}, we show in Proposition \ref{prop:prop3} that it can be checked through convolution bounds whether some distributions are JM if they have monotone densities.
\begin{proposition} \label{prop:prop3}
	$\boldsymbol \mu =(\mu_1,\dots,\mu_n)\in \M_1^n$ is JM if and only if    \eqref{eq:main4} holds, in the following cases:
	\begin{enumerate}[(i)]
		\item Each of $\mu_1,\dots,\mu_n$ admits a decreasing density on its support.
		\item Each of $\mu_1,\dots,\mu_n$ admits an increasing density on its support.
		\item $\mu_1,\dots,\mu_n$ are from the same location-scale family with unimodal and symmetric densities on their supports.
	\end{enumerate}
\end{proposition}

\begin{proof}%[Proof of Proposition \ref{prop:prop3}]
	%\proof{Proof of Proposition \ref{prop:prop3}.}
	% Note that \eqref{eq:main3} and \eqref{eq:main4} are equivalent in this case.
	The necessity of \eqref{eq:main4} is stated in Proposition \ref{prop:2}, and hence we only show its sufficiency.
	\begin{enumerate}[(i)]
		\item By Theorem \ref{th:qa-4} and  \eqref{eq:main4}, we know
		$$\sup_{\nu \in \Lambda(\boldsymbol \mu) }q_0^+ (\nu) = \inf_{\boldsymbol \beta\in \Delta_n}   R^{+}_{\boldsymbol \beta}(\boldsymbol \mu)  =\sum_{i=1}^n R_{0,1}(\mu_i) .$$
		By Proposition \ref{lem:lem2} (ii)$\Rightarrow$(i), $\boldsymbol \mu$ is JM.
		\item This is symmetric to (i).
		\item Without loss of generality, we may assume that $\mu_1,\dots,\mu_n$ all have mean zero and they have scale parameters $a_1\ge \dots \ge a_n>0$, respectively. By Corollary 3.6 of \cite{WW16}, we know that $(\mu_1,\dots,\mu_n)$ is JM if and only if $2 \bigvee_{i=1}^n a_i  \le \sum_{i=1}^n a_i$.
		Take $\boldsymbol \beta =(1-\epsilon,\epsilon,0,0,\dots,0)\in \overline{\Delta}_n$ for some $\epsilon\in (0,1)$.
		Since $\mu_1,\dots,\mu_n$ are from the same location-scale family with symmetric densities,
		we have
		$- R_{\epsilon,1-\epsilon}(\mu_i)/a_i = - R_{0,1-\epsilon}(\mu_i)/a_i = R_{0,1-\epsilon}(\mu_1)/a_1>0$ for $i=1,\dots,n$.
		By \eqref{eq:main4},  we have
		\begin{align*}
			0 = \sum_{i=1}^n R_{0,1}(\mu_i)   \ge \sum_{i=1}^n R_{1-\beta_0-\beta_i, \beta_0}(\mu_i)
			&=R_{0,1-\epsilon}(\mu_1)+
			\sum_{i=2}^n R_{\epsilon,1-\epsilon}(\mu_i)\\
			& = \frac{  R_{0,1-\epsilon}(\mu_1)  }{a_1}\left (a_1 - \sum_{i=2}^n a_i\right).
		\end{align*}
		Therefore,  $ a_1 - \sum_{i=2}^n a_i \le 0$, which implies
		$2 \bigvee_{i=1}^n a_i  \le \sum_{i=1}^n a_i$.  %\Halmos%\qedhere
	\end{enumerate}
	%\endproof
\end{proof}

\begin{remark}
	
	By Theorem 3.2 of \cite{WW16}, for $\mu_1,\dots,\mu_n\in \M$ with decreasing densities, $(\mu_1,\dots,\mu_n)$ is JM if and only if
	\begin{equation}\label{eq:jm}
		\bigvee_{i=1}^n \left(q^-_{1} (\mu_i)  - q^+_{0}(\mu_i) \right) \le   \sum_{i=1}^n\left( R_{0,1}(\mu_i) - q_0^+(\mu_i)\right).
	\end{equation}
	We already know that \eqref{eq:main4} is necessary for $(\mu_1,\dots,\mu_n)$ to be JM.
	One can directly check that \eqref{eq:jm} is implied by \eqref{eq:main4}, thus showing the equivalence of \eqref{eq:main4} and \eqref{eq:jm}.
	% Note that if $q_1^-(\mu_i)=\infty$ for some $i=1,\dots,n$, then \eqref{eq:th1-2} fails to hold.
	%Take $\boldsymbol \beta =(\beta_1,\dots,\beta_n)\in \Delta_n$ and write $\beta=\sum_{i=1}^n \beta_i$.
	%By \eqref{eq:main4}, we have, for $\epsilon\in [0,1]$,
	%$$
	% \sum_{i=1}^nR_{[ \epsilon \beta - \epsilon \beta_i  ,1- \epsilon \beta_i  ] }(\mu_i)= R^{+}_{\epsilon \boldsymbol \beta}(\mu_1,\dots,\mu_n)
	%\ge \sum_{i=1}^n R_{[0,1]}(\mu_i).$$
	%As a consequence, the function $f$ defined by
	%$$
	% f(\epsilon)=\sum_{i=1}^n \int_{\epsilon\beta - \epsilon\beta_i }^{1-   \epsilon \beta_i     } q^-_t (\mu_i) \d t - (1-\epsilon \beta)  \sum_{i=1}^nR_{[0,1]}(\mu_i)$$
	% is non-negative on $[0,1]$.
	% Note that $f$ is differentiable in a neighbourhood of $0$, and $f(0)=0$. Therefore,
	% $f'(0)\ge 0$, which means
	% $$
	% \sum_{i=1}^n \left(\beta_i  q_{1} (\mu_i)  + (\beta-\beta_i) q_{0}(\mu_i) \right) -  \beta \sum_{i=1}^nR_{[0,1]}(\mu_i) \le 0 . $$
	% Rearranging terms, we have
	%\begin{equation}\label{eq:th1-1}
	% \sum_{i=1}^n\frac{ \beta_i}{\beta}  \left(q_{1} (\mu_i)  - q_{0}(\mu_i) \right) \le   \sum_{i=1}^n\left( R_{[0,1]}(\mu_i) - q_0^+(\mu_i)\right)  . \end{equation}
	%Taking the supremum of the left-hand side of \eqref{eq:th1-1} over all $\boldsymbol \beta \in \Delta_n$, we arrive at
	%  \eqref{eq:th1-2}. Therefore, \eqref{eq:main4} is sufficient for the joint mixability of $(\mu_1,\dots,\mu_n)$.
\end{remark}

\begin{remark}\label{rem:4}
	A similar situation of Proposition \ref{prop:prop3} is obtained for distributions without the mean: if each of $\mu_1,\dots,\mu_n$ is a standard Cauchy distribution, the set of all centers of $(\mu_1,\dots,\mu_n)$ is precisely given by \eqref{eq:main5}. This statement is based on Example 4.1 and Theorem 4.2 of \cite{PRWW18}. 	
\end{remark}

\section{Comparison of different methods in computation} \label{sec:R2-1}

%{\color{blue}
In this appendix, we compare three potential ways of computing the quantile aggregation problem or its approximations.
%\subsection{Decription of the three methods}
The first two approaches, RA and linear program, require a discretization, whereas the third approach, the convolution bound,  can be applied with either discrete input or functional input.

\begin{enumerate}[(a)]
\item \textbf{Original problem}: Let $[n]:=\{1,\dots,n\}$. Consider the quantile aggregation problem
\begin{equation}\label{eq:R2-1}
	\sup_{\nu \in \Lambda(\boldsymbol \mu)} q_0^+(\nu)=  \sup\{ q_0^+(X_1+\dots+X_n): X_i\sim \mu_i,~i\in [n]\},
\end{equation}
where $\mu_1,\dots,\mu_n$ are distributions on $\R$, with supports bounded from below, and $\boldsymbol \mu=(\mu_1,\dots,\mu_n)$. 
Note that $\mu_1,\dots,\mu_n$ are assumed to have supports bounded from below because otherwise $q_0^+$ may be infinite. 
The probability level $0$ is chosen here without loss of generality, because any risk aggregation problem for $q_t^+$ can be equivalently formulated as one for $q_0^+$ as shown in Proposition 1.  

\item 
\textbf{Discretization}: To tackle problem \eqref{eq:R2-1} numerically for given distributions  $\mu_1,\dots,\mu_n$, a common step is to discretize using their quantiles, that is, to consider a number $m$ (ideally large) and for each $i\in [n]$, a distribution $  \mu_i^m$  over $m$ points  (some may be equal) each with probability $1/m$:
\begin{equation}
\label{eq:z}
z^i_1 = q^+_0(\mu_i) ,~~\dots,~~ z^i_m = q^+_{(m-1)/m}(\mu_i).
\end{equation}
The input values of the discrete problem are these $z^i_j$ for $i\in[n]$, $j\in[m]$. 
This discretization is asymptotically consistent in the following sense: Let $\boldsymbol \mu^m= (\mu_1^m,\dots,\mu_n^m)$. As $m\to \infty$,  since $\mu_i^m \to \mu_i$ weakly, $\mu_i^m\le_{\rm st} \mu_i$ (where $\le_{\rm st}$ stands for stochastic order), and $q^+_0$ is upper semi-continuous, we have
$$
\sup_{\nu \in \Lambda(\boldsymbol \mu^m)} q_0^+(\nu)\to 	\sup_{\nu \in \Lambda(\boldsymbol \mu)} q_0^+(\nu) .
 $$

\item
\textbf{Rearrangement algorithm (RA)}:  
We  first write the input values of the discrete problem into a matrix 
\begin{equation}\label{eq:matrix}
	\begin{pmatrix}
		z^1_1 & z^2_1 & \dots & z^n_1 \\
		\vdots & \vdots & \ddots & \vdots \\
		z^1_m & z^2_m & \dots & z^n_m \\
	\end{pmatrix}.
\end{equation}
The RA, introduced by \cite{EPR13}, tries to maximize the minimal row sum of the matrix resulting from rotating the elements within each column of \eqref{eq:matrix}. Note that rotating elements within each column corresponds to  changing the dependence structure of a discrete random vector while maintaining its marginals, and the row-sum vector corresponds to the distribution of the sum of components of the random vector; see Section G.2 for this problem in a different context.
The problem of finding the maximum of the minimal row sum is NP-hard (e.g., \cite{H15}), but RA can compute a suboptimal answer very quickly, which typically has good accuracy. The theoretical computational complexity of RA is unknown in the literature.
For discrete distributions, as in our setting here, the output of RA is  $\underline{s}_N$, which  is a lower bound on the true value of \eqref{eq:R2-1}.
For continuous distributions, RA outputs an interval $[\underline{s}_N, \bar{s}_N]$ with $\underline{s}_N$ again being a lower bound for the original (continuous) problem, but there is no guarantee for $\bar{s}_N$ to be either a lower bound or an upper bound for the original problem.   
RA is the most popular and standard method in the risk management literature to compute the quantile aggregation problem;  see \cite{EPR13, EPRWB14}.

\item

\textbf{Linear program (LP)}:   
Recall that
 $\Gamma( \boldsymbol \mu)$  represents the set of all distributions on $\R^n$ with marginals $\boldsymbol \mu$. 
The problem \eqref{eq:R2-1} can be equivalently formulated as 
\begin{equation}\label{eq:R4}
\sup_{\nu \in \Lambda(\boldsymbol \mu)} q_0^+(\nu)=\sup\left \{x\in \R: \int_{\R^n} \id_{\{x_1+\dots+x_n\le x \}} \Pi (\d x_1,\dots,\d x_n)=0 \mbox{~for some $\Pi\in \Gamma(\boldsymbol \mu)$}\right\}. 
\end{equation}
Writing $\Pi_{\mathbf k} = \Pi(\{(z^1_{k_1} ,\dots, z^n_{k_n})\})$  for $\mathbf k=(k_1,\dots,k_n)\in [m]^n$, which represents the probability 
$$
\p( X_1 = z^1_{k_1}, \dots, X_n = z^n_{k_n})
$$
for a random vector $(X_1,\dots,X_n)$ with the given marginals. 
The above problem can be formulated as the following program
\begin{equation}
	\label{eq:R2-2} 
	\begin{aligned}
		\sup ~ x & \\
		\mbox{subject to} &  ~~\mbox{$\Pi_{\mathbf k}\in [0,1]$ for $\mathbf k=(k_1,\dots,k_n)\in [m]^n$}
		\\& \sum_{\mathbf k \in [m]^n}  \id_{\{z^1_{k_1} +\dots+ z^n_{k_n}\le x\}} \Pi_{\mathbf k} = 0;  \\
		& \sum_{\mathbf k:k_i=j}\Pi_{\mathbf k}=\frac 1m~~~ \mbox{for each $i\in [n]$ and $j\in [m]$}.
	\end{aligned}
\end{equation}
%The problem \eqref{eq:R2-2} has $m^n+1$ variables and $n\times m+1$ constraints. %,
%and it involves the indicator function, which is non-linear. 
Practically, we need to try to solve for discrete values of $x$. 
The method has two steps. 

\textbf{Step 1}: {We specify a real interval $T$ that covers the range of values for the optimal value of \eqref{eq:R2-2}.  
% $T = \{ x_1, \dots, x_m \}$ of $m$ real numbers which corresponds to our grid of accuracy on $\R$. We solve, 
For each fixed $x\in T$, we define
\begin{equation}
	\label{eq:R2-3} 
	\begin{aligned}
		% {\textbf{(LP)}}
  \quad\quad\quad\quad  V(x) &= \min ~  \sum_{\mathbf k \in [m]^n}  y_{\mathbf{k}} \Pi_{\mathbf k}  \\
		&\quad \mbox{subject to}   ~~\mbox{$\Pi_{\mathbf k}\geq0$ for $\mathbf k=(k_1,\dots,k_n)\in [m]^n$},\\ 
		& \quad\quad\quad\quad \sum_{\mathbf k:k_i=j}\Pi_{\mathbf k}=\frac 1m~~~ \mbox{for each $i\in [n]$ and $j\in [m]$},
	\end{aligned}
\end{equation}
where $y_{\mathbf{k}} = \id_{\{z^1_{k_1} +\dots+ z^n_{k_n}\le x\}}$ are parameters. 
Problem \eqref{eq:R2-3} is a linear program (LP) with $m^n$ variables and $n\times m$ equality constraints. }

\textbf{Step 2}: The optimal value output by the algorithm is given by 
\begin{equation}\label{eq:bisec}
	\sup\{ x \in T: \text{the optimal value $V(x)$ in Problem } \eqref{eq:R2-3} \leq 0\}.
\end{equation} 
For each fixed $x \in T$, this method has $m^n$ variables and $n\times m$ constraints. Note that instead of solving for each $x \in T$, to search for the optimal $x$, a bisection approach can replace the specification of $T$. %\tbl{
The complexity of this problem is discussed in Remark \ref{rem:R3}.
%}
 
\item
\textbf{Convolution bound (CB)}: We directly take the quantile functions of $\mu_1,\dots,\mu_n$ as input, let $t = 0$, and solve 
for    $$B_{\rm conv} =\inf_{\boldsymbol \beta\in (1-t)\Delta_n}   \sum_{i=1}^n R_{\beta_i,\beta_0 }(\mu_i).$$
Each term $R_{\beta_i,\beta_0 }(\mu_i)$ is an integral of the corresponding quantile function.
If the marginal distributions are given as discrete data points, then the input values for CB are the empirical quantile functions. 
Although this minimization is not convex and we do not know the  theoretical  computational complexity of CB, in all   numerical examples we find that it can be computed very fast and accurately. 
\end{enumerate}

%\subsection{Numerical comparison}
Next, we provide numerical results to compare the three methods, RA, LP and CB, described above.
For a comparison, we assume that each marginal distribution $\mu_i$ is uniform on $m$ points, denoted by $z^i_1,\dots,z^i_m$ as in the discussion above.
Note that CB can also take  quantile functions as input, whereas RA and LP can only take discrete input.
These $m$ points are specified in two different ways.
\begin{enumerate}[(i)]
\item They are the values of the quantile functions at different levels as in \eqref{eq:z}.
\item They are randomly sampled from some distributions. 
\end{enumerate}
For this specification of the marginal distributions (there is no discretization involved, as the marginal distributions are themselves discrete), LP produces a true value of \eqref{eq:R2-1}, RA produces a lower bound on \eqref{eq:R2-1}, and   CB produces  an upper bound on \eqref{eq:R2-1}.

 The numerical results are reported in Table  \ref{tab:R2-6}.  
We make the following observations from the results. 
 \begin{enumerate}
 \item All methods become slower when either $m$ or $n$ increases, as expected. 
 \item 
LP produces the true value for \eqref{eq:R2-1}, but it has some drawbacks for implementation. As each LP problem involves $m^n$ variables, 
the applicability seems to be very limited. For $m\geq180$ with $n=3$, or $m\geq20$ with $n=5$, the computation is either unavailable in MATLAB or costs more than 24 hours. In a real risk management problem where loss distributions are typically continuous (such as asset prices or insurance losses), the value of $m$ needs to be relatively large to ensure good approximation (typically at least $10^5$; see \cite{EPR13}).
\item RA is fast in most cases and simple in coding. It can handle  $m=10^6$ and $n=200$ as demonstrated by  \cite{EPR13}. However, it only provides a lower bound,  and sometimes this lower bound may not be close to the true value provided by LP.  There are no theoretical results on the convergence of RA.
\item  CB is much faster than LP, but slower than RA. CB can handle dimensions up to $n=200$. 
In some cases studied in this paper, such as continuous distributions with monotone densities, it is theoretically proved that it produces an exact true value. 
In  our numerical results, the CB value often coincides with, or is very close to that of  LP.  It also enjoys the interpretability  of the dependence structure (see Section 6). 
It can directly handle continuous distributions often encountered in risk management without the need to discretize (for such setting, $m$ is practically infinity).  
We comment on two disadvantages of CB. First, in the computation of CB, the convergence is based on the optimization function in the software (\texttt{fmincon} in MATLAB in our case). There is no theoretical guarantee that this function finds the global optimum, due to lack of convexity. Nevertheless, in all numerical results where CB and LP agree, we know that global optimum is reached. 
Second, in case the conditions in Theorem 2 do not hold, we do not know whether CB is equal to the original problem \eqref{eq:R2-1}.  {Despite  the gap due to the non-convexity in computing CB and  the gap between CB and \eqref{eq:R2-1}, if we obtain a solution $\boldsymbol{\beta}$ from the optimization software, the objective value evaluated at $\boldsymbol{\beta}$ is guaranteed to give an upper bound for \eqref{eq:R2-1}.}

\item Combining RA and CB gives a theoretically proven interval in which the true value of \eqref{eq:R2-1} lies. This gives a fast and reliable way of finding the range of \eqref{eq:R2-1}. {While RA is practically fast and commonly used, it only provides a one-sided bound (and arguably the less important side), and CB essentially  closes the other side.}
 \end{enumerate} 
% ($m > 160$ for $n = 3$, ``{Requested array exceeds the maximum possible variable size}" in \texttt{intlinprog} in MATLAB).   

\begin{table}[htpb]\centering
%\color{blue}
\setlength{\tabcolsep}{6pt}
\renewcommand{\arraystretch}{1.35}
	\caption{Comparison of the rearrangement algorithm (RA), the linear program (LP), and the convolution bound (CB).  
	  In (a)-(e), the distribution $\mu_i$ is uniform on a set $\{z_1^i,\dots,z^i_m\}$ for each $i\in [n]$. In (f), the marginal distributions are continuous, and RA produces an interval $[\underline{s}_N, \bar{s}_N]$ using a discretization with $N=10^5$ steps. 
	 } 
\begin{subtable}{
\centering
	\footnotesize 
	\begin{tabular}{c | c c | c c | c c} 
		& $z^i_j= i\times j$ & $m=120$ & $z^i_j= i\times j$ & $m=160$  & $z^i_j= i\times j$&   $m = 180$   \\
		&    value & time &  value & time & value & time \\ 
		\hline
		LP &  363 & 11988s &  483 & {24960s} & NA & $>$24h \\
		RA &  358 & 66s  &   478 & 140s & 534 & 212s\\
		CB &  363 & 1447s & 483 &  1687s & 543 & 2315s\\ 
		\hline
	\end{tabular}} \\~\\
	\small (a) $n=3$ and $z^i_1,\dots,z^i_m$ are equidistant for each $i\in[n]$
\end{subtable}
% 204    1957    18489

\begin{subtable}{\centering
	\footnotesize 
	\begin{tabular}{c | c c | c c }  
		 &  $z^i_j \mbox{ iid } \sim \exp(100\times i)$ & $m = 60$ &   $z_j^i \mbox{ iid } \sim  \mathrm{U}[0, 100\times i^2]$ &$m = 100$   \\ 
		&   value &  time  &  value  &time \\ 
	%	&  &  & $z_j^3  \mbox{ iid } \sim  \mathrm{U}[0, 550]$ & \\
		\hline
		LP %& 854.75 & 4010.3s 
		& 589.6 & 104.7s &  499.3 & 2865.3s  \\
		RA %& [854.75, 854.75] & 37.8s 
		&583.2 & 0.02s & 499.3 & 35.3s  \\
		CB %\eqref{eq:main1pr} 
		%& 882.27 &  31.2s 
		& 608.0 & 1.8s & 499.3 & 692.1s  \\
		\hline
	\end{tabular}} \\~\\
		\small (b) $n=3$ and $z^i_1,\dots,z^i_m$ are randomly generated  for each $i\in[n]$
	\label{tab:R2-2}
\end{subtable}

\begin{subtable}{\centering
	\footnotesize 
	\begin{tabular}{c | c c | c c } 
		%& $X_1 \sim \mathrm{U}\{a_1, \dots, a_{100}\}$, $a_i \mbox{ iid } \sim U[0, 100]$ & time
		& $z_j^i= i^2 \times j$ & $m = 30$ & $z_j^i=  j^2$ &  $m = 30$   \\ 
		&value  &   time & value &time  \\
	%	& $z_j^3= 8\times j$ &  &  $z_j ^3 \mbox{ iid } \sim \mathrm{U}(0, 150)$, &  \\
	%	& $z_j^4= 20\times j$ &  & $z_j^4 \mbox{ iid } \sim \mathrm{U}(0, 160)$ &   \\
		\hline
		LP  & 436 & 5280s & 1260 & 6707s \\
		RA  & 436 & 0.02s & 1230 & 0.01s \\
		CB %\eqref{eq:main1pr} 
		& 446 & 0.9s & 1260.7 & 0.5s \\
		\hline
	\end{tabular}} \\~\\
		\small (c) $n=4$ and $z^i_1,\dots,z^i_m$ are deterministic for each $i\in[n]$
	\label{tab:R2-3}
\end{subtable}

\begin{subtable}{\centering
	\footnotesize 
	\begin{tabular}{c | c c | c c} 	
		 &  $z^i_j \mbox{ iid } \sim \exp(100\times i)$ & $m = 30$ &   $z_j^i \mbox{ iid } \sim  \mathrm{Binomial}(300,0.2\times i)$ &$m = 30$   \\ 
		&value  &   time & value &time  \\ 
		\hline
		LP  
		& 845.3 & 2937s & 1932 &  3091s \\
		RA  
		& 828.8 
		&  0.03s & 1924  & 0.02s \\
		CB 
		& 881.5 & 1.7s & 1935.5 & 0.9s \\
		\hline
	\end{tabular}} \\~\\
		\small (d) $n=4$ and $z^i_1,\dots,z^i_m$ are randomly generated  for each $i\in[n]$

	\label{tab:R2-4}
\end{subtable}

\begin{subtable}{\centering
	\footnotesize 
	\begin{tabular}{c | c c | c c} 
	 &  $z^i_j \mbox{ iid } \sim \exp(100\times i)$ & $m = 15$ &   $z_j^i \mbox{ iid } \sim  \mathrm{U}[0, 100\times i^2]$ &$m = 20$   \\ 
		&   value &  time  &  value  &time \\ 
		\hline
		LP  & 1198.0 & 4988.9s & NA & $>$24h \\
		RA  &  1158.7 & 0.02s & %[2973.8, 2981.1] 
        2715.8 & 0.03s \\
		CB %\eqref{eq:main1pr} 
		& 1214.0 & 1.3s & %2991.7 
          2760.0  & 1.8s  \\
		\hline
	\end{tabular}
	} \\~\\
		\small (e) $n=5$ and $z^i_1,\dots,z^i_m$ are randomly generated  for each $i\in[n]$
	\label{tab:R2-5}
\end{subtable}

\begin{subtable}{\centering
	\footnotesize 
	\begin{tabular}{c | c c } 
		& $  \mu_i = \mathrm{Gamma}(3, 1)$ & $n = 200$     \\
		& value & time  \\
		\hline
		RA & [599.9, 600.0] & 1710s  \\
		CB & 600 & 2021s  \\
		\hline
	\end{tabular}} \\~\\
		\small (f) $n=200$ with continuous distribution; LP cannot handle such large $n$
	\label{tab:R2-6}
\end{subtable} 
\end{table}

\begin{remark}
	\label{rem:R3} 
This binary search on the values of $x$ in the LP \eqref{eq:R2-3}  inside  \eqref{eq:bisec}  can be done efficiently  as $V(x)$ is  monotone (although not necessarily strictly monotone). This implies that the number of binary queries is $O(\log(\text{length}(T)/\epsilon)$, where $\text{length}(T)$ is the length of the interval $T$ and $\epsilon$ is the (additive) error tolerance to which we want to compute $x$. Hence, Problem \eqref{eq:R2-2} has a poly$(m,n)$ complexity if Problem \eqref{eq:R2-3} with any fixed $x\in T$ has; see Chapter 8.7 of \cite{PS98} for details.

	Next we focus on the complexity of Problem \eqref{eq:R2-3}. While it has an exponential in $n$, namely $m^n$, number of decision variables, its dual problem, given by
	\begin{equation}
	 	\label{eq:R2-3-dual} 
	 	\begin{aligned}
	 	% {\textbf{(LP-dual)}}
   \quad\quad\quad\quad  & \max_{p_{ij} \in \R, i = 1, \cdots, n, j = 1, \cdots, m} \sum_{i = 1}^n \sum_{j = 1}^m p_{ij} \frac{1}{m}, \\
	 	&\quad \mbox{subject to}   ~~\mbox{$y_{\mathbf k} - \sum_{i=1}^n p_{i, k_i} \geq 0$ for any $\mathbf k=(k_1,\dots,k_n)\in [m]^n$},%,\\ 
	 	%& \quad\quad\quad\quad \sum_{\mathbf k:k_i=j}\Pi_{\mathbf k}=\frac 1m~~~ \mbox{for each $i\in [n]$ and $j\in [m]$},
	 	\end{aligned}
	\end{equation}
	is an LP with $n \times m$ variables and $m^n$ constraints. The polynomial number of variables in the dual problem \eqref{eq:R2-3-dual} suggests the potential of yielding poly$(m,n)$-time algorithms for the primal problem \eqref{eq:R2-3}. In particular, via the ellipsoid method (e.g., Lemma 3.2 of \cite{GLS81}), %(Chapter 8.7 of \cite{PS98}),
	the LP \eqref{eq:R2-3} is polynomial-time solvable if there is a separation oracle for \eqref{eq:R2-3-dual} that runs in polynomial time; see Definition 6.2.2 of \cite{GLS12}.  
 However, deducing the availability of such a polynomial-time separation oracle appears challenging. In fact, \eqref{eq:R2-3} belongs to the multi-marginal optimal transport (MOT) problem (\cite{AB21,AB23}) with a so-called set-optimization structure (Section 6.1 of \cite{AB23}). More precisely, given the fixed matrix $\mathbf{z}$ defined by \eqref{eq:matrix} %with $\mathbf{z}_{i,j}=z^j_i$ 
 and $x$, we define the set
    $S = \{ \mathbf{k} \in [m]^n: z_{k_1}^1 + \cdots + z_{k_n}^n > x \}
    $. Based on Definition 6.5 and Theorem 6.8 of \cite{AB23}, Problem \eqref{eq:R2-3} for fixed $x$ has a polynomial complexity if the problem  $\min_{\mathbf{k} \in S} -\sum_{i=1}^n p_{i, k_i}$, for any arbitrary matrix $\mathbf{p} \in \R^{m \times n}$, has. To this end, neither \cite{AB23} or any other works to our best knowledge has worked out the polynomial complexity of the problem $\min_{\mathbf{k} \in S} -\sum_{i=1}^n p_{i, k_i}$ with our considered $S$. With this, it appears that the question of whether \eqref{eq:R2-2} has a polynomial complexity remains open.
    %Developing an efficient separation oracle for the cost function that we consider is of significant interest and a non-trivial undertaking.%, e.g., Chapter 8.7 of \cite{PS98}%
 %because of $n \times m$ variables. 
	% gives a polynomial-time solution to a feasibility version of (LP), i.e., to determine the feasibility of the linear strict inequality constraints of (LP), but this algorithm is polynomial in parameters $  m n$ and  $m^n$ in our context. %This feasibility problem has an equivalent complexity to (LP) itself and the complexity is polynomial-time in the parameters $\hat{m}$ and $\hat{n}$, which are  $\hat{m} = m n$ and $\hat{n} = m^n$ in our context.
%	As a result, the complexity of the LP method for fixed $x$ is poly$(m, n)$. %equivalent to that of the initialization stage \eqref{eq:R2-3}, which is $O(f(mn, m^n))$ with the ellipsoid method, where $f$ is a polynomial.
		%}
\end{remark}

\section{Two further applications}\label{app:applications}
%In this appendix, w
We illustrate the convolution bounds in two additional applications. Section \ref{sec:calibration} constructs a new robust test for simulation calibration. Section \ref{sec:scheduling}  discusses the classic assembly line crew scheduling problem.

%{\color{red}
%\subsection{Additional examples} 
%\begin{example}[Wireless communication]
%	Quantile aggregation techniques have also been applied recently in computer networks and wireless communication; see \cite{BJ20}. In mobile communications, a fading channel is featured as not fully reliable, where some error may occur in transmission. Theoretically, the concept  $\epsilon$-outage capacity (the largest transmission rate with an error probability less than $\epsilon$) is used to describe the performance of the transmission. To improve the capacity, techniques of diversity  are adopted to allow the codeword to travel along multiple fading paths to the receiver so that the probability that the codeword is correctly received is increased. In this regard, the $\epsilon$-outage capacity is exactly the objective quantile function of the sum variables of multiple (marginal) paths. Moreover, the joint behavior of multiple paths is uncertain in practice. Hence, our investigation on quantile aggregation under dependence uncertainty can be applied to examine the performance of techniques of diversity, as in the setting of e.g., \cite{BJ20}.
%\end{example}
%
%}

\subsection{Simulation calibration}\label{sec:calibration}

	In multiple statistical hypothesis testing, quantile aggregation gives critical values for various methods to combine p-values from different tests among which, most often, no dependence information is available; see e.g., \cite{RBWJ19}, \cite{VW20} and \cite{VWW22}. This problem also arises in operations research,  especially
	in the context of stochastic simulation model calibration (\cite{K95, S10}).  
	
	In simulation analysis, calibration refers to the search for parameters of simulation models to best match real data. These models are constructed to resemble the hidden dynamics of a system, which are often complex and not amenable to closed-form analysis. Instead, by running Monte Carlo, we can obtain the model outputs for prediction and other downstream decision-making tasks such as sensitivity analysis and optimization (see, e.g., \cite{LK00} for a range of applications in production and operations management). However, to ensure that the conclusions of these analyses are reliable, it is critical that the input parameters in the hidden dynamics are correctly tuned. This calls for the need for calibration, where the outputs from the simulation model are matched against the real data in order to locate the parameter values. 
	In a setting of one-dimensional (continuous) output, one could rely on a two-sample goodness-of-fit test such as the Kolmogorov-Smirnov (KS) test, which looks at the KS statistic 	
	$$
	\text{KS}=\sup_{x\in\mathbb R}|\hat F_{\text{sim}}(x)-\hat F_{\text{real}}(x)|.
	$$
	Here $\hat F_{\text{sim}} $ and $\hat F_{\text{real}} $ are the empirical distributions of the simulated and real output data, respectively.
	Under the null hypothesis that the parameters are correctly  calibrated, the asymptotic distribution of $\text{KS}$ is equal to the supremum difference between two independent scaled Brownian bridges. %When the output dimension is $n$, we would look at the KS-statistic $\text{KS}_k$ for each dimension $k=1,\ldots,n$. It is then natural to construct an aggregated statistic over all dimensions: the sum of individual KS-statistics, $\sum_{k=1}^n\text{KS}_k$. If it is larger than a critical value, then we reject the hypothesis that the model is the same as reality. 
	%The key issue in this problem is to compute the robust critical values of different significant levels. Denote by $\mu$ the theoretical marginal distribution of the KS statistic, and let $\mu_1 = \cdots = \mu_n = \mu$. The robust critical value is exactly the quantity \eqref{eq:intrinsic_prob} associated with the significant level $t$.   
  In the multi-dimensional case, one can look at multiple KS-statistics, one for each dimension, and further use a Bonferroni correction to adjust the critical value. More precisely, when the dimension is $K$, we would look at the KS-statistic $\text{KS}_k$ for each dimension of output $k=1,\ldots,K$. If any of the $\text{KS}_k$ is above the adjusted critical value $q_{1-\gamma/K}$, then we conclude that the simulation model is different from the real data. Put another way, if $\max_{k=1,\ldots,K}\text{KS}_k>q_{1-\gamma/K}$, then we reject the hypothesis that the model is the same as reality.

It is known that the Bonferroni correction is conservative, especially when different dimensions of the outputs are highly dependent. The question is whether one can improve it without losing validity. This resembles the problem of so-called $p$-value aggregation (\cite{RBWJ19,VW20}), which aims to construct tight family-wise $p$-values from merging multiple $p$-values in individual   experiments. In the considered case, a natural alternative way to construct an aggregated statistic over all dimensions is the sum of individual KS-statistics, $\sum_{k=1}^K\text{KS}_k$. Note that each $\text{KS}_k$ has the same marginal distribution (with quantile $q_{1-\gamma}$), thus we can use Proposition \ref{prop:reduced_bound} to derive a new critical value
%\footnote{The critical value is computed by $\sup_{\nu \in \Lambda_K(\mu)}q_{1-\gamma}^+(\nu)$, where $\mu$ is the distribution of the KS statistic with sample size $M$.}
given by
$$ \inf_{\alpha \in (0, \gamma/K)} \frac{K}{\gamma-K\alpha}\int_{1-\gamma+(K-1)\alpha}^{1-\alpha} q_u \d u.$$
The Kolmogorov probability density function is decreasing at its tail part.
For $\gamma$ sufficiently small (e.g., $\gamma < 0.3$), this critical value is sharp among all possible dependence structures of the $\text{KS}_k$, since the marginal densities are monotonically decreasing beyond $(1-\gamma)$-quantile.
Particularly, for $K=5$, Table \ref{table:KS5} lists this new critical value\footnote{It is known that if $\hat F_{\text{real}}$ is continuous, then under the null hypothesis, the distribution of each $\sqrt{M}\text{KS}_k$ converges to the Kolmogorov distribution as the sample size $M$ goes to infinity. But the convergence rate is slow. We use the method in \cite{V18} for the asymptotic approximation (replacing $x$ by $x+\frac{1}{6\sqrt{M}}+\frac{x-1}{4 M}$ in the Kolmogorov cumulative distribution function $F(x)$).} for different $\gamma$. 

%(ADD TABLE for $\gamma = 0.01, 0.02, 0.05, 0.1, 0.2$ etc., AND FOR DIFFERENT SAMPLE SIZES FOR THE TWO SAMPLES)

\begin{table}[t]
	\centering
	\caption{Critical values for $\sum_{k=1}^5\text{KS}_k$ in the two-sample test (all sample sizes are $M$).}
	\label{table:KS5}
	\begin{tabular}{c|c|c|c|c|c}
		\diagbox{$M$}{$\gamma$}  & 0.2 & 0.1 & 0.05 & 0.02 & 0.01 \\
		\hline
		100 & 0.8875 & 0.9801 & \textbf{1.0645} & 1.1667 & 1.2384 \\
		\hline
		1000 & 0.2833 & 0.3127 & 0.3394 & 0.3718 & 0.3945\\
	\end{tabular}
\end{table}
% A large $\mu_k$ refers to fast service and a large $\lambda_k$ refers to  arrivals.
We illustrate our sum-of-KS statistic and newly derived critical values, and compare them with using the Bonferroni correction, in a multi-class queueing model (M/M/1/$\infty$). In this model, there are 5 types of customers ($k=1,\cdots,5$), each with its own exponential service rate ($\mu_k$) and Poisson arrival rate ($\lambda_k$). The system is first-come-first-served and starts from empty. Suppose we do not know the arrival and service time parameters in the model. On the other hand, suppose we have real output data on the average waiting times for each class of customers among 1000 total arrivals. Such a setting where only output- but not input-level data are observed can arise due to various administrative or operational constraints; see, e.g., \cite{MZ98,FK10,W81,GLQZ19}. Then, to validate a given set of parameter values in Table \ref{table:queue}, we can generate simulation outputs from several conjectured configurations (four in our example) and run the aggregated KS tests described above, which treats the average waiting time of each customer class as one output dimension.

%(DESCRIBE RESULTS; FIRST LIST SEVERAL PARAMETER CONFIGURATIONS, IN WHICH ONE OF THEM IS TRUE. THEN FOR EACH CONFIGURATION, GENERATE A NUMBER OF SIMULATION RUNS, AND RUN THE "SUM OF KS STATISTIC" AND "SUP OF KS STATISTIC" TESTS. HOPEFULLY, BOTH TESTS ACCEPT THE TRUE CONFIGURATION, BUT OURS REJECT MORE CONFIGURATIONS THAN THE BONFERRONI COUNTERPART AMONG THE CONFIGURATIONS THAT ARE NOT TRUE, THUS SHOWING THAT OUR APPROACH HAS A SMALLER TYPE II ERROR) It is clear that the total number of customers waiting in the queue tends to infinity over time if $\rho \geq 1$ while the queue is stable if $\rho < 1$

More precisely, we consider four parameter configurations that are listed in Table \ref{table:queue}. The first column shows the true configuration and the rest are incorrectly conjectured. To facilitate the presentation and to test our approach on several ground-truth models, we define a model parameter
$
\rho = \sum_{k=1}^{5} \frac{\lambda_k}{\mu_k}
$,
which can be viewed as a summary of the traffic intensity. We experiment on three ground-truth settings: $\rho = 1.1, 1, 0.9$, representing scenarios with respectively long, medium and short waiting times. %(CAN USE MORE DIFFERENT NUMBERS? EG $\rho=0.6,0.9,1.5$ OR EVEN MORE EXTREME ONES; NOTE THAT WE CAN USE SHORT TIME HORIZON AS COMMENTED ABOVE SO HOPE THAT THIS REDUCES OUR TIME ON TRIAL-AND-ERROR).
For each model parameter $\rho$, we conduct 1000 experimental repetitions, where in each repetition we independently generate a synthetic data set of size $M=100$, run simulation with the same size on each of the four configurations depicted in Table \ref{table:queue}, and then use our sum statistic and Bonferroni correction on KS to do multiple hypothesis tests. The results are summarized in Table \ref{table:KSreal}.
\begin{table}[t]
	\centering
	\caption{Parameters $(\mu_k, \lambda_k)$ of classes $k = 1, \cdots, 5$ in each configuration}
	\label{table:queue}
	\begin{tabular}{c|c|c|c|c}
		\diagbox{Class}{Configuration}  & Both True & False $\mu$ & False $\lambda$ & Both False  \\
		\hline
		1. Slow-service-small-arrival & (6, $\rho$) & (6.05, $\rho$) & (6, 0.95$\rho$) & (6.05, 0.95$\rho$) \\
		\hline
		2. Slow-service-large-arrival & (6, 2$\rho$) & (6.05, 2$\rho$) & (6, 1.95$\rho$) & (6.05, 1.95$\rho$) \\
		\hline
		3. Medium-service-medium-arrival & (8, 1.6$\rho$) & (8.05, 1.6$\rho$) & (8, 1.45$\rho$) & (8.05, 1.45$\rho$) \\
		\hline
		4. Quick-service-small-arrival & (10, $\rho$) & (10.05, $\rho$) & (10, 0.95$\rho$) & (10.05, 0.95$\rho$) \\
		\hline
		5. Quick-service-large-arrival & (10, 2$\rho$) & (10.05, 2$\rho$) & (10, 1.95$\rho$) & (10.05, 1.95$\rho$)
	\end{tabular}
\end{table}

\begin{table}[t]
	\centering
	\caption{Testing outputs of two methods. Here, the number of different classes is $K =5$, the significant level is $\gamma = 0.05$, the sample size on the synthetic data and simulation runs for each configuration is $M=100$ and the number of total arrivals is $1000$. By computation, the critical value for the sum statistic is $1.0645$ and that for Bonferroni correction is $0.2302$. We use $1000$ experimental repetitions. Type-I Error means the percentage of cases that the configuration with both true parameters in Table \ref{table:queue} is mistakenly rejected, a number set by us. Power records the percentage of cases that the wrong configuration is successfully rejected in our experiment. %(WHY BONFERRONI HAS 0 SIG LEVEL IN ONE OF THE CASES)
	}
	\label{table:KSreal}
	\begin{tabular}{c | c c | c c | c c}
		\multirow{2}{*}{\diagbox{Method}{Model}}  &
		\multicolumn{2}{c|}{Long waiting $(\rho = 1.1)$} &
		\multicolumn{2}{c|}{Medium waiting $(\rho = 1)$} &
		\multicolumn{2}{c}{Short waiting $(\rho = 0.9)$}  \\
		\cline{2-7}
		& Type-I Error & Power & Type-I Error & Power & Type-I Error & Power  \\
		\hline
		Sum statistic & 0.0140 & \textbf{0.5907} & 0.0100 & \textbf{0.6460} & 0.0120 & \textbf{0.6353}\\
		\hline
		Bonferroni correction & 0.0140 & 0.5887 & 0.0080 & 0.6327 & 0.0100 & 0.6247
	\end{tabular}
\end{table}

Our sum-of-KS statistic is shown to be useful in the model where data across dimensions are highly and complicatedly dependent. We find in Table \ref{table:KSreal} that compared to Bonferroni correction, the new sum statistic has consistently slightly greater statistical power under various close-to-critical traffic intensities in this example.
%a slightly greater power in various waiting models.
The basic reason is that the waiting times of different classes are highly dependent, and hence the sum statistic takes advantage from the bound \eqref{eq:main1hom} on quantile aggregation with dependence uncertainty. {In case of near independence (e.g.,  small $\rho$), the Bonferroni correction is known to have a very good power, and it outperforms the sum-of-KS method.} {The drawback of this sum-of-KS statistic may be an overemphasis on the worst-case scenario. There are other methods of multivariate goodness-of-fit tests adapting well to the environment of dependence, such as Peacock's test and its later variants.\footnote{We thank an anonymous referee for pointing out Peacock's test in this application.}}

\subsection{Assembly line crew scheduling}
\label{sec:scheduling}
	
{The quantile aggregation problem is closely related to the problem of assembly line crew scheduling, which we explain in this section.

 A manufacturing facility produces items which require the completion of $n$ tasks in series.  Suppose that there are $m$ assembly lines (rows) and $n$ operations (columns). Each operation has $m$ crews to be assigned to each line.  
  If the $i$-th operation  is put in the $j$-th  assembly line, it costs $z_j^i$ units of time (or another type of resource) to complete the task.
  Therefore, we can use the  $m \times n$ matrix  in \eqref{eq:matrix}, where the number $z_j^i$ at $(i,j)$-position represents the processing time of the $i$-th crew in the $j$-th operation. 	  
	The objective is to appropriately assign crews in each operation to the lines in order to minimize the makespan, that is, the maximum total processing time of all assembly lines.  For $j = 1, \cdots, n$, denote by $\mu_j$ the distribution measure for a discrete uniform distribution on the $j$-th column ($m$ elements).  
	The objective is to find an optimal arrangement of elements in each column to minimize the maximum row sum (the makespan). We denote  the minimal makespan  by $s_{\min}$.  
	 A similar problem appears in many other fields, e.g., in healthcare operations where usage of operating rooms among all types of elective surgeries is to be optimized.   	 
	This problem is essentially the same as the matrix rotation problem explained in Appendix \ref{sec:R2-1}, but   
	 minimizing the maximum row sum instead of maximizing the minimum row sum.
 These two problems can be converted into each other by simply putting a negative sign in front of all values $z_j^i$. 
	
	A few remarks on  the problem of assembly line crew scheduling  and that of quantile aggregation   are needed. 
	Denote by $q_{\min}=\inf_{\nu \in \Lambda(\boldsymbol \mu) } q_1^{-} (\nu)$, the infimum value of the quantile aggregation problem with the same marginals. 
	First, the problem of assembly line crew scheduling 
  is know to be NP-complete; see \cite{H84} and \cite{H15}. 
  Second, 
  since each arrangement induces a dependence structure among random variables with marginal distributions $\mu_1,\dots,\mu_n$, 
  the two problems are closely connected with one difference:
   the assembly line crew scheduling problem only allows for discrete dependence structures taking $m$ different values in $\R^n$ (indeed, it can be seen as a quantile aggregation problem restricted on a discrete probability space of $m$ states), whereas the quantile aggregation also allows for other dependence structures, such as independence. 
   Third,  $q_{\min}$
  is a lower bound  for $s_{\min}$, and RA can produce an upper bound for $s_{\min}$ (RA can be used to compute both the maximum of minimum row sum or minimum of maximum row sum).
    CB in \eqref{eq:main2} in Proposition \ref{prop:qa-1}  produces a lower bound  $B_{\rm conv}$ on $q_{\min}$
 and hence also on $s_{\min}$.  
  Fourth, the difference between $q_{\min}$ and $s_{\min}$ is relatively small.  
  These two values often coincide: as seen from our numerical examples, often RA  coincides with CB, making the inequalities  in 
  $$
  \mbox{RA} \ge s_{\min} \ge q_{\min} \ge B_{\rm conv}
  $$
   all exact. Fifth, one can also use the LP method in Appendix \ref{sec:R2-1}  to compute $q_{\min}$ when the values of $m$ and $n$ are small (see the numerical experiments in Appendix \ref{sec:R2-1}). This yields a lower bound for $s_{\min}$.

Suppose that a matrix of representation is given by the left-hand side of \eqref{eq:assembly_matrix}, with a makespan of $87+60+83=230$. Let $X_i$ be a uniform discrete random variable valued on the $i$-th column of the matrix and $\mu_i$ be the corresponding distribution. For example ($i=1$), $X_1$ takes each value of the first column $\{44,66,67,71,87\}$ with probability $1/5$. For the discrete distributions $\mu_1,\mu_2,\mu_3$, the minimal makespan is at least $\inf_{\nu \in \Lambda(\boldsymbol \mu) } q_1^{-} (\nu)$. According to Proposition \ref{prop:qa-1}, $\sup_{\boldsymbol \beta \in \Delta_n} \sum_{i=1}^n R_{1-\beta_i-\beta_0,\beta_0 }(\mu_i)$ serves as a lower bound for the minimal makespan and the maximizer $\boldsymbol \beta$ provides a hint to the optimal scheduling rule. In this example, the explicit bound is attainable: $\sup_{\boldsymbol \beta \in \Delta_3} \sum_{i=1}^3 R_{1-\beta_i-\beta_0,\beta_0 }(\mu_i)=160$ with the maximizer $\boldsymbol \beta = (0, 0.2,0.6,0.2)$. 
If an arrangement yields a minimal makespan of 160, it must optimal. Indeed, one optimal arrangement is given by the right-hand side of \eqref{eq:assembly_matrix}.

\begin{equation}\label{eq:assembly_matrix}%\small
	\mbox{Convolution bound: }
	\begin{pmatrix}
		44 & 10 & 24\\
		66 & 32 & 37\\
		67 & 48 & 41\\
		71 & 57 & 43\\
		87 & 60 & 83
	\end{pmatrix}
	\quad\quad \Longrightarrow \quad\quad
	\begin{pmatrix}
		87 & 10 & 43\\
		71 & 60 & 24\\
		67 & 48 & 41\\
		44 & 32 & 83\\
		66 & 57 & 37
	\end{pmatrix}
\end{equation}
%\subsection{Comparison between different methods}

There are several algorithms in the literature for the problem of assembly line crew scheduling. \cite{CS76} and \cite{H84} naturally adopted greedy-type (largest first) methods. \cite{CY84} improved and developed an algorithm to approximate the problem. \cite{EPR13} proposed the RA in the context of risk management. These numerical algorithms provide an upper bound for the minimal makespan because they always return to a plausible scheduling rule. 

The $k$-partitioning problem is  a similar problem to the one in this section, as it can be solved by finding the minimal maximum row sum of a matrix; see \cite{BJV18}. From a different perspective, our convolution bound provides analytical assistance for this type of problem. As it is a lower bound for the minimal makespan, if it is equal to the RA results, we can guarantee that the scheduling rule is optimal. 
}
	
\end{document}